{}

\documentclass[11pt,a4paper,reqno]{amsart}

\usepackage[a4paper,text={160mm,247mm},centering]{geometry}

\usepackage{setspace}
\usepackage[utf8]{inputenc}
\usepackage{mathrsfs}

\usepackage[english]{babel}
\usepackage{hyphenat}
\usepackage[nosort]{cite}

\newcommand{\rcite}[1]{ref.~\cite{#1}}
\newcommand{\rcites}[1]{refs.~\cite{#1}}

\usepackage[usenames,dvipsnames,table]{xcolor}
\definecolor{mygreen}{rgb}{0,0.4,0}
\definecolor{myblue}{rgb}{0,0.0,0.4}
\definecolor{refrcolor}{rgb}{0,0.4,0}
\definecolor{cgreen}{rgb}{0,0.7,0}
\definecolor{ecolor}{rgb}{.52,.03,.06}
\definecolor{bgcolor}{rgb}{.96,.95,.80}
\definecolor{bgcolordark}{rgb}{.80,.80,.67}
\definecolor{faint}{rgb}{.80,.80,.80}

\usepackage{tcolorbox}
\tcbuselibrary{breakable, skins}
\usetikzlibrary{patterns}
\newtcolorbox{myproofbox}[1][]{
  enhanced,
  breakable,
  borderline west={2pt}{0pt}{faint},
  notitle,
  before skip=10pt,
  after skip=10pt,
  colback=white, 
  colframe=white,
  frame hidden,
  boxrule=0pt, 
  boxsep=0pt,
  sharp corners,
  left=8pt, right=0pt, top=1pt, bottom=0pt,
  fontupper=\small,
  #1
}

\makeatletter
\renewenvironment{proof}[1][\proofname]{%
  \begin{myproofbox}\par\pushQED{\qed}\normalfont%
  \topsep6\p@\@plus6\p@\relax%
  \trivlist%
  \item[\hskip\labelsep\itshape#1\@addpunct{.}]%
  }%
  {\popQED\endtrivlist\end{myproofbox}\@endpefalse%
  \if@noskipsec\leavevmode\fi\noindent\ignorespacesafterend}
\makeatother

\makeatletter
\def\mr@ignsp#1 {\ifx\:#1\@empty\else #1\expandafter\mr@ignsp\fi}%
\newcommand{\multiref}[1]{\begingroup
\xdef\mr@no@sparg{\expandafter\mr@ignsp#1 \: }%
\def\mr@comma{}%
\@for\mr@refs:=\mr@no@sparg\do{\mr@comma\def\mr@comma{,}\ref{\mr@refs}}%
\endgroup}
\renewcommand{\eqref}[1]{(\multiref{#1})}
\makeatother

\makeatletter
\newcommand{\namedref}[2]{\hyperref[#2]{#1~\ref*{#2}}}%
\newcommand{\namedreff}[2]{\hyperref[#2]{#1\,\ref*{#2}}}%
\newcommand{\secref}{\namedreff{\S}}

\newcommand{\appref}{\@ifstar{\namedref{Appendix}}{\namedref{Appendix}}}
\newcommand{\tabref}{\@ifstar{\namedref{Table}}{\namedref{Table}}}
\newcommand{\figref}{\@ifstar{\namedref{Figure}}{\namedref{Figure}}}
\newcommand{\propref}{\namedref{Prop.}}
\newcommand{\thmref}{\namedref{Thm.}}
\newcommand{\rmkref}{\namedref{Rmk.}}
\newcommand{\defref}{\namedref{Def.}}
\makeatother

\providecommand{\href}[2]{#2}

\usepackage{amsmath, amssymb, amsfonts, amsthm, mathtools}

\theoremstyle{plain}
\newtheorem{theorem}{Theorem}[section]
\newtheorem*{theorem*}{Theorem}
\newtheorem{proposition}[theorem]{Proposition}
\newtheorem{rmk}[theorem]{Remark}

\newtheorem{defn}[theorem]{Definition}

\newenvironment{sketch}{%
  \renewcommand{\proofname}{Sketch of the proof}\proof}{\endproof}

\RequirePackage{mathfixs}
\ProvideMathFix{autobold}
\ProvideMathFix{multskip}
\ProvideMathFix{frac,root,rootclose}

\usepackage[margin=10pt,font=small,labelfont=bf, labelsep=colon]{caption} [2022/02/20]

\usepackage{fancyhdr, dsfont}

\usepackage{graphbox}

\usepackage{tikz}

\usepackage[usenames,dvipsnames,table]{xcolor}
\definecolor{mygreen}{rgb}{0,0.4,0}
\definecolor{myblue}{rgb}{0,0.0,0.4}
\definecolor{refrcolor}{rgb}{0,0.4,0}
\definecolor{cgreen}{rgb}{0,0.7,0}
\definecolor{ecolor}{rgb}{.52,.03,.06}
\definecolor{bgcolor}{rgb}{.96,.95,.80}
\definecolor{bgcolordark}{rgb}{.80,.80,.67}
\definecolor{faint}{rgb}{.80,.80,.80}

\usepackage{enumitem}

\usepackage{tabularx}

\let\Re\undefined\DeclareMathOperator{\Re}{Re}
\let\Im\undefined\DeclareMathOperator{\Im}{Im}

\DeclareMathOperator{\End}{End}

\DeclareMathOperator{\Mat}{Mat}

\DeclareMathOperator{\Sel}{S}
\DeclareMathOperator{\SelE}{S^E}

\ProvideMathFix{rfrac,vfrac}

\newcommand{\der}{\mathrm{d}}
\newcommand{\diff}[2][.]{\mathinner{\der#2\if #1.\else^{#1}\fi}}

\newcommand{\grp}[1]{\mathrm{#1}}

\newcommand{\defeq}{\mathrel{:=}}

\newcommand{\mat}[1]{\begin{pmatrix}#1\end{pmatrix}}

\newcommand{\smax}{s_{\mathrm{max}}}

\makeatletter
\providecommand*{\shuffle}{%
  \mathbin{\mathpalette\shuffle@{}}%
}
\newcommand*{\shuffle@}[2]{%
  \sbox0{$#1\vcenter{}$}%
  \kern .15\ht0 
  \rlap{\vrule height .25\ht0 depth 0pt width 2.5\ht0}%
  \raise.1\ht0\hbox to 2.5\ht0{%
    \vrule height 1.75\ht0 depth -.1\ht0 width .17\ht0 %
    \hfill
    \vrule height 1.75\ht0 depth -.1\ht0 width .17\ht0 %
    \hfill
    \vrule height 1.75\ht0 depth -.1\ht0 width .17\ht0 %
  }%
  \kern .15\ht0 
}
\makeatother

\usepackage{xparse}
\newcommand{\SI}[1]{\Sel[#1]}
\ExplSyntaxOn
\NewDocumentCommand{\SIE}{m m}
{
\SelE\!\Big[\begin{smallmatrix}
 \SI_print:n {#1} \\
 \SI_print:n {#2}
 \end{smallmatrix}\Big]
}
\seq_new:N \l_SI_list_seq
\cs_new_protected:Npn \SI_print:n #1
{
  \seq_set_split:Nnn \l_SI_list_seq { , } { #1 }
  \seq_use:Nn \l_SI_list_seq { , & }
}
\ExplSyntaxOff

\usepackage[pdfencoding=auto,bookmarks=true,hyperfigures=true]{hyperref}%
\PassOptionsToPackage{unicode}{hyperref}
\providecommand{\hypersetup}[1]{}
\providecommand{\texorpdfstring}[2]{#1}
\hypersetup{plainpages=false}
\hypersetup{pdfpagemode=UseNone}
\hypersetup{bookmarksnumbered=true}
\hypersetup{pdfstartview=FitH}
\hypersetup{colorlinks=false}
\hypersetup{citebordercolor={0.7 0.7 1}}
\hypersetup{urlbordercolor={.4 .8 1}}
\hypersetup{linkbordercolor={1 .8 .6}}
\hypersetup{colorlinks=true, urlcolor=[rgb]{0.13,0.30,0.45}, linkcolor=[rgb]{0.13,0.30,0.45}, citecolor=[rgb]{0.55,0.0,0.05}}

\makeatletter
\let\@keywords\@empty
\let\@subject\@empty
\providecommand{\keywords}[1]{\gdef\@keywords{#1}}
\providecommand{\subject}[1]{\gdef\@subject{#1}}
\def\thetitle{\@title}
\def\theauthor{\@author}
\def\thesubject{\@subject}
\def\thedate{\@date}
\def\thekeywords{\@keywords}
\makeatother
\AtBeginDocument{
\hypersetup{pdftitle={\thetitle}}%
\hypersetup{pdfauthor={\theauthor}}%
\hypersetup{pdfsubject={\thesubject}}%
\hypersetup{pdfkeywords={\thekeywords}}%
}

\numberwithin{equation}{section}

\newcommand{\eqn}[1]{eq.~\eqref{#1}}

\allowdisplaybreaks[1]

\newif\ifnote 
\notetrue

\let\qed\relax\newcommand{\qed}
{\hfill\ensuremath{\Box}}

\newcommand{\pd}{\partial}

\newcommand{\ap}{\alpha'}

\newcommand{\zC}{\mathbb C}

\newcommand{\zR}{\mathbb R}

\newcommand{\cH}{\mathcal{H}}

\newcommand{\cL}{\mathcal{L}}

\newcommand{\cO}{\mathcal{O}}

\usepackage{nicefrac}

\newcommand{\PC}{\mathbb{P}^1_{\mathbb{C}}}
\newcommand{\zb}{\overline{z}}
\newcommand{\wb}{\overline{w}}
\newcommand{\yb}{\overline{y}}

\newcommand{\ub}{\overline{u}}

\newcommand{\tb}{\overline{t}}

\newcommand{\SC}{\mathscr{S}}
\newcommand{\cF}{\mathscr{F}}
\newcommand{\hcF}{\hat{\mathscr{F}}}

\newcommand{\Fhyper}[4]{{_2F_1}\!\!\left[\!\!\begin{array}{c}#1,\ #2\\#3\end{array}\!\!\Big|#4\right]}

\newcommand{\mC}{\mathbb{C}}

\newcommand{\hF}{\hat{F}}

\newcommand{\hf}{\hat{f}}

\newcommand{\llangle}{\langle\!\langle}
\newcommand{\rrangle}{\rangle\!\rangle}

\title[Closed-string amplitude recursions from the Deligne associator]{\textbf{
    Closed-string amplitude recursions \texorpdfstring{\\}{} from the Deligne associator
    }}
\author[Konstantin Baune,
        Johannes Broedel,
        Federico Zerbini]{Konstantin Baune\textsuperscript{*},
        Johannes Broedel\textsuperscript{*},
        Federico Zerbini\textsuperscript{\textdagger}}
\address{\textsuperscript{*} Institute for Theoretical Physics, ETH Zurich, Wolfgang-Pauli-Str.~27, 8093 Zurich, Switzerland}
\address{\textsuperscript{\textdagger} Departamento de Matem\'aticas Fundamentales, UNED, Calle de Juan del Rosal~10, 28040 Madrid, Spain}
\email{baunek@ethz.ch, jbroedel@ethz.ch, f.zerbini@mat.uned.es}
\date{\today}

\begin{document}
%
\begin{abstract}
	Inspired by earlier results on recursions for open-string tree-level amplitudes, and by a result of Brown and Dupont relating open- and closed-string tree-level amplitudes via single-valued periods, we identify a recursive relation for closed-string tree-level amplitudes. We achieve this by showing that closed-string analogues of Selberg integrals satisfy the Knizhnik--Zamolodchikov equation for a suitable matrix representation of the free Lie algebra on two generators, and by identifying the limits at $z\,{=}\,1$ and $z\,{=}\,0$, which are related by the Deligne associator, with $N$-point and $(N{-}1)$-point closed-string amplitudes, respectively.
\end{abstract}

\maketitle

\setcounter{tocdepth}{1}
\setcounter{secnumdepth}{3}
\tableofcontents


\section{Introduction}
\label{sec:introduction}

Developing techniques for the calculation of perturbative scattering amplitudes in string theory has led to various advancements in mathematical physics: while on the one hand providing a laboratory for testing techniques and exploring features of corresponding scattering amplitudes in quantum field theory, string scattering amplitudes on the other hand link well-defined physics questions to problems in number theory and algebraic geometry.

String theory scattering amplitudes are calculated as correlators in a conformal field theory on a surface, the so-called worldsheet. Accordingly, contributions to string scattering are ordered by the genus of the worldsheet. String states come in open and closed versions: open-string states are inserted on the boundary of the worldsheet, while closed-string states are inserted in the bulk. Correspondingly, the leading order of open-string scattering is calculated on a disc, while the leading order of closed-string scattering refers to a sphere. In this article, we are going to deal with open- and closed-string scattering amplitudes on genus-zero oriented worldsheets exclusively.

Leaving out several technicalities, the calculation of string scattering amplitudes boils down to evaluating integrals of Selberg type: these are integrals over configuration spaces of points on a given surface, which encompass all configurations of insertions of string states on the respective worldsheet. We will refer to the different classes of integrals arising in the calculation of genus-zero string scattering as open and closed Selberg integrals.
The physical momenta of the string states, as well as the inverse string tension $\alpha'$, whose limit $\alpha'\to 0$ recasts quantum field theory amplitudes, enter the Selberg integrals through kinematic parameters $s_{ij}$ known as \emph{Mandelstam variables}. 

Various features of string scattering amplitudes become more apparent after expressing them in terms of open and closed Selberg integrals.
\begin{itemize}
	\item \textit{MZVs and the sv-map}: since the Mandelstam variables $s_{ij}$ are proportional to the inverse string tension $\alpha'$, the Taylor expansion of string theory amplitudes at $s_{ij}\,{=}\,0$ can be viewed as a deformation of quantum field theory amplitudes, the coefficients of the series beyond leading order being ``string theory corrections''. These coefficients are of remarkable interest to number theorists. In the case of genus-zero open-string amplitudes, they were shown in \cite{Broedel:2013aza} to belong to the ring $\mathcal Z$ generated over $\mathbb Q$ by the \emph{multiple zeta values} (MZVs) $\zeta(n_1,\ldots ,n_r)$, namely the (conjecturally transcendental) real numbers defined for integers $n_1,\ldots ,n_{r-1}\geq 1$, $n_r\geq 2$ by the nested series
    \begin{equation}\label{eqdefMZVs}
        \zeta(n_1,\ldots ,n_r)=\sum_{0<k_1<\cdots <k_r}\frac{1}{k_1^{n_1}\cdots k_r^{n_r}},
    \end{equation}
    which can be seen as the special values at $z=1$ of the (one-variable) \emph{multiple polylogarithms}
    \begin{equation}\label{eqdefMPLs}
        \mathrm{Li}_{n_1,\ldots ,n_r}(z)=\sum_{0<k_1<\cdots <k_r}\frac{z^{k_r}}{k_1^{n_1}\cdots k_r^{n_r}}.
    \end{equation}
    It was also conjectured by Stieberger in \cite{Stieberger:2013wea} that, in the case of genus-zero closed string amplitudes, the Taylor coefficients should also be MZVs, but in fact of a very special kind, called \emph{single-valued MZVs}. These numbers, first studied by Brown in \cite{Brown:2013gia}, are the special values at $z=1$ of a single-valued version\footnote{Eq.~\eqref{eqdefMPLs} defines a holomorphic function of $z$ in the unit disc, which can be analytically continued to the punctured complex plane $\mathbb C\smallsetminus \{0,1\}$ thanks to a suitable integral representation. This analytic continuation is not unique, as winding around the points $0,1$ changes the value of the multiple polylogarithms, hence they are (holomorphic) \emph{multivalued functions}. Single-valued but non-holomorphic (only real analytic) analogues of these functions can be constructed as special linear combinations of products of holomorphic and antiholomorphic multiple polylogarithms, generalizing the formula $\log|z|^2=\log(z)+\log(\bar z)$.} constructed in \cite{BrownSVMPL} of the multiple polylogarithms \eqref{eqdefMPLs}. There is a canonical map ${\rm sv}$ which sends a multiple zeta value $\zeta(n_1,\ldots ,n_r)$ to its corresponding single-valued analogue $\zeta^{\rm sv}(n_1,\ldots ,n_r)$, which we call here\footnote{This is not standard terminology in mathematics, where sometimes a related but different map takes the same name \cite{Brown:2013gia}.} the \emph{sv-map}, which is (assuming standard conjectures) a ring morphism from~$\mathcal Z$ to the ring~$\mathcal Z^{\rm sv}$ generated over $\mathbb Q$ by the single-valued MZVs. For example, one has ${\rm sv}(\zeta(2n))=\zeta^{\rm sv}(2n)=0$ and ${\rm sv}(\zeta(2n+1))=\zeta^{\rm sv}(2n+1)=2\zeta(2n+1)$.
    The stronger version of Stieberger's conjecture claimed that one can map the small $\alpha'$ expansion of
    genus-zero open-string amplitudes to the corresponding expansion for closed strings by applying, coefficient by coefficient, the sv-map. A sketch of a proof was given in \cite{Schlotterer:2018zce}, while a complete proof, due to Brown--Dupont, appeared in \rcite{Brown:2019wna}. A more elementary proof of the weaker version of the conjecture was given in \cite{VanZerb}.
	\item \textit{Recursions}: for open-string amplitudes at genus zero there exists a recursive algorithm for calculating the $L$-point amplitude from a known expression for the $(L{-}1)$-point amplitude. The algorithm is based on constructing a vector $\hF_L(x)$ parametrized by an unintegrated point $x\in\,]0,1[\,$ out of a class of Selberg integrals, which we refer to as ``open Selberg integrals'', such that $\hF_L(x)$ extends to a multivalued function of~$x\in\mathbb C\smallsetminus\{0,1\}$ and satisfies the Knizhnik--Zamolodchikov (KZ) equation: 
	\begin{equation}\label{eqholKZintro}
		\frac{\partial}{\partial x}\hF_L(x)=\left(\frac{e_0}{x}+\frac{e_1}{x-1}\right)\hF_L(x)\,,
	\end{equation}
	where $e_0,e_1$ denote suitable matrix representations with entries linear in the parameters~$s_{ij}$. The regularized boundary values $\hF_L(0)$ and $\hF_L(1)$ of such a solution can be shown to be related by the \emph{Drinfeld associator} $\Phi(e_0,e_1)$, which is a generating series (in the non-commuting variables $e_0,e_1$) of all MZVs, via
	\begin{equation*}
		\hF_L(1)=\Phi(e_0,e_1)\,\hF_L(0).
	\end{equation*}
    After identifying the open Selberg integrals appearing in $\hF_L(1)$ and $\hF_L(0)$ with $L$-point and $(L{-}1)$-point genus-zero open-string amplitudes, respectively, by taking a further ``physical limit'', the above equation allows to determine the former from the latter, thus leading to a recursive algorithmic procedure to compute $L$-point open-string amplitudes for any $L$.
	This algorithm was first suggested by Terasoma \cite{Terasoma}, who was interested in proving that the coefficients of the Taylor expansions of a class of Selberg integrals are MZVs. Its application to the computation of open string amplitudes was presented in \rcite{Broedel:2013aza} and refined in \rcite{Kaderli}. The recursive process for open-string amplitudes is pictured in the first line of \figref{fig:flowchart} below.
\end{itemize}

Combining knowledge of the open-string recursive algorithm with the result of Brown--Dupont from \cite{Brown:2019wna}, it follows that closed-string amplitudes should also satisfy a recursion in the number~$L$ of string insertions, where the role played by the Drinfeld associator in the open case should be taken by the \emph{Deligne associator} $\Phi^\mathrm{sv}(e_0,e_1)$ \cite{Brown:2013gia},  namely a generating series of all single-valued MZVs, but this result was never explicitly spelled out in the literature. This expectation was already hinted by Stieberger in \rcite{Stieberger:2013wea}.

In this article, we take a different route to demonstrate that closed-string amplitudes indeed satisfy a recursion in the number~$L$ of string insertions involving the Deligne associator, by providing a ``single-valued version'' of the above algorithm: by spelling out the theory of single-valued solutions to the KZ equation, we mimic the open-string amplitude recursion of \rcite{Broedel:2013aza}, working out the mathematical formalism described by the second line in \figref{fig:flowchart}. While the sv-map and the open-string procedure will serve as a guideline, our result is independent from \cite{Brown:2019wna}, as we only work on the ``closed-string side'' and combine elementary tools with known facts in twisted cohomology theory.

More specifically, we construct out of a class of single-valued analogues of open Selberg integrals, which we call ``closed Selberg integrals'', a vector of interpolating real-analytic single-valued functions $\hcF_L(z)$ of an auxiliary point $z\in\mathbb{C}\smallsetminus\{0,1\}$ which satisfy the (holomorphic) KZ equation \eqref{eqholKZintro} and is the image of $\hF_L$ under the sv-map.
Similarly to the holomorphic case, the regularized boundary values $\hcF_L(0),\,\hcF_L(1)$ are related by (see \thmref{prop:equationclosed})
	\begin{equation}\label{eqmainintro}
		\hcF_L(1)=\Phi^\mathrm{sv}(e_0,e_1)\,\hcF_L(0)
	\end{equation}
where the Deligne associator $\Phi^\mathrm{sv}(e_0,e_1)$ can be obtained by applying coefficient-wise the sv-map to the Drinfeld associator $\Phi(e_0,e_1)$. This is the main step towards showing in \thmref{thm:closedrec} that, by taking appropriate physical limits of the Mandelstam variables $s_{ij}$ involved, eq.~\eqref{eqmainintro} relates $L$-point to $(L{-}1)$-point closed string amplitudes. For low multiplicity, we provide explicit formulas for this recursive relation. As a corollary of our construction, we obtain yet another proof of the weaker version of Stieberger's conjecture, namely that the coefficients of genus-zero closed-string amplitudes are single-valued MZVs.

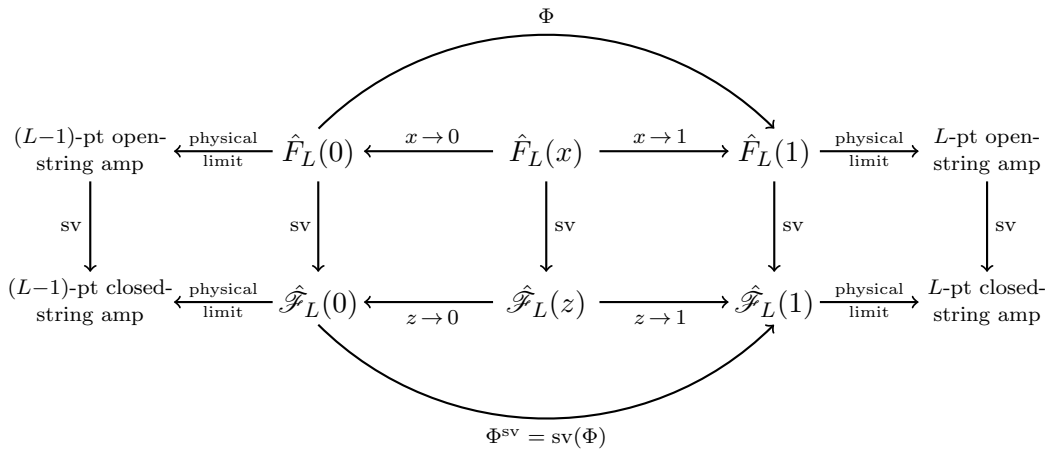
\begin{figure}[ht]
\centering
\begin{tikzpicture}
	\node at (0,1) {$\hF_L(x)$};
	\node at (-3,1) {$\hF_L(0)$};
	\node at (3,1) {$\hF_L(1)$};
	\node at (0,-1) {$\hcF_L(z)$};
	\node at (-3,-1) {$\hcF_L(0)$};
	\node at (3,-1) {$\hcF_L(1)$};
	%
	\draw[thick,->] (-0.7,1) -- (-2.4,1);
	\draw[thick,->] (0.7,1) -- (2.4,1);
	\draw[thick,->] (-0.7,-1) -- (-2.4,-1);
	\draw[thick,->] (0.7,-1) -- (2.4,-1);
	%
    \draw[thick,->] (3.6,-1) -- (4.9,-1);
    \draw[thick,->] (3.6,1) -- (4.9,1);
    \draw[thick,->] (-3.6,-1) -- (-4.9,-1);
    \draw[thick,->] (-3.6,1) -- (-4.9,1);
    %
	\draw[thick,<-] (3,1.3) arc (45:135:4.24);
	\draw[thick,->] (-3,-1.3) arc (225:315:4.24);
    %
	\node at (0,2.8) {\scriptsize$\Phi$};
	\node at (0,-2.8) {\scriptsize$\Phi^{\mathrm{sv}}=\mathrm{sv}(\Phi)$};
    %
	\node at (-1.5,1.2) {\scriptsize$x\,{\to}\,0$};
	\node at (1.5,1.2) {\scriptsize$x\,{\to}\,1$};
	\node at (-1.5,-1.2) {\scriptsize$z\,{\to}\,0$};
	\node at (1.5,-1.2) {\scriptsize$z\,{\to}\,1$};
	%
	\draw[thick,->] (0,0.6) -- (0,-0.6);
	\draw[thick,->] (-3,0.6) -- (-3,-0.6);
	\draw[thick,->] (3,0.6) -- (3,-0.6);
    \draw[thick,->] (5.8,0.6) -- (5.8,-0.6);
    \draw[thick,->] (-6,0.6) -- (-6,-0.6);
	%
    \node at (-3.23,0) {\scriptsize sv};
	\node at (0.25,0) {\scriptsize sv};
	\node at (3.25,0) {\scriptsize sv};
    \node at (6.05,0) {\scriptsize sv};
    \node at (-6.25,0) {\scriptsize sv};
	%
    \node at (4.25,-0.97) {\scriptsize$\substack{\text{physical}\smallskip\\\text{limit}}$};
    \node at (4.25,1.03) {\scriptsize$\substack{\text{physical}\smallskip\\\text{limit}}$};
    \node at (-4.25,-0.97) {\scriptsize$\substack{\text{physical}\smallskip\\\text{limit}}$};
    \node at (-4.25,1.03) {\scriptsize $\substack{\text{physical}\smallskip\\\text{limit}}$};
	%
    \node at (5.8,-1) {\scriptsize $\begin{array}{c}L\text{-pt closed-}\\\text{string amp}\end{array}$};
    \node at (5.8,1) {\scriptsize $\begin{array}{c}L\text{-pt open-}\\\text{string amp}\end{array}$};
    \node at (-6,-1) {\scriptsize $\begin{array}{c}(L{-}1)\text{-pt closed-}\\\text{string amp}\end{array}$};
    \node at (-6,1) {\scriptsize $\begin{array}{c}(L{-}1)\text{-pt open-}\\\text{string amp}\end{array}$};
\end{tikzpicture}
\caption{Relation of the algorithms described in this article.}
\label{fig:flowchart}
\end{figure}

\subsection*{Structure of the article} We will start in \secref{sec:kze} by recalling some known facts about the Knizhnik--Zamolodchikov equation, its holomorphic solutions and its relation with multiple polylogarithms and MZVs. Furthermore, we extend these results to single-valued solutions, culminating in \thmref{thmKZvectorspaceSV}, which will be a key ingredient to demonstrate the closed-string recursion. In \secref{sec:stringamps}, we will review the general structure of string scattering amplitudes at genus zero. \secref{sec:Selberg} is devoted to defining open and closed Selberg integrals, which will be the building blocks of the recursive procedures, and we will also discuss their limiting behavior. The recursions for string amplitudes are then shown in~\secref{sec:openrec} and~\secref{sec:closedrec}: first, the open-string amplitude recursion from \rcites{Broedel:2013aza,BK,Kaderli} is reviewed in~\secref{sec:openrec} with some additional details not given in these references, then a corresponding recursion for closed-string amplitudes is shown in~\secref{sec:closedrec}, including the main result of this article in \thmref{thm:closedrec}. For both of these sections, the general formalism of the recursions is discussed, with the two simplest non-trivial calculations shown explicitly. Two appendices supplements this article showing additional lengthy details on the calculations of certain limit vectors for the recursion.

\subsection*{Open questions}
The current article settles the question of a recursion for closed-string tree-level amplitudes using the Deligne associator connecting two regularized boundary values. Given that for one-loop open-string scattering a recursive algorithm employing a genus-one analogue of the Drinfeld associator already exists \cite{BK,Mafra:2019xms}, the question for a similar construction for genus-one closed-string amplitudes based on a genus-one analogue of the Deligne associator is to be tackled next. 
Since there is no analogue of the sv-map at genus one in the mathematical literature, we hope that the approach of this article, which bypasses the use of the sv-map at genus zero, may provide a successful strategy to solve this open question.  
\medskip

For the genus-zero recursions treated in this article, we have been assuming massless external states. It would be interesting to know in which way the recursive algorithms would survive when the external string states are promoted to massive ones. This might be considered as a deformation, whose relation to the so-called deformed string amplitudes \cite{Cheung:2023adk} and string amplitudes in curved backgrounds \cite{Alday:2023mvu,Alday:2024ksp} should be settled in the future.

\subsection*{Acknowledgments}
We are grateful to Egor Im, Beat Nairz and Stephan Stieberger for discussions on related topics. The work of K.B.~and J.B.~is partially supported by the Swiss National Science Foundation through the NCCR SwissMAP. The research of F.Z.~was partly supported by the Royal Society, under the grant URF\textbackslash R1\textbackslash 201473, as well as by the Spanish Ministry of Science, Innovation and Universities under the 2023 grant ``Proyecto the generación de conocimiento'' PID2023-152822NB-I00. The authors thank the Munich Institute for Astro-, Particle and BioPhysics (MIAPbP), which is funded by the Deutsche Forschungsgemeinschaft (DFG, German Research Foundation) under Germany’s Excellence Strategy - EXC-2094 - 390783311, for the hospitality during the workshop ``Special functions: from geometry to fundamental interactions''.
K.B.~thanks the Max Planck Institut für Physik München for hospitality during the final stages of the project.


\section{The KZ equation}\label{sec:kze}

\begin{defn}
Let $U$ be an open connected subset of $\mathbb{C}\smallsetminus \{0,1\}$, $V$ be a complex vector space, and $e_0,e_1\in \End(V)$. We say that a real-analytic function $F:U\to V$ satisfies the \emph{Knizhnik--Zamolodchikov (KZ) equation} if
\begin{equation}\label{KZeq}
\frac{\partial}{\partial z}F(z)\,=\,\bigg(\frac{e_0}{z}+\frac{e_1}{z-1}\bigg)\cdot F(z)\,.
\end{equation}
\end{defn}
Let us denote by $e_0,e_1$ also a pair of formal non-commutative variables, let $\langle e_0,e_1\rangle_{\mathbb C}$ be the complex vector space spanned by $e_0,e_1$ and let $\mathbb{C}\llangle e_0,e_1\rrangle$ be the $\mathbb{C}$-algebra of formal power series in $e_0,e_1$, with product induced by concatenation of monomials in $e_0,e_1$. Then $\mathbb{C}\llangle e_0,e_1\rrangle$ is isomorphic to the completion of the tensor algebra $T(\langle e_0,e_1\rangle_{\mathbb C})$ with respect to the canonical grading. Through this isomorphism, one can endow $\mathbb{C}\llangle e_0,e_1\rrangle$ with the structure of a (cocommutative) Hopf algebra, with (completed) coproduct $\Delta$ induced by $e_i\mapsto e_i\otimes 1 + 1\otimes e_i$ ($i=0,1$), counit which sends a series to its constant term, and antipode~$a$ induced by $a(w)\coloneq(-1)^{|w|}\tilde{w}$, where if~$w$ is a monomial in $e_0,e_1$,~$\tilde{w}$ is the monomial obtained by reversing the order\footnote{For instance, $\widetilde{e_0e_0e_1e_0}=e_0e_1e_0e_0$.} of the terms of~$w$, and~$|w|$ is the degree of~$w$.
Notice that $\mathbb{C}\llangle e_0,e_1\rrangle$ is in particular a complex vector space, and $e_0,e_1$ can be viewed as elements of $\End(\mathbb{C}\llangle e_0,e_1\rrangle)$, acting by left concatenation on $\mathbb{C}\llangle e_0,e_1\rrangle$; the KZ equation in this setting is known as the \emph{universal KZ equation}.

We also define $\{e_0,e_1\}^\times$ to be the non-commutative monoid which is freely generated by $e_0, e_1$, namely the set of all words formed using non-commutative letters $e_0$ and $e_1$, including the empty word, endowed with the concatenation product.

\subsection{Multiple polylogarithms}

A proof of the following result can be found for instance in \cite{Brownhyperlogs}.

\begin{theorem}\label{thmBrown}
Let $U$ be simply connected and containing the open interval $]0,1[$, and let $\log$ denote the principal branch of the logarithm on $U$. Then there is a unique holomorphic solution~$L$ of the universal KZ equation such that $\lim_{z\to 0}\exp (-e_0\log(z))L(z)=1$.
\end{theorem}

By definition, $\lim_{z\to 0}\exp (-e_0\log(z))L(z)=1$. We call the left-hand side the \emph{regularized value of $L$ at $0$}, and denote it by $L(0)$. Similarly, $\lim_{z\to 1}\exp (-e_1\log(1-z))L(z)$ is finite (by the theory of differential equations with regular singularities); we call it the \emph{regularized value of~$L$ at~$1$}, and we denote it by $L(1)$. 

\begin{defn}\label{defn:Drinfeld}
The element $\Phi\coloneq L(1)\, L(0)^{-1}$ ($=L(1)$) of $\mathbb{C}\llangle e_0,e_1\rrangle$ is called \emph{Drinfeld associator}.
\end{defn}

\begin{rmk}
	One can consider the formal expansions $L(z)=\sum_w L_w(z)\,w$ and $\Phi=\sum_w \zeta_w\,w$ in all possible non-commutative words $w\in\{e_0,e_1\}^\times$. The functions $L_w$ are (regularized) iterated integrals\footnote{The iterated integral over a path $\gamma:[0,1]\to M$ on a manifold $M$ of~$r$ differential one-forms $\omega_1,\ldots ,\omega_r$ is the integral $\int_{0\leq t_r\leq \cdots \leq t_1\leq 1}\gamma^*\omega_1(t_1)\wedge\cdots \wedge\gamma^*\omega_r(t_r)=\int_0^1\gamma^*\omega_1(t_1)\int_0^{t_1}\gamma^*\omega_2(t_2)\cdots \int_0^{t_{r-1}}\gamma^*\omega_r(t_r)$.} from $0$ to $z$ of differential forms $dz/z$ and $dz/(z-1)$, which correspond to appearances in the word~$w$ of the letters~$e_0$ and~$e_1$, respectively, and can be written (in the unit disk) in terms of the multiple polylogarithms \eqref{eqdefMPLs}, the dictionary being
    \begin{equation*}
        L_{e_0^{n_r-1}e_1\cdots e_0^{n_1-1}e_1}(z)=\mathrm{Li}_{n_1,\ldots ,n_r}(z).
    \end{equation*}
    The coefficients $\zeta_w\in\mathbb{R}$, i.e.~the (regularized) special values at~$1$ of the functions $L_w$, are integral representations of the MZVs of equation \eqref{eqdefMZVs}, as one can check that
\begin{equation}\label{eqidMZVs}
   \zeta_{e_0^{n_r-1}e_1\cdots e_0^{n_1-1}e_1}=\zeta(n_1,\ldots ,n_r).
\end{equation}
\end{rmk}

\subsection{Multivalued holomorphic solutions of the KZ equation}\label{sec:holsols}

The following result is essentially a consequence of \thmref{thmBrown}, but we include a proof for completeness.

\begin{theorem}\label{thmKZvectorspace}
	Let $V$ be a complex vector space, let $e_0,\,e_1\in \End(V)$, and suppose that the generating series $L(z)=\sum_w L_w(z)\,w$ and $\Phi=\sum_w \zeta_w\,w$, where $w\in\{e_0,e_1\}^\times$, also define elements of $\End(V)$. Then:
\begin{itemize}
\item[(i)] If $F:U\to V$ is a holomorphic solution of the KZ equation \eqref{KZeq}, with $U$ as in the statement of \thmref{thmBrown}, 
then there exists a constant $C\in V$ such that $F(z)=L(z)\cdot C$. 
\item[(ii)] The regularized values at zero $F(0)\coloneq\lim_{z\to 0}\exp (-e_0\log(z))F(z)$ and at one $F(1)\coloneq\lim_{z\to 1}\exp (-e_1\log(1-z))F(z)$ are finite and satisfy
\begin{equation*}
F(1)\,=\,\Phi\cdot F(0)\,.
\end{equation*}
\end{itemize}
\end{theorem}

\begin{proof}
	\textit{(i)} Consider $L(z)^{-1}\cdot F(z)$, where the inverse is understood w.r.t.~concatenation in $\mathbb{C}\llangle e_0,e_1\rrangle$ (any formal series with non-zero constant term is invertible). This is a holomorphic function from~$U$ to~$V$, whose derivative vanishes, because 
\begin{align*}
\frac{\partial}{\partial z}\big(L(z)^{-1}\cdot F(z)\big)&\,=\, -L(z)^{-1}\bigg(\frac{\partial}{\partial z}L(z)\bigg)L(z)^{-1}\cdot F(z)\,+\,L(z)^{-1}\cdot \frac{\partial}{\partial z}F(z)\\
&\,=\, -L(z)^{-1}\bigg(\frac{e_0}{z}+\frac{e_1}{z-1}\bigg)L(z)L(z)^{-1}\cdot F(z)\,+\,L(z)^{-1}\bigg(\frac{e_0}{z}+\frac{e_1}{z-1}\bigg)\cdot F(z)\\
&\,=\, 0,
\end{align*}
where the first identity follows by representing $L(z)$ as a matrix and using the rule for the derivative of an inverse matrix, and the second identity follows from \eqref{KZeq}. This implies that $L(z)^{-1}\cdot F(z)$ is a constant $C\in V$, i.e.~that $F(z)=L(z)\cdot C$, thus proving~\textit{(i)}. 

\textit{(ii)} Notice that the regularized limits which define $F(1)$ and $F(0)$ exist and are finite, as a consequence of \textit{(i)} and of the existence and finiteness of the same limits for $L(z)$. Moreover, \textit{(i)} yields $F(1)=L(1)\cdot C=\Phi L(0)\cdot C=\Phi\cdot F(0)$, which proves \textit{(ii)} because $L(0)=\mathrm{id}$.
\end{proof}

\begin{rmk}\label{remMPLS}
	The functions $L_w$ can be seen as multivalued functions on the punctured complex plane $\mathbb{C}\smallsetminus\{0,1\}$, i.e.~they are functions on a universal cover $\widetilde{\mathbb{C}\smallsetminus\{0,1\}}$ of $\mathbb{C}\smallsetminus\{0,1\}$, and they provide a (multivalued) analytic continuation of the multiple polylogarithms \eqref{eqdefMPLs}. We define the $\mathbb C$-algebra of multiple polylogarithms $\mathcal{H}\coloneq\mathbb{C}[(L_w)_{w\in\{e_0,e_1\}^\times}]$ as the subalgebra of the holomorphic functions on $\widetilde{\mathbb{C}\smallsetminus\{0,1\}}$ generated by multiple polylogarithms.
\end{rmk}

\begin{rmk}\label{rmk:Terasoma}
	Terasoma considers in \cite{Terasoma} the special representation where $V=\Mat_{n\times 1}(R_{\mathbb C})$, with $R_{\mathbb C}:=R\otimes_{\mathbb Q}\mathbb C$ and $R=\hat\oplus_d R_d$ is a complete graded ring\footnote{One should think of a ring of power series in the degree-one generators. We recall that the notation $\hat\oplus$ stands for the completed direct sum of the graded pieces, which allows to take infinite linear combinations of monomials of different degrees, whereas the usual direct sum only allows to take finite sums.} generated over $\mathbb Q$ by finitely many degree-one elements with homogeneous relations. Furthermore, he assumes that $e_0, e_1\in \Mat_{n\times n}(R_1)\subset \End(V)$ (they act on $V$ via matrix-vector multiplication), namely they induce a \emph{homogeneous rational representation} $\mathbb C\llangle e_0,e_1\rrangle\to \Mat_{n\times n}(R_{\mathbb C})$. In this setting, it is proven (see \cite[Corollary 2.2]{Terasoma}) that the entries of the matrix $\Phi$ representing the Drinfeld associator belong to $\hat\oplus_d \mathcal Z_d\otimes_{\mathbb Q} R_d$, where $\mathcal Z_d$ denotes the $\mathbb Q$-vector space spanned by weight-$d$ MZVs, namely all $\zeta_w$ with $|w|=d$. As we will see, the main results of the present article are obtained within the special setting considered in ref.~\cite{Terasoma}.
\end{rmk}


\subsection{Single-valued multiple polylogarithms}\label{ssec:svmpls}

\begin{proposition}[Brown, \cite{BrownSVMPL}]
(i) There exists a unique automorphism $\psi$ of $\mathbb{C}\llangle e_0,e_1\rrangle$ such that $\psi(e_0)=e_0$ and $\psi(e_1)$ satisfies
\begin{equation*}\label{AlphabetSV}
\psi(\tilde{\Phi}\,e_1\,\tilde{\Phi}^{-1})\,=\,\Phi^{-1}\,e_1\,\Phi,
\end{equation*}
where $\tilde{\Phi}$ is the image of $\Phi$ via the endomorphism of $\mathbb{C}\llangle e_0,e_1\rrangle$ induced by $w\mapsto \tilde w$.\\
(ii) The real-analytic function $\cL :U\to \mathbb{C}\llangle e_0,e_1\rrangle$ defined by
\begin{equation*}
\cL (z)\coloneq L(z)\psi(\widetilde{\overline{L(z)}}),
\end{equation*}
uniquely extends to a real-analytic (single-valued)\footnote{The function $L(z)$ does not extend to a single-valued function on the domain $\mathbb{C}\smallsetminus\{0,1\}$ and it is thus a multivalued function. Conversely, it is usually emphasized that $\mathcal{L}(z)$ is a \emph{single-valued function} on $\mathbb{C}\smallsetminus\{0,1\}$.} function on the whole punctured complex plane $\mathbb{C}\smallsetminus\{0,1\}$.
\end{proposition}

\begin{rmk}\label{rmk:230803}
As in the holomorphic case considered in the previous section, one has the formal expansion $\cL (z)=\sum_{w\in\{e_0,e_1\}^\times}\cL _w(z)\,w$. The coefficients $\cL _w(z)$ define real-analytic (single-valued) functions on $\mathbb{C}\smallsetminus\{0,1\}$, known as \emph{single-valued multiple polylogarithms}, given by linear combinations of products $L_{u}\overline{L_v}$ of holomorphic and antiholomorphic multiple polylogarithms. We call \emph{sv-map} the map
\begin{equation}\label{eqdefsvmap}
    {\rm sv}:L_w\to \cL_w.
\end{equation}
The name sv-map stems from the fact that it maps the multivalued functions~$L_w$ to the single-valued functions~$\cL_w$, and one can prove that it is a morphism of rings from the ring $\mathcal H$ defined in \rmkref{remMPLS} to the ring $\mathcal H^{\rm sv}\coloneq\mathbb{C}[(\cL_w)_{w\in\{e_0,e_1\}^\times}]$ generated by the single-valued multiple polylogarithms.
\end{rmk}

The following result can be immediately deduced from Thm.~0.1 of~\cite{BrownSVMPL}.
\begin{proposition}[Brown,~\cite{BrownSVMPL}]\phantomsection\label{BrownsvHyperlog}
The series $\cL (z)$ is the unique $\mathbb{C}\llangle e_0,e_1\rrangle$-valued real-analytic solution on $\mathbb{C}\smallsetminus\{0,1\}$ to the universal KZ equation such that $\lim_{z\to 0}\exp(-e_0\log|z|^2)\cL (z)$ $=1$ as $z\to 0$ and whose coefficients $\cL _w(z)$ are linear combinations of products $L_{u}(z)\overline{L_v(z)}$ of holomorphic and antiholomorphic multiple polylogarithms. It satisfies an antiholomorphic version of the KZ equation,
\begin{equation*}\label{AntiHoloSV}
\frac{\partial}{\partial \zb}\cL (z)=\cL (z)\cdot\bigg(\frac{e_0}{\zb}+\frac{\psi(e_1)}{\zb-1}\bigg).
\end{equation*}
\end{proposition}
Similarly to the previous section, one can consider the regularized values $\cL(0)$, given by $\lim_{z\to 0}\exp(-e_0\log|z|^2)\cL (z)$ ($=1$), and $\cL(1)$, given by $\lim_{z\to 1}\exp(-e_1\log|1-z|^2)\cL(z)$, which are finite by the theory of linear differential equations with regular singularities.

\begin{defn}\label{defn:Deligne}
The element $\Phi^{\rm sv}\coloneq\mathcal L(1)\,\mathcal L(0)^{-1}$ ($=\mathcal L(1)$) of $\mathbb{C}\llangle e_0,e_1\rrangle$ is called \emph{Deligne associator}. The coefficients $\zeta_{w}^{\rm sv}\in\mathbb{R}$ of its formal expansion $\Phi^{\rm sv}=\sum_{w\in\{e_0,e_1\}^\times} \zeta_w^{\rm sv}\,w$ define the \emph{single-valued MZVs} mentioned in \secref{sec:introduction}, and inspired by \eqn{eqidMZVs} one can also define 
\begin{equation*}
    \zeta^{\rm sv}(n_1,\ldots ,n_r)\coloneq\zeta^{\rm sv}_{e_0^{n_r-1}e_1\cdots e_0^{n_1-1}e_1}.
\end{equation*}
\end{defn}
\begin{rmk}
Single-valued MZVs are the special values of $\cL_w$ at $z\,{=}\,1$, and one can prove\footnote{The coefficients of the linear combinations of products $L_{u}(z)\overline{L_v(z)}$ which define the single-valued multiple polylogarithms $\mathcal L_w(z)$ can be shown to belong to the ring $\mathcal Z$ of MZVs \cite{Schnetz:2013hqa}. It follows that $\mathcal L_w(1)\in\mathcal Z$.} that they are contained in the ring $\mathcal Z:=\mathbb Q[(\zeta_w)_{w\in\{e_0,e_1\}^\times}]$ generated by MZVs. They can be formally obtained from MZVs by specializing the sv-map \eqref{eqdefsvmap} to the value $z\,{=}\,1$, namely
\begin{equation*}
    \zeta_w^{\rm sv}=\cL_w(1)=\mathrm{sv}(L_w(1))=\mathrm{sv}(\zeta_w).
\end{equation*}
At the level of numbers, it is only a conjecture that this map defines a ring morphism from $\mathcal Z$ to the ring $\mathcal Z^{\rm sv}:=\mathbb Q[(\zeta^{\rm sv}_w)_{w\in\{e_0,e_1\}^\times}]\subset \mathcal Z$ generated by single-valued MZVs. It is however known to be a morphism on the corresponding rings of motivic MZVs \cite{Brown:2013gia}.
\end{rmk}

\subsection{Single-valued real-analytic solutions of the KZ equation}

We are now ready to prove a single-valued analogue of \thmref{thmKZvectorspace}, which will be later exploited to prove the main result of this article. 
\begin{theorem}\label{thmKZvectorspaceSV}
Let $V$ be a complex vector space, let $e_0,\,e_1\in \End(V)$, and suppose that the generating series $\cL(z)=\sum_{w\in\{e_0,e_1\}^\times} \cL_w(z)\,w$ and $\Phi^{\rm sv}=\sum_{w\in\{e_0,e_1\}^\times} \zeta^{\rm sv}_w\,w$ also define elements of $\End(V)$. Then:
\begin{itemize}
	\item[(i)] If $\cF:\zC\smallsetminus\{0,1\}\to V$ is a real-analytic solution of the KZ equation \eqref{KZeq} then there exists a purely antiholomorphic function $C\in \overline{{\rm Hol}(\zC\smallsetminus\{0,1\},V)}$ such that $\cF(z)=\cL(z)\cdot C(z)$.
\item[(ii)] If $V\,{=}\,\Mat_{n\times 1}(A)$, with $A$ a complete graded $\mathbb C$-algebra, and $\cF\in \Mat_{n\times 1}(A\hat\otimes_{\zC} \cH\overline{\cH})$, where $\cH$ is the $\mathbb C$-algebra generated by the (holomorphic) multiple polylogarithms and $\overline{\cH}$ is the $\mathbb C$-algebra generated by their (antiholomorphic) complex conjugates\footnote{In other words, we require that the coefficient of each graded component of each entry of  $\cF(z)$ is a linear combination of products $L_u(z)\overline{L_v(z)}$ of holomorphic and antiholomorphic multiple polylogarithms.}, then $C(z){=}C{\in} V$ is constant. In this case, the regularized values $\cF(0){\coloneq}\lim_{z\to 0}\exp (-e_0\log|z|^2)\cF(z)$ and $\cF(1)\coloneq\lim_{z\to 1}\exp (-e_1\log|1-z|^2)\cF(z)$ are finite and satisfy
\begin{equation*}
\cF(1)\,=\,\Phi^{\rm sv}\cdot \cF(0)\,.
\end{equation*}
\end{itemize}
\end{theorem}

\begin{proof}
	\textit{(i)} By the same argument of the proof of \thmref{thmKZvectorspace}, one finds $\frac{\partial}{\partial z}(\cL (z)^{-1}\cdot \mathcal{F}(z))=0$, which implies that there exists an antiholomorphic function $C:\mathbb{C}\smallsetminus \{0,1\}\to V$ such that $\mathcal{F}(z)=\cL (z)\cdot C(z)$, thus proving \textit{(i)}.

	\textit{(ii)} Because $A$ is a complete graded $\mathbb C$-algebra, we can view $C$ as an element of $\Mat_{n\times 1}(A\hat\otimes_{\zC} \overline{{\rm Hol}(\zC\smallsetminus\{0,1\})})$. By \rmkref{rmk:230803}, we know that $\cL \in \zC\llangle e_0,e_1\rrangle\hat\otimes_\zC \cH\overline{\cH}$; this implies that $C\in \Mat_{n\times 1}(A\hat\otimes_{\zC} \overline{{\rm Hol}(\zC\smallsetminus\{0,1\})}\cap \cH\overline{\cH})=V$, where the last equality follows from the fact that the only single-valued antiholomorphic elements of $\cH\overline{\cH}$ are the constants (this follows from \cite[Prop.~2.1]{BrownSVMPL}). We can then apply the same argument of the proof of \thmref{thmKZvectorspace} and deduce that $\cF(0)$ and $\cF(1)$ exist, are finite, and satisfy $\cF(1)=\Phi^{\rm sv}\cdot \cF(0)$, using also that, by \propref{BrownsvHyperlog}, $\cL(0)={\rm id}$.
\end{proof}

\begin{rmk}
A straightforward adaptation of the argument of \cite{Terasoma} implies that, if we are in the setting of \rmkref{rmk:Terasoma}, the entries of the matrix~$\Phi^{\rm sv}$ representing the Deligne associator belong to $\hat\oplus_d \mathcal Z^{\rm sv}_d\otimes_{\mathbb Q} R_d$, where $\mathcal Z^{\rm sv}_d$ denotes the $\mathbb Q$-vector space spanned by weight-$d$ single-valued MZVs.
\end{rmk}


\section{Genus-zero string scattering amplitudes}\label{sec:stringamps}
In this section we will review the structure of tree-level open- and closed-string scattering amplitudes, for which we will formulate the recursion relations in \secref{sec:openrec} and \secref{sec:closedrec}. Furthermore, we will review the result of ref.~\cite{Brown:2019wna}, which relates open- and closed-string amplitudes at genus zero via the sv-map.

Scattering amplitudes in string theories are calculated as worldsheet correlators in a conformal field theory. Hereby each of the~$N$ external string states of momentum $k_i,\, i\,{\in}\,\lbrace 1,...,N\rbrace$, is represented by a vertex operator $V(x_i)$ inserted at a position~$x_i$ on the corresponding surface and those vertex operators are connected by the string Green's functions within the conformal field theory framework. The kinematical data for a scattering amplitude consists of the momenta~$k_i$ of the external string states. In this work we are considering massless string states ($k_i^2=0$) exclusively, which leads to the definition of Mandelstam variables as 
\begin{equation*}
   s_{ij}=s_{ji}=\ap(k_i+k_j)^2=2 \ap k_i{\cdot} k_j,
\end{equation*}
with the inverse string tension $\ap$.
Given that the square refers to the usual Minkowskian scalar product, Mandelstam variables are invariant under Lorentz transformations of the system of momentum vectors $k_i$. Using the Lorentz invariance, one considers the scattering amplitude in a frame where all external momenta add up to zero, $k_1+\cdots +k_N=0$, which translates into the momentum conservation equation for Mandelstam variables
\begin{equation}\label{eqn:MandelstamMomentumConservation}
	\sum_{i=1}^N s_{ij}=0\,,\quad\text{for } j\in\lbrace 1,\ldots,N\rbrace\,.
\end{equation}
Strings are available in open (stretching between two endpoints) and closed (a tiny loop) versions. Correspondingly, open-string vertex operators are inserted at the boundaries of a worldsheet, while closed-string vertex operators are inserted in the bulk of the worldsheet. In this article, we will exclusively consider the leading order contribution to string scattering amplitudes, which corresponds to the simplest topologies: a disc for the open string and a sphere for the closed string.  A string scattering amplitude is obtained by integrating over all possible configurations of positions $x_i$ for the vertex operators $V(x_i)$ on the worldsheet. 

Different open-string vertex operators correspond to different string states. All dependence on the type of external state can be captured by a scattering amplitude in quantum field theory $A_\mathrm{QFT}$, whereas the string-theory character of the scattering amplitude is expressed in terms of the string correction $\mathrm{SC}(\ap)$, so that the full amplitude reads
\begin{equation*}
	A_{\mathrm{QFT}}\times\mathrm{SC}(\ap)\,,
\end{equation*}
where both objects also depend on the momenta $k_i$.
The string correction is a function of the Mandelstam variables exclusively and is itself frequently referred to as \textit{string scattering amplitude} or more frequently as \textit{string integral}.

\subsection*{Open-string scattering amplitudes}
Vertex operators for open string theories are inserted in the worldsheet boundary in a certain succession, referred to as color-ordering. Using conformal symmetry on the worldsheet, the boundary can be mapped to $\zR\cup\infty$ and the integral over the configuration space becomes an iterated integral on the real line. One also needs to fix three insertions\footnote{Albeit conformal symmetry having been used already to move all insertion points to the real line, the integral would be divergent because of overcounting the configuration space. This can be remedied by fixing three of the positions $x_i$, which in the physics literature is often seen as ``dividing by the volume of $\grp{SL}(2,\zR)$''.}, usually at $0,1,\infty$, and eventually genus-zero open-string integrals take the form
\begin{equation}
	\label{eqn:SL2fixed}
	I_N^\sigma=Z_\mathrm{id}(\sigma)=\int_{0\leq x_1\leq\cdots\leq x_n\leq1} 
	\frac{\prod_{1\leq i,j\leq n}(x_j\,{-}\,x_i)^{s_{ij}}}{\prod_{j=0}^{n}\left(x_{\sigma(j+1)}\,{-}\,x_{\sigma(j)}\right)} d x_1 \cdots d x_n\,,
\end{equation}
where $n=N{-}3$, $x_0\,{=}\,0$, $x_{n+1}\,{=}\,1$ and $\sigma$ is a permutation of $\lbrace 0,\ldots, n+1 \rbrace$.
The numerator term in the integrand, containing the dependence on the Mandelstam variables, is usually referred to as the \textit{Koba--Nielsen factor} \cite{Koba:1969rw,Koba:1969kh}, while the denominator is called \emph{Parke--Taylor factor} \cite{Parke:1986gb}. The above formula is a so-called \emph{$Z$-integral} \cite{Mafra:2016mcc}, calculated using the pure-spinor formalism in \cite{Mafra:2011nv,Mafra:2011nw}. The general formula for $Z$-integrals $Z_\rho(\sigma)$ depends on two permutations $\rho$ and $\sigma$, where the permutation $\rho$ acts on the set of indices of the integration variables (in \eqn{eqn:SL2fixed} we chose the case where the permutation $\rho$ is simply the identity). The above expression $I_N^\sigma$ is also the type of open-string integral considered in \rcite{Brown:2019wna}.

The point fixed to $\infty$ does not explicitly show up in \eqn{eqn:SL2fixed}: it vanishes in the Koba--Nielsen factor by means of momentum conservation (see~\eqn{eqn:MandelstamMomentumConservation}). Accordingly, Mandelstam variables incorporating the momentum of the particle associated with the vertex operator inserted at~$\infty$ do not make an appearance either, which renders the remaining Mandelstam variables free parameters, only constrained by ensuring convergence of the integral in \eqn{eqn:SL2fixed} from where an analytic continuation to other regions can be performed.

In order to express full color-ordered amplitudes, certain combinations of the integrals $I_N^\sigma$ are used, which define so-called \emph{$F$-integrals} as \cite{Mafra:2011nv,Mafra:2011nw,Broedel:2013aza}
\begin{equation}\label{eqn:F}
    F_N^\sigma=(-1)^{n}\int_{0\leq x_1\leq \cdots\leq x_n\leq1}\prod_{0\leq i<j\leq n+1}(x_j\,{-}\,x_i)^{s_{ij}}\prod_{k=1}^{n}\sum_{j=1}^{k-1}\frac{s_{\sigma(j)\sigma(k)}}{x_{\sigma(j)}\,{-}\,x_{\sigma(k)}}dx_k,
\end{equation}
with the permutation $\sigma\,{\in}\,\mathfrak{S}_{n}$ acting on the indices $1,\ldots,n$ of the positions $x_i$ and Mandelstam variables $s_{ij}$. These integrals are holomorphic in the Mandelstam variables at the origins $s_{ij}$, with the coefficients of their Taylor expansion given by MZVs \cite{Broedel:2013aza}.
Using these integrals, the full $N$-point color-ordered disk amplitude is given by
\begin{equation*}
    A_N(1,\ldots,N)=\sum_{\sigma\in\mathfrak{S}_{n}}A_\mathrm{YM}\big(\sigma(1),\ldots,\sigma(n),N{-}2,N{-}1,N\big)\,F_N^\sigma,
\end{equation*}
with $A_\mathrm{YM}$ being the corresponding tree-level amplitudes in Yang--Mills theory. The functions $F_N^\sigma$ are linearly independent and form an $n!=(N{-}3)$-dimensional basis of the function space spanned by the $(N{-}1)!$ $Z$-integrals of \eqn{eqn:SL2fixed} \cite{Bern:2008qj,Broedel:2013tta}.

For the purpose of the open-string recursion from \cite{Broedel:2013aza}, which we will review in \secref{sec:openrec}, we will later introduce a deformation of the integrals $F_N^\sigma$ which includes the insertion of an auxiliary parameter $x$, and call this function $\hF_N^\sigma(x)$.

\subsection*{Closed-string scattering amplitudes} 
The calculation of closed-string scattering amplitudes is similar to that of open-string scattering amplitudes, the main difference being that vertex operators $V(z_i)$ for closed-string states are now inserted at complex positions $z_i$ on the Riemann sphere $\PC$, which lacks the ordering prescription present in the open-string case. Upon again addressing the $\mathrm{SL}(2,\zC)$ symmetry by fixing three insertions at $0,1,\infty$, the point at infinity again decouples and one ends up with the closed-string analogue of \eqn{eqn:SL2fixed}:
\begin{equation}\label{eqn:closedamp}
    \mathscr{I}_N^\sigma=J_\mathrm{id}(\sigma)=\pi^{-n} \int_{(\PC)^n} \frac{\prod_{0 \leq i<j \leq n+1}\left|z_j-z_i\right|^{2 s_{i j}}}{\prod_{i=0}^n\left(z_{i+1}\,{-}\,z_i\right) \prod_{i=0}^n\left(\bar{z}_{\sigma(i+1)}\,{-}\,\bar{z}_{\sigma(i)}\right)} d^2 z_1 \cdots d^2 z_n\,,
\end{equation}
where $d^2z=d\Re(z)\wedge d\Im(z)$. Closed-string integrals of this kind are usually referred to as \emph{$J$-integrals} (often before fixing the three points) \cite{Stieberger:2014hba,Mizera:2017cqs}. $J$-integrals $J_\mathrm{\rho}(\sigma)$ in general depend on two permutations, one acting on the indices of the holomorphic part of the integrand, the other acting on the antiholomorphic part ($\sigma$ in the above formula). In \eqn{eqn:closedamp}, we again chose one of these permutations to be the identity to match with the conventions in \cite{Brown:2019wna}.

As in the open-string case, it will be useful to consider certain combinations of the integrals~$\mathscr{I}_N^\sigma$, for which we analogously define
\begin{equation}\label{eqn:cF}
    \cF_N^\sigma=\frac{(-1)^{n}}{\pi^{n}}\int_{(\PC)^n}\frac{1}{\prod_{i=0}^n(\zb_{i+1}\,{-}\,\zb_i)}\prod_{0\leq i<j\leq n+1}|z_j\,{-}\,z_i|^{2s_{ij}}\prod_{k=1}^{n}\sum_{j=1}^{k-1}\frac{s_{\sigma(j)\sigma(k)}}{z_{\sigma(j)}\,{-}\,z_{\sigma(k)}}d^2z_k,
\end{equation}
which is a function of the kinematics of the scattering process through the Mandelstam variables~$s_{ij}$. Just like $F_N^\sigma$, the complex integral $\cF_N^\sigma$ is holomorphic at the origins $s_{ij}=0$, but the coefficients of its Taylor expansion are svMZVs \cite{Schlotterer:2012ny,Stieberger:2013wea,Stieberger:2014hba,Fan:2017uqy,Schlotterer:2018zce,Brown:2019wna,VanZerb}.
In order to formulate the recursion for closed-string amplitudes in \secref{sec:closedrec}, we will there define a deformation of these integrals $\cF_N^\sigma$, parametrized by an auxiliary point $z$ and denote the deformed function as $\hcF_N^\sigma(z)$.

\subsection*{Open vs closed genus-zero amplitudes}
The sv-map on MZVs mentioned in the introduction can be extended to more general \emph{period integrals}\footnote{A period integral on an algebraic variety~$X$ is the pairing, prescribed by Grothendieck's algebraic de Rham theorem, of an algebraic de Rham cohomology class~$\eta$ with a singular homology class~$\gamma$ of the manifold $X(\mathbb C)$ given by considering the complex points of~$X$. If~$X$ is affine then one can think of $\eta$ as a global algebraically defined differential form.} $\int_\gamma\eta$ on certain kinds of algebraic varieties~$X$ in which there is a natural way to map a cycle $\gamma$ in the (relative) singular homology of $X(\mathbb C)$ to an algebraic de Rham cohomology class $\nu_\gamma$ \cite{Brown:2018omk}. It was also proven in \rcite{Brown:2018omk} that, under certain technical assumptions on the convergence of the integrals involved, the image of $\int_\gamma\eta$ under the sv-map can be computed as $\int_{X(\mathbb C)}\nu_\gamma\wedge \overline{\eta}$. 

An extension of this machinery to the case of ``twisted periods'', namely period integrals $\int_\gamma u\,\eta$, the ``twist''~$u$ being a multivalued function as in \rmkref{rmk:twistedderham} below, was applied in \rcite{Brown:2019wna} to prove that the closed-string $J$-integrals~$\mathscr{I}_N^\sigma$ from \eqref{eqn:closedamp} can be obtained from the open string $Z$-integrals~$I_N^{\sigma}$ from \eqref{eqn:SL2fixed} via the sv-map as follows. If we denote
\begin{equation}\label{eqn:PT}
  u\coloneq \prod_{1\leq i<j\leq n}(x_j\,{-}\,x_i)^{s_{ij}},\quad \eta_{\sigma}\coloneq\frac{dx_1\cdots dx_n}{\prod_{j=0}^n(x_{\sigma(j+1)}\,{-}\,x_{\sigma(j)})},
\end{equation}
and we denote by $\gamma$ the cycle\footnote{Notice that this cycle, which is viewed here as a half-dimensional cycle in the moduli space $\mathfrak{M}_{0,N}(\mathbb C)$ of complex algebraic curves, can be also be seen as a top-dimensional connected component of the moduli space $\mathfrak{M}_{0,N}(\mathbb R)$ of real algebraic curves.} $0\leq x_1\leq\cdots\leq x_n\leq1$ of the manifold $(\PC)^n$, then the open-string integral \eqref{eqn:SL2fixed} can be written as a twisted period as
\begin{equation*}
    I_N^{\sigma}=\int_{\gamma}u\,\eta_\sigma.
\end{equation*}
It was proven in \rcite{Brown:2019wna} that the cycle $\gamma$ is naturally mapped to the differential form 
\begin{equation}\label{defnugam}
    \nu_{\gamma}\coloneq \frac{(-1)^n}{(2\pi i)^n}\frac{dx_1\cdots dx_n}{\prod_{i=0}^n(x_{i+1}\,{-}\,x_i)},
\end{equation}
which in turn implies that
\begin{equation*}
    \mathrm{sv}(I_N^{\sigma})=\int_{(\PC)^n}|u|^2\,\nu_{\gamma}\,\eta_\sigma=\mathscr{I}_N^\sigma,
\end{equation*}
namely that the closed string integrals $\mathscr{I}_N^\sigma$ are the image of the open string integrals $I_N^{\sigma}$ under the (twisted) sv-map. As a corollary, it was possible to re-derive and assign a deep geometrical meaning to the so-called KLT relations \cite{Kawai:1985xq} between open and closed genus-zero string integrals.

Furthermore, it was also shown in \rcite{Brown:2019wna} that, upon Laurent expanding both the open string $Z$-integrals~$I_N^{\sigma}$ and the closed-string $J$-integrals~$\mathscr{I}_N^\sigma$, whose coefficients are given by (classical) period integrals over $\gamma$ and $(\PC)^n$, respectively, one can apply the same prescription \eqref{defnugam} to find that the image of $I_N^{\sigma}$ under the sv-map is $\mathscr{I}_N^\sigma$ also when the sv-map in the classical period setting is applied coefficient-wise to the Laurent expansion of $I_N^{\sigma}$, thus settling Stieberger's conjecture \cite{Stieberger:2013wea}.


\section{Selberg integrals}\label{sec:Selberg}

Classical Selberg integrals are functions of three complex variables, defined as multiple integrals, which generalize the Euler beta function. They are named after Atle Selberg, who found an explicit formula for them in terms of Gamma functions~\cite{Selberg:411367}:
\begin{align*}
    S_L(\alpha,\beta,\gamma)&\coloneq\int_{[0,1]^L}\prod_{1\leq i<j\leq L}|t_i\,{-}\,t_j|^{2\gamma}\prod_{i=1}^Lt_i^{\alpha-1}(1\,{-}\,t_i)^{\beta-1}dt_i\notag\\
    &=\prod_{j=0}^{L-1}\frac{\Gamma(\alpha+j\beta)\Gamma(\beta+j\gamma)\Gamma(1+(j+1)\gamma)}{\Gamma(\alpha+\beta+(n+j-1)\gamma)\Gamma(1+\gamma)}.
\end{align*}
Generalizations thereof are related to the theory of generalized hypergeometric functions, where the three variables $\alpha,\beta,\gamma$ are replaced by three sets of variables $\alpha_i,\beta_i,\gamma_{ij}$ and the integration domain is replaced by a simplex $\Delta$ (possibly of dimension lower than $L$). These functions were subsequently introduced and studied independently by Gelfand \cite{Gelfand} and Aomoto \cite{Aomoto87}, using tools from twisted homology theory. The integrals investigated by Aomoto are of the form \cite{Aomoto87}
\begin{equation}\label{eqn:Aomoto}
    A_{p,L}(x_1,\ldots,x_p)=\int_\Delta \prod_{\substack{i,j=1\\i<j}}^L(x_i\,{-}\,x_j)^{\lambda_{ij}}\prod_{k=p+1}^Ldx_k,
\end{equation}
depending on $1<p\leq L$ points $x_1\ldots,x_p$, with suitable coefficients $\lambda_{ij}$.
We refer to \rcite{forrester2007importance} for history and application of integrals of the Selberg type. 

In this section we will review the notions about Selberg integrals necessary to deduce recursions on string amplitudes in \secref{sec:openrec} and \secref{sec:closedrec} below. In \secref{sec:openSelberg} we will introduce a class of generalized Selberg integrals, which we refer to as open Selberg integrals and which will be related to open-string amplitudes. Similarly, in \secref{sec:closedSelberg} we will consider real-analytic single-valued analogues of the open Selberg integrals, which we will refer to as closed Selberg integrals, and which will be related to closed string amplitudes later on. 
Both open and closed Selberg integrals are functions of one variable: in \secref{sec:LimitopenSelberg} and \secref{sec:LimitclosedSelberg} we consider the limits where this variable approaches the two special points zero and one. In \secref{sec:StringSelberg} we will discuss more closely the momentum configuration relevant for the physical interpretation of the Selberg integrals in the limits.

\subsection{The class of open Selberg integrals}\label{sec:openSelberg}

We consider now a special subclass of generalized Selberg integrals, which was also studied in \rcite{Terasoma} and later in \cite{BK,Kaderli}.
\begin{defn}\label{defn:classicselb}
    Let $L\,{\geq}\, 4$, $0\,{=}\,x_1\,{<}\,x_3\,{<}\,x_2\,{=}\,1$ and for every $k\,{\in}\,\{4,\ldots,L\}$ let $i_k\,{\in}\,\{1,\ldots,L\}\smallsetminus\{k\}$. Then we define
    \begin{align}
    	\label{eqn:Selbergzero}
    	S[i_4,\dots,i_L](x_3)\coloneq\int_{0\leq x_L\leq\ldots\leq x_4\leq  x_3}\prod_{\substack{i,j=1\\i<j}}^L|x_{ij}|^{s_{ij}}\prod_{k=4}^{L}\frac{dx_k}{x_{ki_k}}\,,
    \end{align}
    where $x_{ij}\coloneq x_i-x_j$ and provided that the exponents $s_{ij}\in\mathbb C$ are such that the integral on the right-hand side is absolutely convergent\footnote{The region of convergence can be explicitly deduced using for instance the result in the Appendix of ref.~\cite{VanZerb}, but here we content ourselves to mention that $\mathrm{Re}(s_{ij})>0$ is a sufficient condition for the absolute convergence.}. We also denote by $S[i_4,\dots,i_L](x_3)$ the holomorphic multivalued function of $x_3\in\mathbb C\smallsetminus \{0,1\}$ obtained by analytically continuing the iterated integral~\eqref{eqn:Selbergzero}, and we call it an \emph{open Selberg integral}\footnote{As mentioned earlier, many kinds of integrals have been called Selberg integrals in the literature. Since those in~\eqref{eqn:Selbergzero} are the only kind considered in this article, no confusion will arise by assigning them this name. We remark that they can also be called Aomoto--Gelfand hypergeometric function, following the terminology of \cite{VanZerb}.}.
\end{defn}
\noindent Note that since $|x_{12}|=1$, open Selberg integrals do not depend on $s_{12}$.

For the following, it will be useful to rewrite open Selberg integrals using the substitution $w_k=x_k/x_3$ for $k=4,\ldots,L$
\begin{align}
    \label{eqn:opensub}
    S&[i_4,\dots,i_L](x_3)\notag\\
    &=\int_{0\leq w_L\leq\ldots\leq w_4\leq1}x_3^{s_{13}}(1\,{-}\,x_3)^{s_{23}}\prod_{\substack{i,j=4\\i<j}}^Lx_3^{s_{ij}}w_{ij}^{s_{ij}}\prod_{k=4}^Lx_3^{s_{1k}}w_k^{s_{1k}}(1\,{-}\,x_3w_k)^{s_{2k}}x_3^{s_{3k}}(1\,{-}\,w_k)^{s_{3k}}\frac{dw_k}{w_{ki_k}}\notag\\
    &=x_3^{\smax}(1\,{-}\,x_3)^{s_{23}}\int_{0\leq w_L\leq\ldots\leq w_4\leq1}\prod_{\substack{i,j=4\\i<j}}^Lw_{ij}^{s_{ij}}\prod_{k=4}^Lw_k^{s_{1k}}(1\,{-}\,x_3w_k)^{s_{2k}}(1\,{-}\,w_k)^{s_{3k}}\frac{dw_k}{w_{ki_k}},
\end{align}
where we set (and which will be used throughout the rest of the article)
\begin{equation}
\label{eqn:subw}
	w_1\coloneq 0,\quad w_2\coloneq 1/x_3,\quad w_3\coloneq 1
\end{equation}
and defined\footnote{This notation is motivated by the fact that $\smax$ will be the maximal eigenvalue of the matrix~$e_0$ for the KZ equation satisfied by an appropriate vector of open Selberg integrals (see \secref{sec:openrec}, \secref{sec:closedrec} and \cite[Lemma 5.5]{Terasoma}).}
\begin{equation}\label{eqn:smax}
    \smax\coloneq\sum_{\substack{1\leq i<j\leq L\\i,j\neq2}}s_{ij}\,.
\end{equation}

\begin{rmk}
    The open Selberg integrals defined in \eqref{eqn:Selbergzero} are related to open-string tree-level amplitudes as in \eqn{eqn:SL2fixed} as follows: once we keep one of the points in \eqn{eqn:SL2fixed} unintegrated, that is, we set $n=L{-}3$, the integral in \eqn{eqn:SL2fixed} is precisely of the type considered here: it is an open Selberg integral.
\end{rmk}

\begin{rmk}
    The above definition of open Selberg integrals can be generalized for multiple unintegrated points as in \eqn{eqn:Aomoto} or as defined in \cite{BK}, but for our purpose one unintegrated point $x_3$ will be sufficient. In Aomoto's language of \eqn{eqn:Aomoto}, we are dealing here with three unintegrated points $x_1,\,x_2,\,x_3$, two of which are fixed, namely $x_1=0,\, x_2=1$, so that $x_3$ is the free parameter of the function $S[i_4,\ldots,i_L](x_3)$.
\end{rmk} 

\begin{rmk}\label{rmk:twistedderham}
Aomoto's original study \cite{Aomoto87} of Selberg integrals uses tools from ``twisted de Rham cohomology'', i.e.~cohomology with coefficients in a vector bundle equipped with a flat connection~$\nabla$. In the case of open Selberg integrals, the ``twist'' is given (for fixed $x_3\in\mathbb C \smallsetminus\{0,1\}$ and fixed values of the parameters $s_{ij}\in\mathbb C\smallsetminus\mathbb Z$) by the multivalued function
\begin{equation}\label{eq:twistKNfactor}
    u\defeq \prod_{k=4}^L x_k^{s_{1k}}(1\,{-}\,x_k)^{s_{2k}}\,\prod_{3\leq i<j\leq L}^L x_{ij}^{s_{ij}}
\end{equation}
on $\mathrm{Conf}_{L-3}(\mathbb P^1_{\mathbb C}\smallsetminus \{0,1,x_3,\infty\})$, the configuration space of $L{-}3$ distinct points on $\mathbb P^1_{\mathbb C}\smallsetminus \{0,1,x_3,\infty\}$. The function~$u$ can be seen as the Koba--Nielsen factor of an open string integral (see eq. \eqref{eqn:PT}), and its logarithmic derivative $\omega=d\log (u)$ can be used to define the flat connection $\nabla_\omega\defeq d+\omega$. Similarly to the open string integrals from the previous section, open Selberg integrals can be seen as twisted periods arising from the pairing of elements of the (algebraic) twisted de Rham cohomology group $H^{L-3}_{\mathrm dR}(\mathrm{Conf}_{L-3}(\mathbb P^1_{\mathbb C}\smallsetminus \{0,1,x_3,\infty\}),\nabla_{\omega})$, which can be represented by global holomorphic $(L{-}3,0)$ forms, with elements of the corresponding twisted Betti homology group $H_{L-3}^{\mathrm B}(\mathrm{Conf}_{L-3}(\mathbb P^1_{\mathbb C}\smallsetminus \{0,1,x_3,\infty\}),\mathcal L_{\omega})$, i.e.~the homology group with coefficients in the local system $\mathcal L_{\omega}$ induced by $\nabla_{-\omega}$. Notice that a differential form~$f$ is exact in the twisted complex if $f=\nabla_\omega(g)$, and this is equivalent to $uf=d(ug)$.

It follows from \cite{Aomoto87} that $H^{L-3}_{\mathrm dR}(\mathrm{Conf}_{L-3}(\mathbb P^1_{\mathbb C}\smallsetminus \{0,1,x_3,\infty\}),\nabla_{\omega})$ has dimension $(L{-}2)!$, and various bases of this space can be found in the literature \cite{Aomoto87, Terasoma, Brown:2019wna, Mizera:2019gea, Kaderli}. In particular, there are several linear relations among the $(L{-}1)^{L-3}$ open Selberg integrals \eqref{defn:classicselb}.
\end{rmk}

\subsection{Limit behavior of open Selberg integrals.}
\label{sec:LimitopenSelberg}
In the following, we are going to investigate the regularized limit behavior of the open Selberg integrals when the auxiliary point $x_3$ approaches $x_1=0$ or $x_2=1$. The limit behavior will be exploited in the calculation of the boundary values of the KZ equation involved in the recursion for open-string integrals in \secref{sec:openrec}.

\begin{proposition}[Terasoma, \cite{Terasoma}]\label{prop:openlimits}
Let $L\geq4$, then we have the following limits:
    \begin{enumerate}[label=(\roman*)]
        \item  For $x_3\to0$ the regularized limit is only non-vanishing if $i_k\,{\neq}\,2$ for all $k\in\{4,\ldots,L\}$, and then we find
                \begin{align*}
                    \qquad\lim_{x_3\to0}x_3^{-\smax}S[i_4,\ldots,i_L](x_3)=\int_{0\leq w_L\leq\ldots\leq w_4\leq 1}\prod_{k=4}^{L}\frac{dw_k}{w_{ki_k}}\,w_k^{s_{1k}}(1\,{-}\,w_k)^{s_{3k}}\prod_{\substack{i,j=4\\i<j}}^Lw_{ij}^{s_{ij}}.
                \end{align*}
        \item  For the limit $x_3\to1$ the regularized behavior of the open Selberg integrals is
        		\begin{align*}
        			\qquad\lim_{x_3\to1}(1\,{-}\,x_3)^{-s_{23}}S[i_4,\ldots,i_L](x_3)=\int_{0\leq x_L\leq\ldots\leq x_4\leq 1}\prod_{\substack{i,j=1\\i<j;\,i,j\neq3}}^{L}x_{ij}^{s'_{ij}}\prod_{k=4}^{L}\left.\frac{dx_{k}}{x_{ki_k}}\right|_{x_3=1},
        		\end{align*}
                where $s'_{ij}\coloneq s_{ij}$ if $i\neq2$ and $s'_{2j}\coloneq s_{2j}+s_{3j}$.
    \end{enumerate}
\end{proposition}
\begin{proof}
For the proof, see ref.~\cite[Lemmas 5.5 and 5.7]{Terasoma} or in a slightly different notation in ref.~\cite[Sec.~2.4]{BK}.
\end{proof}

\begin{rmk}\label{rmk:nos12}
	We find that the result \textit{(i)} of the above proposition is independent of any of the $s_{2\bullet}$, due to the limit $\lim_{x_3\to0}(1\,{-}\,x_3w_k)^{s_{2k}}=1$ after the substitution. Note further, that terms $1\,{-}\,w_k$ now have the exponent $s_{3k}$. The physical interpretation of this fact will be discussed in~\secref{sec:StringSelberg}.
\end{rmk}

We now furthermore consider the regularized limits of the open Selberg integrals where additionally the variables $s_{3\bullet}$ related to the auxiliary point $x_3$ tend to zero. This step is important, as it renders one of the $L\,{-}\,3$ integrations trivial in the case $x_3\to0$. This limit will return physically meaningful objects, as discussed later.

\begin{proposition}\label{prop:openphysical}
    One finds the following limits:
    \begin{enumerate}[label=(\roman*)]
        \item  For arbitrary labels $i_k\in\{1,\ldots,L\}\smallsetminus\{k\}$ we find
                \begin{align*}
                     \lim_{s_{3\bullet}\to0}\lim_{x_3\to0}x_3^{-\smax}S[i_4,\ldots,i_L](x_3)=\int_{0\leq w_L\leq\ldots\leq w_4\leq 1}\prod_{\substack{i,j=4\\i<j}}^Lw_{ij}^{s_{ij}}\prod_{k=4}^{L}w_k^{s_{1k}}\frac{dw_k}{w_{ki_k}}
                \end{align*}
                if $i_k\,{\neq}\,2$ for all $k\in\{4,\ldots,L\}$ and vanishing otherwise.\\
                If now $i_k\,{\neq}\,2$ for all $k\in\{4,\ldots,L\}$ and additionally $i_k\neq3$ for all $k\in\{4,\ldots,L\}$, the above result can be simplified to read
                \begin{align*}
                \lim_{s_{3\bullet}\to0}\lim_{x_3\to0}&x_3^{-\smax}S[i_4,\ldots,i_L](x_3)\notag\\
                &=\frac{1}{s_{145\cdots L}}\int_{0\leq y_L\leq\ldots\leq y_5\leq 1}\frac{1}{1\,{-}\,y_{i_4}} \prod_{\substack{i,j=5\\i<j}}^Ly_{ij}^{s_{ij}}\prod_{k=5}^{L}y_k^{s_{1k}}(1\,{-}\,y_k)^{s_{4k}}\frac{dy_k}{y_{ki_k}}
                \end{align*}
               where $s_{145\cdots L}\coloneq\sum_{k=4}^Ls_{1k}+\sum_{4\leq i<j\leq L}s_{ij}$ and $y_1\coloneq0$, $y_4\coloneq1$.
        \item  For the other limit one has
                    \begin{align*}
        			\ \ \qquad\lim_{s_{3\bullet}\to0}\lim_{x_3\to1}(1\,{-}\,x_3)^{-s_{23}}S[&i_4,\ldots,i_L](z_3)\notag\\
			&=\int_{0\leq x_L\leq\ldots\leq x_4\leq 1}\prod_{\substack{i,j=4\\i<j}}^{L}x_{ij}^{s_{ij}}\prod_{k=4}^{L}x_k^{s_{1k}}(1\,{-}\,x_k)^{s_{2k}}\left.\frac{dx_{k}}{x_{ki_k}}\right|_{x_3=1}.
        		\end{align*}
    \end{enumerate}
\end{proposition}
\begin{proof}
    Statement \textit{(ii)} follows immediately from the result \textit{(ii)} of \propref{prop:openlimits} when taking the limit $s_{3\bullet}\to 0$ and rewriting the products as 
    \begin{equation*}
        \prod_{\substack{i,j=1\\i<j;\,i,j\neq3}}^{L}|x_{ij}|^{s_{ij}}=\prod_{k=4}^L x_k^{s_{1k}}(1\,{-}\,x_k)^{s_{2k}}\prod_{\substack{i,j=4\\i<j}}^{L}x_{ij}^{s_{ij}}\,.
    \end{equation*}
    For statement \textit{(i)} the limit $s_{3\bullet}\to 0$ can also be taken immediately using the result \textit{(i)} of \propref{prop:openlimits}. If at least one index $i_k$ is equal to 2, the expression vanishes and otherwise the limit results in 
    \begin{equation*}
        \lim_{s_{3\bullet}\to0}\lim_{x_3\to0}x_3^{-\smax}S[i_4,\ldots,i_L](x_3)=\int_{0\leq w_L\leq\ldots\leq w_4\leq 1}\prod_{\substack{i,j=4\\i<j}}^Lw_{ij}^{s_{ij}}\prod_{k=4}^{L}w_k^{s_{1k}}\frac{dw_k}{w_{ki_k}}\,.  
    \end{equation*}
    If we have $i_k\neq3$ for all $k$, we will be able to perform the outermost integration in $w_4$: we employ the substitution $y_k=w_k/w_4,\ k=5,\ldots,L$, i.e.~rescaling all but the outermost integration variable by $w_4$ (and define $y_1\coloneq0$ and $y_4\coloneq1$). This leaves us with
    \begin{align*}
        \lim_{s_{3\bullet}\to0}&\lim_{x_3\to0}x_3^{-\smax}S[i_4,\ldots,i_L](x_3)\\&=\underbrace{\int_0^1 w_4^{\sum_{k=4}^Ls_{1k}+\sum_{4\leq i<j\leq L}s_{ij}-1}dw_4}_{=\frac{1}{\sum_{k=4}^Ls_{1k}+\sum_{4\leq i<j\leq L}s_{ij}}}\int_{0\leq y_L\leq\ldots\leq y_5\leq 1}\frac{1}{1\,{-}\,y_{i_4}}\prod_{\substack{i,j=5\\i<j}}^Ly_{ij}^{s_{ij}}\prod_{k=5}^{L}y_k^{s_{1k}}(1\,{-}\,y_k)^{s_{4k}}\frac{dy_k}{y_{ki_k}}\,,
    \end{align*}
    which proves the statement.
\end{proof}

\begin{rmk}
    Choosing $i_k\leq k$ and writing
    \begin{equation*}
    	I^{i_4\cdots i_L}=\int_{0\leq x_L\leq\ldots\leq x_4\leq1}\prod_{1\leq i<j\leq L}|x_{ij}|^{s_{ij}}\,\omega,
    \end{equation*}
    where
    \begin{equation*}
        \omega=\frac{dx_4\cdots dx_L}{\prod_{k=4}^Lx_{ki_k}}\,,
    \end{equation*}
    the results of \propref{prop:openphysical} define bases for (i) $(L{-}1)$-point and (ii) $L$-point open-string tree-level scattering amplitudes according to the conventions of \cite{Brown:2019wna}. Note that $\omega$ is not precisely of the Parke--Taylor form (see \secref{sec:stringamps}), but for example Thms.~6.9 and 6.10 of \rcite{Brown:2019wna} provide the translation between the above (so-called Aomoto basis) and the Parke--Taylor basis.
\end{rmk}

\subsection{The class of closed Selberg integrals}\label{sec:closedSelberg}
Let us now define ``complex analogues'' of the open Selberg integrals introduced in the previous section, in the sense that they will be defined by integrals over the complex projective line, rather than on the real line. These integrals will constitute the building blocks for the recursion of closed-string integrals in \secref{sec:closedrec}.
\begin{defn}\label{defn:closedSelberg}
    Let $L\geq4$ and $z_1=0$, $z_2=1$. We call \emph{closed Selberg integral} the holomorphic (single-valued) function of $z_3\in\mathbb C\smallsetminus\{0,1\}$ defined, for indices $i_k\in\{1,\ldots,L\}\smallsetminus\{k\}$, by
	\begin{equation}
		\label{eqn:complexSel}
		\SC[i_4,\ldots,i_L](z_3)\coloneq\frac{1}{(-2\pi i)^{L-3}}\int_{(\PC)^{L-3}}\frac{\zb_3}{\zb_L}\prod_{\substack{i,j=1\\i<j}}^L|z_{ij}|^{2s_{ij}}\prod_{k=4}^{L}\frac{dz_kd\zb_k}{z_{ki_k}\zb_{k(k-1)}}\,,
	\end{equation}
    provided that the exponents $s_{ij}\in\zC$ are such that the integral on the right-hand side is absolutely convergent\footnote{This region of convergence can be explicitly deduced using for instance the result in the Appendix of \cite{VanZerb}, but here we content ourselves to mention that this is true for sufficiently small $\mathrm{Re}(s_{ij})>0$.}.
\end{defn}

As for the open Selberg integrals from \secref{sec:openSelberg}, it will be useful for later applications to rewrite the closed Selberg integrals using the substitution $w_i=z_i/z_3$ for $i=4,\ldots,L$, which yields
\begin{align}\label{eqn:closedsub}
    &\SC[i_4,\ldots,i_L](z_3)=\frac{1}{(-2\pi i)^{L-3}}|z_3|^{2\smax}|1\,{-}\,z_3|^{2s_{23}}\notag\\
    &\qquad\times\int_{(\PC)^{L-3}}\frac{1}{\wb_L}\prod_{\substack{i,j=4\\i<j}}^{L}|w_{ij}|^{2s_{mn}}\prod_{k=4}^L|w_k|^{2s_{1k}}|1\,{-}\,z_3w_k|^{2s_{2k}}|1\,{-}\,w_k|^{2s_{3k}}\frac{dw_kd\wb_k}{w_{ki_k}\wb_{k(k-1)}}\,,
\end{align}
where we defined $w_1\coloneq0,\,w_2\coloneq1/z_3,\,w_3\coloneq1$, analogous to \eqn{eqn:subw}. Just as in the case of \eqn{eqn:opensub} for open Selberg integrals, we find the exponent of $|z_3|^2$ to be $\smax$ from \eqn{eqn:smax}.
\begin{rmk}
    In the same way as open Selberg integrals can be associated with open string amplitudes at tree-level, the closed Selberg integrals defined in \eqref{eqn:complexSel} are related to the closed-string amplitudes from \eqn{eqn:closedamp}: keeping one of the points in the integrals \eqn{eqn:closedamp} unintegrated, we find the formula \eqref{eqn:complexSel} for closed Selberg integrals.
\end{rmk}

\begin{rmk}\label{rmk:svSelberg}
	The holomorphic part of the integrand of the closed Selberg integrals is exactly the same as for the open Selberg integrals defined in \eqn{eqn:Selbergzero}, while the antiholomorphic part comes from the sv-map formulated in refs.~\cite{Brown:2018omk,Brown:2019wna} and explained in \secref{sec:stringamps}: one way to see this is to start with the open Selberg integral and switch to cubical coordinates, i.e.~one substitutes $y_i=x_i/x_{i-1}$, $i=4,\ldots,L$, so that the integration in the $y_i$s is over the hypercube $\gamma=[0,1]^{L-3}$. From here we can use \cite[Example 4.5.2]{Brown:2018omk} to yield us the form $\nu_\gamma$ corresponding to the hypercube. Thus, we can apply the sv-map and reversing the substitution done before brings us to the closed Selberg integral of \eqn{eqn:complexSel}, showing that the closed Selberg integrals are the images of the open Selberg integrals under the sv-map, $\mathrm{sv}\big(S[i_4,\ldots,i_L](z_3)\big)=\SC[i_4,\ldots,i_L](z_3)$.
\end{rmk}

\begin{rmk}\label{rmk:svMPLscplxSelberg}
    By the methods of~\cite{Brown:2019wna} or~\cite{VanZerb}, one can prove that any closed Selberg integral has a Laurent expansion at the origin in the variables $s_{ij}$. By repeatedly applying Thm.~6.7 of~\cite{VanZerb}, it can also be proven that the coefficients of this Laurent expansion are linear combinations, with coefficients in the $\mathbb Q$-algebra generated by single-valued multiple zeta values, of single-valued multiple polylogarithms. This generalizes the proof of Thm.~7.1 of~\cite{VanZerb}. 
\end{rmk}

\subsection{Limit behavior of closed Selberg integrals.} 
\label{sec:LimitclosedSelberg}
As for the case of open Selberg integrals shown before, we can make statements about the regularized limits of the closed Selberg integrals when the auxiliary variable $z_3$ tends to $z_1=0$ or $z_2=1$.
\begin{proposition}\label{prop:limitsclosed}
    For $L\geq4$ consider the closed Selberg integral $\SC[i_4,\ldots,i_L](z_3)$. Then we have the following limits:
    \begin{enumerate}[label=(\roman*)]
        \item  If $i_k\,{\neq}\,2$ for all $k\in\{4,\ldots,L\}$, the $z_3\to0$ limit is
       		 \begin{align}
                \label{eqn:limit0firstcomps}
                    \qquad\lim_{z_3\to0}&|z_3|^{-2\smax}\SC[i_4,\ldots,i_L](z_3)\notag\\
		            &=\frac{1}{(-2\pi i)^{L-3}}\int_{(\PC)^{L-3}}\frac{1}{\wb_L}\prod_{\substack{i,j=4\\i<j}}^L|w_{ij}|^{2s_{ij}}\prod_{k=4}^{L}|w_k|^{2s_{1k}}|1\,{-}\,w_k|^{2s_{3k}}\frac{dw_kd\wb_k}{w_{ki_k}\wb_{k(k-1)}},
                \end{align}
                and it vanishes otherwise.
        \item For the limit $z_3\to1$ the result is
        		\begin{align}
		              \label{eqn:limit1no23}
        			\qquad\quad\lim_{z_3\to1}|1\,{-}\,z_3|^{-2s_{23}}&\SC[i_4,\ldots,i_L](z_3)\notag\\
			&=\frac{1}{(-2\pi i)^{L-3}}\int_{(\PC)^{L-3}}\frac{1}{\zb_L}\prod_{\substack{i,j=1\\i<j;\,i,j\neq 3}}^{L}|z_{ij}|^{2s'_{ij}}\prod_{k=4}^{L}\frac{dz_kd\zb_k}{z_{ki_k}\zb_{k(k-1)}}\bigg|_{z_3=1}\,,
        		\end{align}
                where $s'_{ij}\coloneq s_{ij}$ if $i,j\neq2$ and $s'_{2j}\coloneq s_{2j}+s_{3j}$.
    \end{enumerate}
\end{proposition}
\begin{proof}
	\textit{(i)}: Note that after the substitution $w_i=z_i/z_3$ for $i=4,\ldots,L$ (see \eqn{eqn:closedsub}) we have the factor $|z_3|^{2\smax}$ in front of the closed Selberg integrals which cancels with the regularization factor $|z_3|^{-2\smax}$. Then one can take the limit $z_3\to0$ and ends up (using $w_1=0$) with \eqn{eqn:limit0firstcomps}, where -- as before in the case of open Selberg integrals -- again if there is at least one index $i_k\,{=}\,2$, the corresponding term scales as
    \begin{equation*}
        \lim_{x_3\to0}\frac{1}{w_k-\frac{1}{x_3}}=0.
    \end{equation*}
    \textit{(ii)}: The regulator $|1\,{-}\,z_3|^{-s_{23}}$ cancels the $|1\,{-}\,z_3|^{s_{23}}$ in the integral and all other terms in the integral neither go to zero, nor diverge, so that we can then take the limit and arrive at \eqn{eqn:limit1no23}, where the exponents $s_{2k}$ of terms $|1\,{-}\,z_k|$ combine with exponents $s_{3k}$ of $|z_{3k}|$ in the limit to give $s'_{2k}=s_{2k}\,{+}\,s_{3k}$, while the other Mandelstam variables remain unchanged.
\end{proof}

As before for the open Selberg integrals, we also want to investigate the physical limit ${s_{3\bullet}\to0}$ of the auxiliary Mandelstam variables vanishing, which will allow us to make a connection to closed-string tree-level scattering amplitudes.

\begin{proposition}\label{prop:limitsclosedphysical}
	In the setting of \propref{prop:limitsclosed} one finds the following physical limits:
    \begin{enumerate}[label=(\roman*)]
        \item   For the limit $z_3\to0$, one finds
                \begin{align*}
                    \qquad\lim_{s_{3\bullet}\to0}\lim_{z_3\to0}&|z_3|^{-2\smax}\SC[i_4,\ldots,i_L](z_3)\notag\\
		            &\qquad=\frac{1}{(-2\pi i)^{L-3}}\int_{(\PC)^{L-3}}\frac{1}{\wb_L}\prod_{\substack{i,j=4\\i<j}}^L|w_{ij}|^{2s_{ij}}\prod_{k=4}^{L}|w_k|^{2s_{1k}}\frac{dw_kd\wb_k}{w_{ki_k}\wb_{k(k-1)}}
                \end{align*}
                if $i_k\,{\neq}\,2$ for all $k\in\{4,\ldots,L\}$ and the expression vanishes otherwise.\\
                If additionally $i_k\neq3\ \forall k\in\{4,\ldots,L\}$, the above expression simplifies to read
                \begin{align}
                \label{eqn:limit0firstcompsphysical}
                    \qquad&\lim_{s_{3\bullet}\to0}\lim_{z_3\to0}|z_3|^{-2\smax}\SC[i_4,\ldots,i_L](z_3)\notag\\
                    &=\frac{1}{s_{145\cdots L}}\frac{1}{(-2\pi i)^{L-4}}\int_{(\PC)^{L-4}}\frac{1}{(1\,{-}\,y_{i_4})\yb_L}\prod_{\substack{i,j=5\\i<j}}^L|y_{ij}|^{2s_{ij}}\prod_{k=5}^{L}|y_k|^{2s_{1k}}|1\,{-}\,y_k|^{2s_{4k}}\frac{dy_kd\yb_k}{y_{ki_k}\yb_{k(k-1)}},
                \end{align}
                with $y_1\coloneq0,\,y_4\coloneq1$ and again $s_{145\cdots L}=\sum_{k=4}^Ls_{1k}+\sum_{4\leq i<j\leq L}s_{ij}$.
        \item   The limit $z_3\to1$ yields
                \begin{align}
		              \label{eqn:limit1no23physical}
        			\qquad&\lim_{s_{3\bullet}\to0}\lim_{z_3\to1}|1\,{-}\,z_3|^{-2s_{23}}\SC[i_4,\ldots,i_L](z_3)\notag\\
                    &\qquad=\frac{1}{(-2\pi i)^{L-3}}\int_{(\PC)^{L-3}}\frac{1}{\zb_L}\prod_{\substack{i,j=4\\i<j}}^{L}|z_{ij}|^{2s_{ij}}\prod_{k=4}^{L}|z_k|^{2s_{1k}}|1\,{-}\,z_k|^{2s_{2k}}\frac{dz_kd\zb_k}{z_{ki_k}\zb_{k(k-1)}}\Bigg|_{z_3=1}.
        		\end{align}
    \end{enumerate}
\end{proposition}
\begin{proof}
	Statement \textit{(ii)} follows immediately from taking the limit $s_{3\bullet}\to0$ for the result \textit{(ii)} of \propref{prop:limitsclosed} and splitting up the product as done in the proof of \propref{prop:openphysical}.\\
    To proof statement \textit{(i)}, we proceed as in the proof of \propref{prop:openphysical}, and start by taking the limit $s_{3\bullet}\to0$ of the result \textit{(i)} of \propref{prop:limitsclosed}: it is trivially 0 if one index $i_k$ is equal to 2 and otherwise we obtain
    \begin{align*}
        \lim_{s_{3\bullet}\to0}\lim_{z_3\to0}|z_3|^{-2\smax}&\SC[i_4,\ldots,i_L](z_3)\notag\\
        &=\frac{1}{(-2\pi i)^{L-3}}\int_{(\PC)^{L-3}}\frac{1}{\wb_L}\prod_{\substack{i,j=4\\i<j}}^L|w_{ij}|^{2s_{ij}}\prod_{k=4}^{L}|w_k|^{2s_{1k}}\frac{dw_kd\wb_k}{w_{ki_k}\wb_{k(k-1)}}\,.
    \end{align*}
    Next, we consider the case where $i_k\,{\neq}\,3$ for all indices and substitute $y_k=w_k/w_4,\ k=5,\ldots,L$, i.e.~we rescale all integration variables except for $w_4$, resulting in (again $y_0\coloneq0,\, y_4\coloneq1$)
    \begin{align*}
        \lim_{s_{3\bullet}\to0}&\lim_{z_3\to0}|z_3|^{-2\smax}\SC[i_4,\ldots,i_L](z_3)
        =\underbrace{\frac{1}{-2\pi i}\int_{\PC} \frac{dw_4 d\wb_4}{\wb_4\,{-}\,1}|w_4|^{2\sum_{k=4}^Ls_{1k}+2\sum_{4\leq i<j\leq L}s_{ij}-2}}_{=\frac{1}{\sum_{k=4}^Ls_{1k}+\sum_{4\leq i<j\leq L}s_{ij}}}\\
        &\hspace{5ex}\times\frac{1}{(-2\pi i)^{L-4}}\int_{(\PC)^{L-4}}\frac{1}{(1\,{-}\,y_{i_4})\yb_L}\prod_{\substack{i,j=5\\i<j}}^L|y_{ij}|^{2s_{ij}}\prod_{k=5}^{L}|y_k|^{2s_{1k}}|1\,{-}\,y_k|^{2s_{4k}}\frac{dy_kd\yb_k}{y_{ki_k}\yb_{k(k-1)}}\,,
    \end{align*}
    where for the calculation of the $w_4$-integral, the identity (for assumptions on the parameters, see the cited article) \cite{Mimachi:2018} 
    \begin{equation}\label{eqn:complexBeta}
	   \frac{1}{2\pi i}\int_{\PC}du\,d\ub\,u^a\ub^{a'}(1\,{-}\,u)^{b}(1\,{-}\,\ub)^{b'}=\frac{\Gamma(1+a)\Gamma(1+b)\Gamma(-a'-b'-1)}{\Gamma(-a')\Gamma(-b')\Gamma(2+a+b)}
    \end{equation}
    or \cite[Lemma 4.23]{Brown:2019wna} can be used. This shows the statement \textit{(i)}.
\end{proof}

\begin{rmk}
	The results \eqref{eqn:limit0firstcompsphysical} and \eqref{eqn:limit1no23physical} are precisely $(L{-}1)$- and $L$-point closed-string tree-level amplitudes, respectively, cf.~\eqn{eqn:closedamp}. Furthermore, using the single-valued map defined in \cite{Brown:2019wna}, the closed-string amplitudes here in \propref{prop:limitsclosedphysical} are the single-valued image of the open-string amplitudes found in \propref{prop:openphysical}.
\end{rmk}

\subsection{Momentum configuration for Selberg parameter integrals in regularized limits} \label{sec:StringSelberg}
In \secref{sec:openSelberg} and \secref{sec:closedSelberg}, we have defined the classes of one-parameter Selberg integrals $S(x_3)$ and~$\SC(z_3)$, whose limits for $x_3\to0$ and $x_3\to 1$ (and $z_3\to 0$ and $z_3\to 1$) have been considered in~\secref{sec:LimitopenSelberg} and~\secref{sec:LimitclosedSelberg}. 

Along with taking the regularized limits, the parameters $s_{ij}$ had to be degenerated as well, to arrive at expressions for physical string integrals, as seen in \propref{prop:openphysical} and \propref{prop:limitsclosedphysical}. In terms of momenta $k_i$ of a scattering amplitude, there is a physics interpretation, which supports and explains the role of the points $x_3$ and $z_3$ as interpolating between an $(L{-}1)$-point and an $L$-point string amplitude. We are going to picture the limiting process for parameters interpreted as products of momenta of string scattering amplitudes for the class of open Selberg integrals in the following. 

Closed Selberg integrands are obtained by combining two open Selberg integrands, each of which contains a Parke--Taylor factor as given in the second term of \eqn{eqn:PT}. The Parke--Taylor factor for the holomorphic variables in eqns.~\eqref{eqn:closedamp} and \eqref{eqn:complexSel} implicitly fixes a ordering of insertion points in the particular representation considered. The processes of taking a limit, squeezing, substitution etc.~should be understood as shifts of holomorphic or antiholomorphic positions on a path keeping the ordering set by the respective Parke--Taylor factor intact.
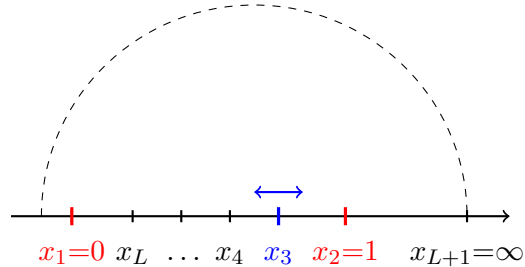
\begin{equation*}
\begin{tikzpicture}
\begin{scope}[scale=0.8]
\draw[thick, ->] (-1.0, 0) -- (7.2, 0);
\foreach \x/\label in {0/{\color{red}$\strut x_1{=}0$}, 1/$\strut x_L$, 1.8/$\strut\ldots$, 2.6/$\strut x_4$, 3.4/{\color{blue}$\strut x_3$}, 4.5/{\color{red}$\strut x_2{=}1$}, 6.5/$\strut x_{L+1}{=}\infty$} {
    \draw[thick] (\x, -0.1) node[below] {\label} -- (\x, 0.1);
}
\draw[very thick, blue] (3.4,-0.15)--(3.4,0.15);
\draw[very thick, red] (0,-0.15)--(0,0.15);
\draw[very thick, red] (4.5,-0.15)--(4.5,0.15);
\draw[<->, thick, blue] (3.0,0.4) -- (3.8,0.4);
\draw[dashed] (-0.5, 0) arc[start angle=180, end angle=0, radius=3.5];
\end{scope}
\end{tikzpicture}
\end{equation*}
Consider the limit $x_3\to 0$: due to the ordering of the iterated integration cycle, performing this limit will squeeze all points $x_4,\ldots x_l$ into a tiny interval close to zero:  
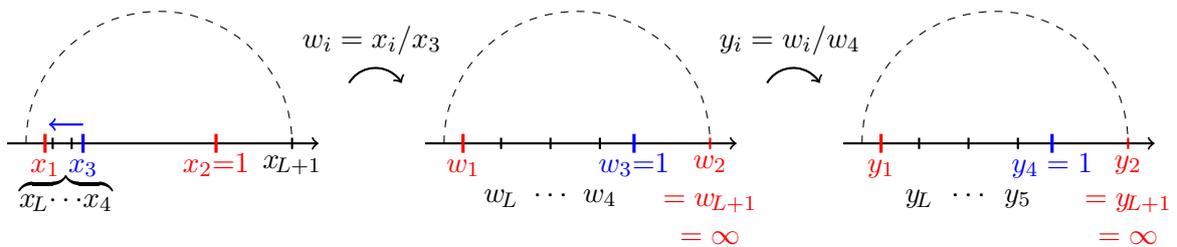
\begin{equation*}
\begin{tikzpicture}
\begin{scope}[scale=0.5]
  \draw[thick, ->] (-1.0, 0) -- (7.2, 0);
  \draw[very thick, red] (0,-0.25) node[below, yshift=7]{$\strut x_1$} -- (0,0.25);
  \draw[thick] (0.2,-0.15) node[below, yshift=-6, xshift=4]{$\strut \overbrace{x_{\!L}\!\cdots\!x_4}$} -- (0.2,0.15);
  \draw[thick] (0.7,-0.15) -- (0.7,0.15);
  \draw[very thick, blue] (1,-0.25) node[below, yshift=7]{$\strut x_3$} -- (1,0.25);
  \draw[<-, thick, blue] (0.1,0.5) -- (1.0,0.5);
  \draw[very thick, red] (4.5,-0.25) node[below, yshift=7]{$\strut x_2{=}1$} -- (4.5,0.25);
  \draw[thick] (6.5,-0.15) node[below, yshift=7]{$\strut x_{\!L+1}$} -- (6.5,0.15);
  \draw[dashed] (-0.5, 0) arc[start angle=180, end angle=0, radius=3.5];
  \draw[->, thick] (8,1.6) arc[start angle=150, end angle=30, radius=0.8cm];
  \draw (8.6,1.6) node[above,yshift=7,align=center] {$w_i=x_i/x_3$};
  \begin{scope}[shift={(11,0)}]
  \draw[thick, ->] (-1.0, 0) -- (7.2, 0);
  \draw[very thick, red] (0,-0.25) node[below, yshift=7]{$\strut w_1$} -- (0,0.25);
  \draw[thick] (1,-0.15) node[below, yshift=-6]{$\strut w_{\!L}$} -- (1,0.15);
  \draw[thick] (2.3,-0.15) node[below, yshift=-6]{$\strut \cdots$}-- (2.3,0.15);
  \draw[thick] (3.6,-0.15) node[below, yshift=-6]{$\strut w_4$}-- (3.6,0.15);
  \draw[very thick, blue] (4.5,-0.25) node[below, yshift=7]{$\strut w_3{=}1$} -- (4.5,0.25);
  \draw[thick,red] (6.5,-0.15) node[below, yshift=7, align=center]{$\strut w_2$\\$=w_{\!L+1}$\\$=\infty$} -- (6.5,0.15);
  \draw[dashed] (-0.5, 0) arc[start angle=180, end angle=0, radius=3.5];
  \end{scope}
  \draw[->, thick] (19,1.6) arc[start angle=150, end angle=30, radius=0.8cm];
  \draw (19.6,1.6) node[above,yshift=7,align=center] {$y_i=w_i/w_4$};
  \begin{scope}[shift={(22,0)}]
  \draw[thick, ->] (-1.0, 0) -- (7.2, 0);
  \draw[very thick, red] (0,-0.25) node[below, yshift=7]{$\strut y_1$} -- (0,0.25);
  \draw[thick] (1,-0.15) node[below, yshift=-6]{$\strut y_{\!L}$} -- (1,0.15);
  \draw[thick] (2.3,-0.15) node[below, yshift=-6]{$\strut \cdots$}-- (2.3,0.15);
  \draw[thick] (3.6,-0.15) node[below, yshift=-6]{$\strut y_5$}-- (3.6,0.15);
  \draw[very thick, blue] (4.5,-0.25) node[below, yshift=7, align=center]{$\strut y_4=1$} -- (4.5,0.25);
  \draw[thick,red] (6.5,-0.15) node[below, yshift=7, align=center]{$\strut y_2$\\$=y_{\!L+1}$\\$=\infty$} -- (6.5,0.15);
  \draw[dashed] (-0.5, 0) arc[start angle=180, end angle=0, radius=3.5];
  \end{scope}
\end{scope}
\end{tikzpicture}
\end{equation*}
The substitution in \eqn{eqn:opensub} has the following implications for the configuration: where $w_3=1$, $w_4=x_4/w_3$ etc.~, the points $w_1=x_1$ and $w_{L+1}=x_{L+1}$ are left inert. The point $w_2=x_2/w_3$ will be sent to infinity and thus merge with $w_{L+1}$. No Mandelstam variables of the form $s_{2i}$ makes an explicit appearance in our representation of the Selberg integral (cf.~\rmkref{rmk:nos12}). Whenever two insertion points merge, the two corresponding momenta should be added. In order to maintain the on-shell property for the new merged point at $\infty$, one can set $k_2=0$.

The resulting integral is not yet the desired $(L{-}1)$-point amplitude: we still have the auxiliary point $w_3=1$ with its associated momentum $k_3$ floating around. This is addressed by performing yet another scaling, however, this time only for all points except for $w_3$: substituting $y_i=w_i/w_4$ for $i\in\lbrace1,4,\ldots,L\rbrace$ as in the proof of \propref{prop:openphysical} will not change the succession of insertion points, but this time merge $y_4$ with $w_3$. Using the same argument as for the above merger, one will have to take the limit $k_3\to 0$, which corresponds to $s_{3\bullet}\to 0$. 

In summary, the process of taking the regularized limit $x_3\to 0$ of an open Selberg integral $S_L(x_3)$ involves two mergers of points and taking the limits $s_{2\bullet}\to 0$ and $s_{3\bullet}\to 0$ in order to reach the $(L{-}1)$-point open-string scattering amplitude. 


In the limit $x_3\to 1$, the situation is simpler: here the auxiliary point merges with the point at $x_2\,{=}\,1$ right away. As with the above mergers, the momenta corresponding to the merged insertions points should be added: $k'_2=k_2+k_3$, which corresponds to $s'_{ij}\coloneq s_{ij}$ if $i\neq2$ and $s'_{2j}\coloneq s_{2j}+s_{3j}$. Maintaining the on-shell property also for the new merged point leads to setting $k_3\to 0$ right away, corresponding to $s_{3\bullet}\to 0$.

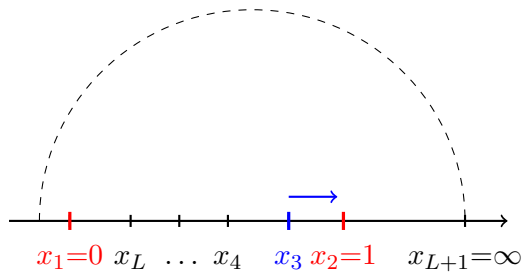
\begin{equation*}
\begin{tikzpicture}
\begin{scope}[scale=0.8]
\draw[thick, ->] (-1.0, 0) -- (7.2, 0);
\foreach \x/\label in {0/{\color{red}$\strut x_1{=}0$}, 1/$\strut x_L$, 1.8/$\strut\ldots$, 2.6/$\strut x_4$, 3.6/{\color{blue}$\strut x_3$}, 4.5/{\color{red}$\strut x_2{=}1$}, 6.5/$\strut x_{L+1}{=}\infty$} {
    \draw[thick] (\x, -0.1) node[below] {\label} -- (\x, 0.1);
}
\draw[very thick, blue] (3.6,-0.15)--(3.6,0.15);
\draw[very thick, red] (0,-0.15)--(0,0.15);
\draw[very thick, red] (4.5,-0.15)--(4.5,0.15);
\draw[->, thick, blue] (3.6,0.4) -- (4.4,0.4);
\draw[dashed] (-0.5, 0) arc[start angle=180, end angle=0, radius=3.5];
\end{scope}
\end{tikzpicture}
\end{equation*}

\section{Recursion for open-string genus-zero amplitudes}\label{sec:openrec}
We will state the general formalism of the recursion for open-string tree-level amplitudes in~\secref{sec:opengeneral} and discuss boundary values in \secref{sec:openboundaryvalues}. Explicit calculations for $L{=}4$ and $L{=}5$ are provided in \secref{sec:open4} and \secref{sec:open5}, interpolating between the three- and four-point and the four- and five-point amplitudes, respectively. The main results from this section are contained in the literature~\cite{Terasoma,Broedel:2013aza,BK,Kaderli}, but we spell them out in detail as they will serve as a guideline for the corresponding construction in the closed-string case.

\subsection{A holomorphic solution of the KZ equation}\label{sec:opengeneral}
In order to find functions that are deformations of open-string amplitudes like in \eqn{eqn:SL2fixed}, we need the possibility to let permutations act on the integrand of the Selberg integrals.
\begin{defn}\label{lem:permSel}
	Let $L\geq 4$, and let $\sigma\in \mathfrak{S}_{L}$ be a permutation of the indices $1,\ldots ,L$ which belongs to the subgroup $\mathfrak{S}^{1,2,3}_{L}$ of permutations that keep the three indices $1,2,3$ fixed.\\
    We denote by $S^\sigma[i_4,\ldots,i_L](x_3)$ the open Selberg integral obtained by letting a permutation $\sigma\in\mathfrak{S}_{L}^{1,2,3}$ act on the labels of the variables $x_{ij}$ and $s_{ij}$ inside the integrand of a Selberg integral $S[i_4,i_5,\ldots,i_L](x_3)$. 
\end{defn}
\begin{rmk}\label{rmk:permSel}
    For the open Selberg integral with permuted indices we can write
    \begin{equation*}
        S^\sigma[i_4,\ldots,i_L](x_3)=S[\sigma(i_{\sigma^{-1}(4)}),\ldots,\sigma(i_{\sigma^{-1}(L)})](x_3),\quad\sigma\in\mathfrak{S}_L^{1,2,3}.
    \end{equation*}
    This follows, as also explained in \cite{Kaderli}, from the fact that $\prod_{\substack{i,j=1\\i<j}}^L|x_{ij}|^{s_{ij}}$ is invariant under permutation, so we only need to consider the action of $\sigma$ on the term $\prod_{k=4}^L\frac{dx_k}{x_{ki_k}}$: for the permutation acting on a single term from this product $\sigma\left[\frac{dx_{\sigma(k)}}{x_{\sigma(k)\sigma(i_k)}}\right]$, we see that the label $\sigma(k)$ will now have the corresponding index $\sigma(i_k)$. Conversely, this means that the label $k=\sigma(\sigma^{-1}(k))$ of the Selberg integral has the corresponding index $\sigma(i_{\sigma^{-1}(k)})$.
\end{rmk}
For $L\geq 4$, $\sigma\in\mathfrak{S}_L^{1,2,3}$ and  $0<x_3<1$, $i\in\{4,\ldots,L+1\}$ and $|s_{ij}|$ small enough\footnote{See the explanation following \eqn{eq:combSelbopen}. A precise statement on the best possible bound on $|s_{ij}|$ can be deduced by combining straightforward generalizations of Prop.~3.5 and Thm.~4.20 of \cite{Brown:2019wna} to the case of one unintegrated variable.}, we define
\begin{subequations}
\begin{align}\label{eqn:basisopen}
    \hF_{L,i}^\sigma(x_3)\!&\coloneq\!\int_{0\leq x_L\leq\cdots\leq x_4\leq x_3}\prod_{\substack{k,m=1\\k<m}}^L\!|x_{km}|^{s_{km}}\sigma\!\left[\prod_{k=i}^L\left(\frac{s_{1k}}{x_k}+\!\sum_{j=k+1}^L\frac{s_{kj}}{x_{kj}}\right)\!\prod_{m=4}^{i-1}\left(\sum_{\substack{n=2\\n\neq 3}}^{m-1}\frac{s_{nm}}{x_{nm}}\right)\right]\!\prod_{k=4}^L\!dx_k\\
    &=(-1)^i\sum_{(i_4,\ldots,i_L)\in I_i}\left(\prod_{k=4}^Ls_{\sigma(k)\sigma(i_k)}\right)S^\sigma[i_4,\ldots,i_L](x_3)\,,\label{eq:combSelbopen}
\end{align}
\end{subequations}
with the set $I_i\coloneq\{(i_4,\ldots,i_L)\,|\,i_\ell\,{\in}\,\{2,4,5,\ldots,\ell{-}1\}\text{ if }\ell{<}i\text{ and }i_\ell\,{\in}\,\{1,\ell{+}1,\ell{+}2,\ldots,L\}\text{ if }\ell{\geq} i\}$ (using \rmkref{rmk:permSel} or \cite[Eq.~(3.18)]{Kaderli} to find the second line) and where $\sigma$ permutes the indices $4,\ldots,L$ of both the positions $x_k$ as well as the Mandelstam variables $s_{mn}$ inside the parenthesis. The notation $\hF_{L,i}^\sigma$ is essentially borrowed from \cite{Broedel:2013aza, Kaderli} and suggests that we are deforming the color-ordered open string amplitudes $F^\sigma$ introduced in \eqn{eqn:F}. 

As explicitly displayed by \eqn{eq:combSelbopen}, the integrals~$\hF_{L,i}^\sigma(x_3)$ are linear combination of Selberg integrals $S[i_4,\dots,i_L](x_3)$, which a priori diverge at $s_{ij}=0$ (c.f.~\defref{defn:classicselb}). However, for any $x_3\!\!\in\,\, ]0,1[$ Selberg integrals can be analytically continued to meromorphic functions on the whole (Mandelstam-variable) complex space, with possible poles at certain hyperplanes with integer coordinates (see \cite{Brown:2019wna, VanZerb}), and it can be checked with the methods of \cite{Brown:2019wna}, Thm.~4.20, that the polar contributions at the origin cancel in each summand of $\hF_{L,i}^\sigma(x_3)$, due to the monomial in the Mandelstam variables which appear as prefactor. Each summand is therefore well-defined (by analytic continuation) in a neighborhood of $s_{ij}=0$, and in this region~$\hF_{L,i}^\sigma$ defines a multivalued function of $x_3\in\mathbb C\smallsetminus \{0,1\}$, single-valued on any simply connected domain $U\subset \mathbb C\smallsetminus\{0,1\}$ containing $]0,1[$.

Each of the $L{-}2$ vectors $\hF_{L,i}(x_3)\coloneqq(\hF_{L,i}^\sigma(x_3))_{\sigma\in \mathfrak{S}^{1,2,3}_{L}}$, with $i\in\{4,\ldots,L{+}1\}$, is an $(L{-}3)!$-dimensional vector. We combine them to obtain an $(L{-}2)!$-dimensional vector $\hF_L(x_3)$, defined as
\begin{equation}\label{eqn:basisBSST}
    \hF_L(x_3)\coloneq\left(\hF_{L,4}(x_3),\hF_{L,5}(x_3),\ldots,\hF_{L,L+1}(x_3)\right)^T\,.
\end{equation}

\begin{rmk}\label{rmk:ind23}
Notice that none of the involved Selberg integrals has an index $i_k$ equal to $3$, and the first subvector $\hF_{L,4}$ has furthermore no index $i_k$ equal to 2, since the second product in the brackets in \eqref{eqn:basisopen} is empty for $i=4$. This will play a role in the recursive mechanism explained later.
\end{rmk}

The following proposition was first stated in \cite{Broedel:2013aza}. A proof was given in \cite{Kaderli}, but it is useful, in view of the proof of our main result in \secref{sec:closedrec}, to recall the strategy.

\begin{theorem}\label{propKZhol}
	Let $V=(\zC[[s_{ij}]])^{(L-2)!}$ and $U$ any simply connected domain containing $]0,1[$. Then the vector $\hF_L:U\to V$ from \eqn{eqn:basisBSST} satisfies the differential equation
    \begin{align}
	   \label{eqn:KZequation}
        \frac{\partial}{\partial x_3}\hF_L(x_3)=\left(\frac{e_0}{x_3}+\frac{e_1}{x_3-1}\right)\hF_L(x_3),
    \end{align}
    with matrices $e_0,e_1\in\Mat_{(L-2)!\times(L-2)!}(\langle s_{ij}\rangle_\zC)$. Moreover, the regularized boundary values
    \begin{equation}\label{eqn:openboundval}
        \hF_L(0)=\lim_{x_3\to0}x_3^{-e_0}\hF_L(x_3),\quad \hF_L(1)=\lim_{x_3\to1}(1-x_3)^{-e_1}\hF_L(x_3)
    \end{equation}
    are related by
    \begin{equation}\label{eqn:openDrinfeld}
        \hF_L(1)=\Phi \,\hF_L(0),    
    \end{equation}
    where $\Phi$ is the Drinfeld associator from \defref{defn:Drinfeld}.
\end{theorem}
\begin{sketch}
Let us recall from \rmkref{rmk:twistedderham} the shorthand notation for the Koba--Nielsen factor
\begin{equation*}
    u\defeq \prod_{k=4}^L x_k^{s_{1k}}(1\,{-}\,x_k)^{s_{2k}}\,\prod_{3\leq i<j\leq L}^L x_{ij}^{s_{ij}}\,,
\end{equation*}
which we now consider, for fixed $x_3\in\mathbb C \smallsetminus\{0,1\}$ and $s_{ij}\in\mathbb C\smallsetminus \mathbb Z$, as the branch of a multivalued function on $\mathrm{Conf}_{L-3}(\mathbb P^1_{\mathbb C}\smallsetminus \{0,1,x_3,\infty\})$ which takes the principal value on $\Delta\defeq 0\leq x_L\leq\cdots\leq x_4\leq x_3$.

For fixed $x_3\in\mathbb C \smallsetminus\{0,1\}$, the integrand of any $\hF_{L,i}^\sigma(x_3)$ can be written (up to a constant depending on $x_3$ and $s_{ij}$) as the product of $u$ with 
\begin{equation}\label{eq:integrandSelbvec}
\hf_{L,i}^\sigma\coloneq \sigma\!\left[\prod_{k=i}^L\left(\frac{s_{1k}}{x_k}+\!\sum_{j=k+1}^L\frac{s_{kj}}{x_{kj}}\right)\!\prod_{m=4}^{i-1}\left(\sum_{\substack{n=2\\n\neq 3}}^{m-1}\frac{s_{nm}}{x_{nm}}\right)\right]\!\prod_{k=4}^L\!dx_k\,,
\end{equation}
a global holomorphic $(L{-}3,0)$-form (i.e.~algebraic) on $\mathrm{Conf}_{L-3}(\mathbb P^1_{\mathbb C}\smallsetminus \{0,1,x_3,\infty\})$.

In order to prove the first statement, a key observation (see \cite{Kaderli}, \S3.2) is that, for any $x_3\in\mathbb C \smallsetminus\{0,1\}$, the $(L{-}2)!$ differential forms $\hf_{L,i}^\sigma$ constitute (representatives for) a $\mathbb C$-vector-space basis of the twisted de Rham cohomology $H^{L-3}_{\mathrm dR}(\mathrm{Conf}_{L-3}(\mathbb P^1_{\mathbb C}\smallsetminus \{0,1,x_3,\infty\}),\nabla_{d\log (u)})$ (see \rmkref{rmk:twistedderham}).

Notice also that one can interchange derivation with respect to $x_3$ with integration because, by the fact that $\mathrm{Re}(s_{ij})>0$ for any $i,j$ and that all $\hf_{L,i}^\sigma$ are combinations of Selberg integrands with all indices $\neq 3$ (see \rmkref{rmk:ind23})), $\frac{\partial}{\partial x_3}(u\hf_{L,i}^\sigma)$ is integrable on the integration domain~$\Delta$ and $u\hf_{L,i}^\sigma$ vanishes when restricted to $x_3=x_4$. One can write $\frac{\partial}{\partial x_3}(u\hf_{L,i}^\sigma)=u\phi_{L,i}^\sigma$, with 
\begin{equation*}
\phi_{L,i}^\sigma\defeq \hf_{L,i}^\sigma\,\sum_{k=4}^{L}\frac{s_{3 k}}{x_{3 k}}+\frac{\partial \hf_{L,i}^\sigma}{\partial x_3}
\end{equation*}
a representative of a cohomology class in $H^{L-3}_{\mathrm dR}(\mathrm{Conf}_{L-3}(\mathbb P^1_{\mathbb C}\smallsetminus \{0,1,x_3,\infty\}),\nabla_{d\log (u)})$.

Combining these two facts it follows that, for every $i=4,\ldots ,L+1$ and $\sigma\in\mathfrak{S}^{1,2,3}_{L}$, and for fixed values of $x_3$ and $s_{ij}$, one has
\begin{equation*}
\phi_{L,i}^\sigma
=\sum_{\substack{i=4,\ldots,L+1\\\sigma\in\mathfrak{S}^{1,2,3}_{L}}}\lambda_{i,\sigma}\hf_{L,i}^\sigma+\nabla_{d\log (u)}(g)\,,
\end{equation*}
where $\lambda_{i,\sigma}\,{\in}\,\mathbb C$, $g\,{=}\,\sum_{j=4}^L g_j\prod_{k\neq j}dx_k$ with $g_j\,{\in}\,\mathbb C[\{x_k,\frac{1}{x_{k}},\frac{1}{x_{k}-1},\frac{1}{x_{k}-x_3}\}_{4\leq k\leq L},\{\frac{1}{x_{km}}\}_{4\leq k<m\leq L}]$, which in turn implies (see \rmkref{rmk:twistedderham}) that
\begin{equation*}
    \frac{\partial}{\partial x_3}(u\hf_{L,i}^\sigma)
    =\sum_{\substack{i=4,\ldots ,L+1\\\sigma\in\mathfrak{S}^{1,2,3}_{L}}}\lambda_{i,\sigma}(u\hf_{L,i}^\sigma)+d(ug)\,.
\end{equation*}
Notice that both the numbers $\lambda_{i,\sigma}$ and the coefficients of the rational functions $g_j$ depend on $x_3$ and~$s_{ij}$.

Integrability of $\frac{\partial}{\partial x_3}(u\hf_{L,i}^\sigma)$ and of all $u\hf_{L,i}^\sigma$ implies integrability of the exact differential $d(ug)$. One can check that this implies that the rational functions~$g_j$ can have at most simple poles, and must be regular along the divisors $x_{i}=x_{i-1}$ whose union forms the boundary $\partial\Delta$ of the integration domain $\Delta$. Combined with the fact that the twist factor~$u$ vanishes along these divisors for $\mathrm{Re}(s_{ij})>0$, it follows by Stokes that 
$\int_{\Delta}d(ug)=\int_{\partial\Delta}ug=0$, which implies that
\begin{equation}
    \frac{\partial}{\partial x_3}\hF_{L,i}^\sigma=\sum_{\substack{i=4,\ldots ,L+1\\\sigma\in\mathfrak{S}^{1,2,3}_{L}}}\lambda_{i,\sigma}\hF_{L,i}^\sigma.
\end{equation}
By careful inspection of the $x_3$-dependence and $s_{ij}$-dependence of the derivatives $\frac{\partial}{\partial x_3}\hf_{L,i}^\sigma$, one can prove that the above linear combination takes indeed the form \eqref{eqn:KZequation}, with matrices $e_0,e_1\,{\in}\,\Mat_{(L-2)!\times(L-2)!}(\langle s_{ij}\rangle_\zC)$ \cite{Kaderli}. 

The second statement of the proposition is then proved by combining the first statement with \thmref{thmKZvectorspace}~\textit{(ii)}. 
\end{sketch}

\subsection{Computing the boundary values}
\label{sec:openboundaryvalues}
In order to compute the regularized boundary values $\hF_L(0)$ and $\hF_L(1)$ we begin by investigating properties of the matrices $e_0$ and $e_1$, which are involved in the regularization procedure:
\begin{proposition}\label{prop:matrices}
	The matrices $e_0$ and $e_1$ from the KZ equation (\ref{eqn:KZequation}) are of the block forms
    \begin{equation*}
		e_0=\left(\begin{array}{cc}
			\smax\,I_{(L-3)!} & a\\
			0 & b
		\end{array}\right),
        \qquad
        e_1=\left(\begin{array}{cc}
			s_{23}\,I_{(L-3)!} & 0\\
			c & d
		\end{array}\right),
	\end{equation*}
    where $a$ and $c$ are matrices of dimension $(L{-}3)!\times(L{-}3)!(L{-}3)$, and $b$ and $c$ are of dimensions $(L{-}3)!(L{-}3)\times(L{-}3)!(L{-}3)$ and $(L{-}3)!(L{-}3)\times(L{-}3)!(L{-}3)$, respectively.
\end{proposition}
\begin{proof}
For $e_1$, the result with derivation can be found in \cite[App.~A.2]{Kaderli}. Note, that in \rcite{Kaderli} a different basis is chosen, which agrees with our basis in the first $(L{-}3)!$ components but differs in the other components. However, this does not change the statement, since one can apply a basis transformation matrix of the block form
\begin{equation*}
    \left(\begin{array}{cc}I_{(L-3)!} & 0\\ Q_1 & Q_2\end{array}\right)
\end{equation*}
to translate between the two bases. Using this to transform the matrix $e_1$, one finds that the structure remains the same, just with different matrices $\tilde{c},\tilde{d}$, which are not specified further anyway.

    For $e_0$ a similar statement was proven in \rcite{Terasoma}, showing that $\smax$ is indeed the maximal eigenvalue. The block matrix form was not explicitly shown there though, which we will prove now:
    To derive  the block form of $e_0$, consider a Selberg integral $S[i_4,\ldots,i_L](x_3)$. For the first block of the vector $\hF_L(x_3)$, i.e.~$\hF_{L,4}(x_3)$, we have for all indices $i_k\neq2,3$ (as can be seen from \eqn{eq:combSelbopen} for $i=4$). Thus, using the formula for a Selberg integral after the substitution $w_k=x_k/x_3$ as in \eqn{eqn:opensub}, direct computation of the derivative yields
    \begin{align*}
        \frac{d}{dx_3}&S[i_4,\ldots,i_L](x_3)=\left(\frac{\smax}{x_3}+\frac{s_{23}}{x_3-1}\right)S[i_4,\ldots,i_L](x_3)\notag\\
        &\quad+\sum_{\ell=4}^Lx_3^{\smax}(1\,{-}\,x_3)^{s_{23}}\int_{0\leq w_L\leq\cdots\leq w_4\leq1}\frac{s_{2\ell}w_\ell}{x_3w_\ell\,{-}\,1}\prod_{\substack{i,j=4\\i<j}}^Lw_{ij}^{s_{ij}}\prod_{k=4}^Lw_k^{s_{1k}}(1\,{-}\,x_3w_k)^{s_{2k}}(1\,{-}\,w_k)^{s_{3k}}\frac{dw_k}{w_{ki_k}}\\
        &=\left(\frac{\smax}{x_3}+\frac{s_{23}}{x_3-1}\right)S[i_4,\ldots,i_L](x_3)+\sum_{\ell=4}^L\frac{s_{2\ell}}{x_3}\int_{0\leq x_L\leq\cdots\leq x_4\leq x_3}\frac{x_\ell}{x_\ell\,{-}\,1}\prod_{\substack{i,j=1\\i<j}}^L|x_{ij}|^{s_{ij}}\prod_{k=4}^L\frac{dx_k}{x_{ki_k}},
    \end{align*}
    where in the second step, we substituted back. Now, partial fractioning can be used to reduce the remaining integrals to the ones without any index equal to 3. Since we have the term $1/(x_\ell\,{-}\,1)$, and all $i_k\neq2$, we will end up with Selberg integrals that have one index equal to 2, so that these terms will contribute to the second block of the matrix $e_0$. This already shows
    \begin{equation*}
		e_0=\left(\begin{array}{cc}
			\smax\,I_{(L-3)!} & a\\
			? & ?
		\end{array}\right),
	\end{equation*}
    for a $(L{-}3)!\times(L{-}3)!(L{-}3)$ dimensional matrix $a$. To find the remaining parts of the matrices, we need to consider the derivative of a Selberg integral where at least one index is equal to 2:
    \begin{align*}
        \frac{d}{dx_3}&S[i_4,\ldots,i_L](x_3)=\frac{d}{dx_3}x_3^{s_{13}}(1\,{-}\,x_3)^{s_{23}}\int_{0\leq x_L\leq\cdots\leq x_4\leq x_3}\prod_{\substack{i,j=1\\i<j}}^L|x_{ij}|^{s_{ij}}\prod_{k=4}^L\frac{dx_k}{x_{ki_k}}\notag\\
        &=\left(\frac{s_{13}}{x_3}+\frac{s_{23}}{x_3-1}\right)S[i_4,\ldots,i_L](x_3)\notag\\
        &\hspace{15ex}+\sum_{k=4}^L s_{3k}x_3^{s_{13}}(1\,{-}\,x_3)^{s_{23}}\int_{0\leq x_L\leq\cdots\leq x_4\leq x_3}\prod_{\substack{i,j=1\\i<j}}^L|x_{ij}|^{s_{ij}}\frac{1}{x_{3k}}\prod_{j=4}^L\frac{dx_j}{x_{ji_j}}.
    \end{align*}
    Again, partial fractioning can be used to reduce the integral to a linear combination of Selberg integrals, producing a factor of $1/x_3$ or $1/(x_3\,{-}\,1)$ in the process (depending on the exact indices $i_k$). Since at least one index was equal to 2, if we result in a factor $1/x_3$ the resulting Selberg integrals after reducing to the basis will have at least one index equal to 2, implying that the third block of $e_0$ must be 0, thus proving the statement \textit{(i)}.
\end{proof}

\subsection*{\boldmath Limit $x_3\,{\to}\,0$} Let us compute the regularized boundary value $\hF_L(0)=\lim_{x_3\to0}x_3^{-e_0}\hF_L(x_3)$:
\begin{proposition}\label{prop:limits0open}
    The vector $\hF_L(0)=(\hF_{L,4}(0),\ldots,\hF_{L,L+1}(0))^T$ is such that
    \begin{align*}
        \hF_{L,4}(0)&=\left(\!\int_{0\leq x_L\leq\ldots\leq x_4\leq 1}\prod_{\substack{i,j=4\\i<j}}^Lx_{ij}^{s_{ij}}\,\sigma\!\left[\prod_{k=4}^L\left(\frac{s_{1k}}{x_k}+\sum_{j=k+1}^L\frac{s_{kj}}{x_{kj}}\right)\right]\prod_{k=4}^L\, x_k^{s_{1k}}(1\,{-}\,x_k)^{s_{3k}}dx_k\!\right)_{\sigma\in \mathfrak{S}^{1,2,3}_{L}}\\
        \hF_{L,i}(0)&=(0,\ldots,0),\quad\text{for }i\in\{5,\ldots,L+1\}.
    \end{align*}
\end{proposition}
\begin{proof}
    Using the results from \propref{prop:matrices} for $e_0$, in the limit $x_3\to0$ the regularized behavior of the matrix exponential is
		\begin{align*}
			\lim_{x_3\to0}x_3^{\smax}x_3^{-e_0}=\left(
			\begin{array}{cc}
				I_{(L-3)!} & A\\
				0 & 0
			\end{array}
			\right)\,,
		\end{align*}
  for a $(L{-}3)!\times (L{-}3)!(L{-}3)$ matrix $A$. Combining this with the decomposition into Selberg integrals \eqref{eq:combSelbopen} and the result \textit{(i)} of \propref{prop:openlimits} the statement of this proposition follows.
\end{proof}

\noindent Taking additionally the auxiliary Mandelstam variables $s_{3\bullet}\to0$ yields the following result:
\begin{proposition}\label{prop:limits0openphys}
    Let $\hF_{L,4}(0)=\left(\hF_{L,4}^{(1)}(0),\hF_{L,4}^{(2)}(0)\right)^T$, where $\hF_{L,4}^{(1)}(0)$ contains the first $(L{-}4)!$ components $\hF_{L,4}^{\sigma}(0)$, where $\sigma\in\mathfrak{S}_L^{1,2,3,4}$ additionally fixes the index 4, while $\hF_{L,4}^{(2)}(0)$ contains the remaining components. Then one finds
    \begin{align*}
        \lim_{s_{3\bullet}\to0}\hF_{L,4}^{(1)}(0)&=\left(\int_{\Delta_{L-4}}\prod_{\substack{i,j=5\\i<j}}^Lw_{ij}^{s_{ij}}\,\sigma\!\!\left[\prod_{k=5}^L\left(\frac{s_{1k}}{w_k}+\sum_{j=k+1}^L\frac{s_{kj}}{w_{kj}}\right)\right]\prod_{k=5}^Lw_k^{s_{1k}}(1\,{-}\,w_k)^{s_{4k}}dw_k\right)_{\sigma\in \mathfrak{S}^{1,2,3,4}_{L}}\\
        \lim_{s_{3\bullet}\to0}\hF_{L,4}^{(2)}(0)&=(0,\ldots,0),
    \end{align*}
    where $\lim_{s_{3\bullet}\to0}$ refers to the limit where all Mandelstam variables $s_{3i}$, for $i=1,\ldots,L$, tend to~0 and $\Delta_{L-4}$ denotes the simplex $\{0\leq w_L\leq\ldots\leq w_5\leq 1\}$.
\end{proposition}
\begin{proof}
    From \propref{prop:limits0open} it follows immediately that the bottom $(L{-}3)!(L{-}1)$ components of $\hF_L(0)$ are zero also in this limit, so we only need to look at the first $(L{-}3)!$ components of the vector which can be written in the notation of \eqn{eqn:basisBSST} as $\hF_{L,4}(0)=(\hF_{L,4}^\sigma(0))_{\sigma\in \mathfrak{S}^{1,2,3}_{L}}$. We will consider the integral after substitution $w_k=x_k/x_3$ as discussed in \eqn{eqn:opensub} and with the regularized limit $x_3\,{\to}\,0$ already taken, as shown in \propref{prop:limits0open}.
    
    In the limit, the components $\hF_{L,4}^\sigma(0)$ have the form
    \begin{align*}
    \lim_{s_{3\bullet}\to0}\hF_{L,4}^\sigma(0)&=\int_{0\leq w_L\leq \ldots\leq w_4\leq 1}\prod_{\substack{i,j=4\\i<j}}^Lw_{ij}^{s_{ij}}\,\sigma\!\!\left[\prod_{k=4}^L\left(\frac{s_{1k}}{w_k}+\sum_{j=k+1}^L\frac{s_{kj}}{w_{kj}}\right)\right]\prod_{k=4}^L w_k^{s_{1k}}dw_k\\
    &=\int_{0\leq w_L\leq\ldots\leq w_4\leq 1}\prod_{\substack{i,j=4\\i<j}}^Lw_{ij}^{s_{ij}}\sigma\!\!\left[\rule{0cm}{0.8cm}\right.\!\left(\frac{s_{14}}{w_4}+\sum_{j=5}^L\frac{s_{4j}}{w_{4j}}\right)\underbrace{\prod_{k=5}^L\left(\frac{s_{1k}}{w_k}+\sum_{j=k+1}^L\frac{s_{jk}}{w_{kj}}\right)}_{\text{independent of }w_4}\left.\rule{0cm}{0.8cm}\!\right]\prod_{k=4}^L w_k^{s_{1k}}dw_k,
    \end{align*}
    where for the first term in the parentheses we have
    \begin{equation*}
        \frac{s_{14}}{w_4}+\sum_{j=5}^L\frac{s_{4j}}{w_{4j}}=\frac{1}{\prod_{k=4}^Lw_k^{s_{1k}}\prod_{\substack{i,j=4\\i<j}}^Lw_{ij}^{s_{ij}}}\frac{d}{dw_4}\left(\prod_{k=4}^Lw_k^{s_{1k}}\prod_{\substack{i,j=4\\i<j}}^Lw_{ij}^{s_{ij}}\right).
    \end{equation*}
    Now we distinguish the case where $\sigma$ leaves the index $4$ invariant and where $\sigma$ includes the index $4$ in its permutation:
    \begin{itemize}
        \item If $\sigma\in\mathfrak{S}_L^{1,2,3,4}$, i.e.~$\sigma(4)=4$, the above structure of the integrand remains intact: the part independent that is indicated to be independent of $x_4$ remains independent after the permutation and the other part $\frac{s_{14}}{w_4}+\sum_{j=5}^L\frac{s_{4j}}{w_{4j}}$ is even invariant under the permutation. Thus, the outermost $w_4$ integration can be performed immediately, resulting in
        \begin{align*}
            \lim_{s_{3\bullet}\to0}\hF_{L,4}^\sigma(0)&=\int_{0\leq w_L\leq \ldots\leq w_4\leq 1}\prod_{k=4}^Ldw_k \,\frac{d}{dw_4}\left(\prod_{k=4}^Lw_k^{s_{1k}}\prod_{\substack{i,j=4\\i<j}}^Lw_{ij}^{s_{ij}}\,\sigma\!\!\left[\prod_{k=5}^L\left(\frac{s_{1k}}{w_k}+\sum_{j=k+1}^L\frac{s_{kj}}{w_{kj}}\right)\right]\right)\\
            &=\int_0^1dw_4 \frac{d}{dw_4} \int_{0\leq w_L\leq\ldots\leq w_5\leq w_4}w_4^{s_{14}}\prod_{k=5}^Ldw_k \,w_k^{s_{1k}}\prod_{\substack{i,j=4\\i<j}}^Lw_{ij}^{s_{ij}}\,\sigma\!\!\left[\prod_{k=5}^L\left(\frac{s_{1k}}{w_k}+\sum_{j=k+1}^L\frac{s_{kj}}{w_{kj}}\right)\right]\\
            &=\int_{0\leq w_L\leq\ldots\leq w_5\leq w_4}w_4^{s_{14}}\prod_{k=5}^Ldw_k \,w_k^{s_{1k}}\prod_{\substack{i,j=4\\i<j}}^Lw_{ij}^{s_{ij}}\,\sigma\!\!\left.\left[\prod_{k=5}^L\left(\frac{s_{1k}}{w_k}+\sum_{j=k+1}^L\frac{s_{kj}}{w_{kj}}\right)\right]\right|_{w_4=0}^{w_4=1}\\
            &=\int_{0\leq w_L\leq\ldots\leq w_5\leq 1}\prod_{k=5}^Ldw_k \,w_k^{s_{1k}}(1\,{-}\,w_k)^{s_{4k}}\prod_{\substack{i,j=5\\i<j}}^Lw_{ij}^{s_{ij}}\,\sigma\!\!\left[\prod_{k=5}^L\left(\frac{s_{1k}}{w_k}+\sum_{j=k+1}^L\frac{s_{kj}}{w_{kj}}\right)\right].
        \end{align*}
        A different representation of this result can also be found by using the expansion \eqref{eq:combSelbopen} into Selberg integrals and then part \textit{(i)} of \propref{prop:openphysical}.
\item If on the other hand the permutation $\sigma$ changes the index $4$ as well, we decompose $\sigma$ into two-cycles, i.e.~$\sigma=(j_1\,j_2)(j_3\,j_4)\cdots(j_{2n-1}\,j_{2n})$ for some integer $n$. The first of these two-cycles which includes the index $4$, say $(4\,j_k)$, will make the previously $w_4$-independent term now $w_4$-dependent but $w_{j_k}$-independent, while the second term will now be
        \begin{equation*}
            \frac{1}{\prod_{k=4}^Lw_k^{s_{1k}}\prod_{\substack{i,j=4\\i<j}}^Lw_{ij}^{s_{ij}}}\frac{d}{dw_{j_k}}\left(\prod_{k=4}^Lw_k^{s_{1k}}\prod_{\substack{i,j=4\\i<j}}^Lw_{ij}^{s_{ij}}\right).
        \end{equation*}
	The next cycle that includes $j_k$ will act also in this way, and so on. This means that after applying $\sigma$ completely, there is an index $\ell\in\{5,\ldots,L\}$, so that we end up with the first term in the brackets being independent of a $w_\ell$, while the second term will be 
            \begin{equation*}
        \frac{s_{1\ell}}{w_\ell}+\sum_{j=5}^L\frac{s_{\ell j}}{w_{\ell j}}=\frac{1}{\prod_{k=4}^Lw_k^{s_{1k}}\prod_{\substack{i,j=4\\i<j}}^Lw_{ij}^{s_{ij}}}\frac{d}{dw_\ell}\left(\prod_{k=4}^Lw_k^{s_{1k}}\prod_{\substack{i,j=4\\i<j}}^Lw_{ij}^{s_{ij}}\right).
    \end{equation*}
    Thus, we can perform the $w_\ell$-integration right away:
        \begin{align*}
            \lim_{s_{3\bullet}\to0}\hF_{L,4}^\sigma(0)&=\int_{0\leq w_L\leq\ldots\leq w_4\leq 1}\prod_{k=4}^Ldw_k\,w_k^{s_{1k}}\prod_{\substack{i,j=4\\i<j}}^Lw_{ij}^{s_{ij}}\\
            &\hspace{5ex}\times\left(\frac{s_{1\ell}}{w_\ell}+\sum_{j=5}^L\frac{s_{\ell j}}{w_{\ell j}}\right)\underbrace{\prod_{k=5}^L\left(\frac{s_{1\sigma(k)}}{w_{\sigma(k)}}+\sum_{j=k+1}^L\frac{s_{\sigma(k)\sigma(j)}}{w_{\sigma(k)\sigma(j)}}\right)}_{\text{independent of }w_\ell}\\
            &=\int_{0\leq w_{\ell-1}\leq\ldots\leq w_{4}\leq1}\prod_{k=4}^{\ell-1}dw_k\, w_k^{s_{1k}} \int_0^{w_{\ell-1}}dw_\ell\\
            &\hspace{5ex}\times\frac{d}{dw_\ell}\left[\int_{0\leq w_L\leq\ldots\leq w_\ell}\prod_{k=\ell+1}^Ldw_k\,w_k^{s_{1k}}\prod_{\substack{i,j=4\\i<j}}^Lw_{ij}^{s_{ij}}\prod_{k=5}^L\left(\frac{s_{1\sigma(k)}}{w_{\sigma(k)}}+\sum_{j=k+1}^L\frac{s_{\sigma(k)\sigma(j)}}{w_{\sigma(k)\sigma(j)}}\right)\right]\\
            &=\int_{0\leq w_{\ell-1}\leq\ldots\leq w_{4}\leq 1}\prod_{k=4}^{\ell-1}dw_k\, w_k^{s_{1k}}\\
            &\hspace{5ex}\times\underbrace{\left[\int_{0\leq w_L\leq\ldots\leq w_\ell}\prod_{k=\ell+1}^Ldw_k\,w_k^{s_{1k}}\prod_{\substack{i,j=4\\i<j}}^Lw_{ij}^{s_{ij}}\prod_{k=5}^L\left(\frac{s_{1\sigma(k)}}{w_{\sigma(k)}}+\sum_{j=k+1}^L\frac{s_{\sigma(k)\sigma(j)}}{w_{\sigma(k)\sigma(j)}}\right)\right]_{w_\ell=0}^{w_\ell=x_{\ell-1}}}_{=0}\\
            &=0.
        \end{align*}
        Thus, all of these components vanish.
    \end{itemize}
    This shows that only those components yield a non-vanishing contribution, where $\sigma$ leaves the index $4$ invariant, meaning that $\sigma\in \mathfrak{S}^{1,2,3,4}_{L}\subset \mathfrak{S}^{1,2,3}_{L}$, and the number of non-vanishing components is $(L{-}4)!$.\phantom{.....}
\end{proof}

\begin{rmk}\label{rmk:openamps0}
    Note, that the resulting integrals from the above Proposition are the open-string integrals $F_{L-1}^\sigma$ from \eqn{eqn:F}, thus yielding a basis of $(L{-}1)$-point open-string amplitudes in this limit.
\end{rmk}

\subsection*{\boldmath Limit $x_3\,{\to}\,1$}
We will now combine \propref{prop:matrices} \textit{(ii)} for the matrix $e_1$ and \propref{prop:openlimits} \textit{(ii)},  in order to calculate the regularized boundary value $\hF_L(1)$:
\begin{proposition}\label{prop:limit1F}
    Let $\hF_L(1)=(\hF_{L,4}(1),\ldots, \hF_{L,L+1}(1))^T$ be the regularized boundary value as defined in \eqn{eqn:openboundval}. Then the first subvector is
    \begin{equation*}
        \hF_{L,4}(1)=\left(\int_{0\leq x_L\leq\ldots\leq x_4\leq 1}\prod_{k=4}^Ldx_k \,\prod_{\substack{i,j=1\\i<j;i,j\neq3}}^Lx_{ij}^{s_{ij}'}\ \sigma\left[\prod_{k=i}^L\left(\frac{s_{1k}}{x_k}+\sum_{j=k+1}^L\frac{s_{kj}}{x_{kj}}\right)\right]\right)_{\sigma\in\mathfrak{S}^{1,2,3}_{L}}\,,
    \end{equation*}
    where $s_{ij}'=s_{ij}$ if $i\neq2$ and $s_{2j}'=s_{2j}+s_{3j}$.
\end{proposition}
\begin{proof}
    Using the result \textit{(ii)} for the matrix $e_1$ from \propref{prop:matrices}, we can calculate the matrix exponential
    \begin{equation*}
        (1-x_3)^{-e_1}=\left(\begin{array}{cc}
            (1-x_3)^{-s_{23}}I_{(L-3)!} & 0 \\
            \ast & \ast
        \end{array}\right).
    \end{equation*}
    Thus, with the definition $\hF_{L}(1)=\lim_{x_3\to1}(1\,{-}\,x_3)^{-e_1}\hF_L(x_3)$, we have $\hF_{L,4}(1)=\lim_{x_3\to1}(1\,{-}\,x_3)^{-s_{23}}\hF_{L,4}(x_3)$ and using additionally result \textit{(ii)} of \propref{prop:openlimits} completes the proof.\phantom{.....}
\end{proof}

As for the other limit, we will again proceed to take the physical limit $s_{3\bullet}\to0$ of this result.
\begin{proposition}\label{prop:openphysical1}
    Taking the limit $s_{3\bullet}\to0$ of $\hF_{L,4}(1)$ yields
    \begin{equation*}
        \lim_{s_{3\bullet}\to0}\hF_{L,4}(1)=\left(\int_{0\leq x_L\leq\ldots\leq x_4\leq 1}\prod_{k=4}^Ldx_k \,\prod_{\substack{i,j=1\\i<j;i,j\neq3}}^Lx_{ij}^{s_{ij}}\ \sigma\left[\prod_{k=i}^L\left(\frac{s_{1k}}{x_k}+\sum_{j=k+1}^L\frac{s_{kj}}{x_{kj}}\right)\right]\right)_{\sigma\in\mathfrak{S}^{1,2,3}_{L}}\,,
    \end{equation*}
    where $\lim_{s_{3\bullet}\to0}$ refers to the limit the Mandelstam variables $s_{3i}$, for $i=1,\ldots,L$, tend to $0$.
\end{proposition}
\begin{proof}
    Starting with the result of \propref{prop:limit1F}, the limit $\lim_{s_{3\bullet}\to0}$ reduces the Mandelstam variables $s_{2j}'$ to $s_{2j}$, and the result follows.
\end{proof}

\begin{rmk}\label{rmk:openlastcomps}
	The first $(L{-}3)!$ components $\lim_{s_{3\bullet}\to0}\hF_{L,4}(1)$ of $\lim_{s_{3\bullet}\to0}\hF_L(1)$ calculated above yield exactly the color-ordered $L$-point open-string amplitudes $F_L^\sigma$ from \eqn{eqn:F}.
    The remaining lower components of $\hF_L(1)$ are thus not needed to span this physically relevant space and can be interpreted as an artifact of the mathematical formalism and are tedious to calculate in general. In the examples in \secref{sec:open4} and \secref{sec:open5} we compute them explicitly, and we notice an interesting structure: they either also have the structure of $L$-point amplitudes (but with different Mandelstam variables) or they are of the schematic form ``$(L$-point amplitude$)-((L{-}1)$-point amplitude$)$''.
\end{rmk}

\subsection*{The open-string amplitude recursion}
The open-string recursion is summarized in the following Theorem, bringing together the results from \secref{sec:opengeneral} and \secref{sec:openboundaryvalues}.
\begin{theorem}\label{thmOpenRec}
    For any $L\geq4$ there exist $(L{-}2)!$-dimensional vectors~$A_L$ and~$B_L$ whose entries are functions of the Mandelstam variables $s_{ij}$, holomorphic in a neighborhood of the origin, such that: 
    \begin{enumerate}[label=(\roman*)]
        \item the first $(L{-}4)!$ components of $A_L$ coincide with the open-string $(L{-}1)$-point integrals $F_{L-1}^\sigma$, $\sigma\in\mathfrak{S}_{L-4}$, from \eqn{eqn:F}, and the remaining components are identically zero;
        \item the first $(L{-}3)!$ components of $B_L$ are the open-string $L$-point integrals~$F_L^\sigma$, $\sigma\in\mathfrak{S}_{L-3}$, from \eqn{eqn:F};
        \item there exist $(L{-}2)!$-dimensional square matrices $e_0$, $e_1$, whose entries are integer linear combinations of Mandelstam variables $s_{ij}$, such that the Taylor expansions of $A_L$ and $B_L$ are related by the Drinfeld associator $\Phi(e_0,e_1)$ via
        \begin{equation*}
            B_L=\Phi(e_0,e_1)A_{L}.
        \end{equation*}
    \end{enumerate}
    In particular, the Taylor expansion of any open-string $L$-point integral~$F_L^\sigma$ can be obtained from the Taylor expansion of the open-string $(L{-}1)$-point integrals $F_{L-1}^\sigma$ and the Drinfeld associator $\Phi(e_0,e_1)$, and therefore the coefficients of these Taylor expansions belong to the ring $\mathcal Z$ of multiple zeta values.
\end{theorem}
\begin{proof}
    We choose $A_{L}=\lim_{s_{3\bullet}\to0}\hF_L(0)$ and $B_L=\lim_{s_{3\bullet}\to0}\hF_L(1)$. Then \propref{prop:limits0openphys} together with \rmkref{rmk:openamps0} tell us that the first $(L{-}4)!$ components of $A_{L}$ are the $(L{-}1)$-point disk amplitudes from \eqn{eqn:F} and correspondingly \propref{prop:openphysical1} together with \rmkref{rmk:openlastcomps} show that the first $(L{-}3)!$ components of $B_L$ yield the $L$-point disk amplitudes from \eqn{eqn:F}. This proves \emph{(i)} and \emph{(ii)}. Finally, taking the physical limit $s_{3\bullet}\to0$ of \eqn{eqn:openDrinfeld} of \thmref{propKZhol} we find the Drinfeld associator relation \emph{(iii)} with the matrices $e_0$ and $e_1$ as in \eqn{eqn:KZequation} after taking the limit $s_{3\bullet}\to0$. The fact that the coefficients of the Taylor expansion of any~$F_L^\sigma$ belong to~$\mathcal Z$ follows inductively by combining the fact that $A_{4}=(1,0)^T$ (see \secref{sec:open4} below) and that, for any~$L$, the coefficients of the Taylor expansion in the variables $s_{ij}$ of the Drinfeld associator $\Phi(e_0,e_1)$ belong to $\mathcal Z$. This last assertion is a consequence of the fact that the coefficients of $\Phi(e_0,e_1)$ as a series in $e_0,e_1$ belong to $\mathcal Z$, combined with the fact that the entries of $e_0,e_1$ are $\mathbb Z$-linear combinations of the variables $s_{ij}$.
\end{proof}

\subsection{From three- to four-point amplitude}\label{sec:open4}
Let us illustrate the above described recursion with the simplest case, which lets us relate the three- and four-point open-string tree-level amplitudes. We fix a pair $(x_1,x_2)=(0,1)$ of real points, use $x_3\in ]0,1[$, and $s_{13},s_{14},s_{23},s_{24},s_{34}\in\mathbb{C}$ with $\Re(s_{14}),\Re(s_{24}),\Re(s_{34})>0$. For $L{=}4$, the vector for the KZ equation is\footnote{Each component of the vector converges in a strictly larger region, e.g.~the first component converges for $\Re(s_{14})\geq 0$, $\Re(s_{34})>-1$ and any $s_{24}$.} (cf.~\eqn{eqn:basisopen}):
\begin{align*}
	\hF_4(x_3)=\mat{\hF_{4,4}(x_3)\\ \hF_{4,5}^{(2)}(x_3)}=\mat{s_{14}\,x_3^{s_{13}}(1\,{-}\,x_3)^{s_{23}}\int_{0}^{x_3}x_4^{s_{14}}(1\,{-}\,x_4)^{s_{24}} (x_3\,{-}\,x_4)^{s_{34}} \frac{dx_4}{x_4}\\
		       s_{24}\,x_3^{s_{13}}(1\,{-}\,x_3)^{s_{23}}\int_{0}^{x_3}x_4^{s_{14}}(1\,{-}\,x_4)^{s_{24}} (x_3\,{-}\,x_4)^{s_{34}} \frac{dx_4}{1-x_4}},
\end{align*}
where we have omitted the superscript for the permutation on the components, since they are just the identity for both components.
Note, that $\hF_4$ can be written in terms of the open Selberg integrals, in agreement with \eqn{eq:combSelbopen}, as $\hF_{4,4}(x_3)=s_{14}S[1](x_3)$ and $\hF_{4,5}(x_3)=-s_{24}S[2](x_3)$. 

For fixed values of the parameters $s_{ij}$, $\hF_4$ can be analytically extended to a multivalued holomorphic function $\hF_4:\mathbb{C}\smallsetminus \{0,1\}\to \mathbb{C}^2$, which (after computing the derivative w.r.t.~$x_3$) fulfills the KZ equation
\begin{equation}\label{eq:230526n1}
	\frac{\pd}{\pd x_3}\hF_4(x_3)=\Bigg[\frac{1}{x_3}\underbrace{\mat{s_{134} & -s_{14} \\ 0 & s_{13}}}_{\eqcolon e_0}+\frac{1}{x_3\,{-}\,1}\underbrace{\mat{s_{23} & 0 \\ -s_{24} & s_{234}}}_{\eqcolon e_1}\Bigg]\hF_4(x_3)\,,
\end{equation}
where we used the shorthand $s_{ijk}\,{=}\,s_{ij}\,{+}\,s_{ik}\,{+}\,s_{jk}$. Notice that $\smax\,{=}\,s_{134}$ is the maximal eigenvalue of $e_0$ as stated in \eqn{eqn:smax}. To realize the above calculation, one can make use of the linear relation
\begin{align}
	0&=\int_0^{x_3}dx_4\frac{d}{dx_4}\big(x_4^{s_{14}}(1\,{-}\,x_4)^{s_{24}}(x_3\,{-}\,x_4)^{s_{34}}\big)\notag\\
	&=\int_0^{x_3}x_4^{s_{14}}(1\,{-}\,x_4)^{s_{24}}(x_3\,{-}\,x_4)^{s_{34}}\left(s_{14} \frac{dx_4}{x_4}+s_{24} \frac{dx_4}{x_4\,{-}\,1}+s_{34}\frac{dx_4}{x_4\,{-}\,x_3}\right)\notag\\
	&=s_{14}S[1](x_3)+s_{24}S[2](x_3)+s_{34}S[3](x_3).\label{eqn:open4relation}
\end{align}
As a function of the Mandelstam variables~$s_{ij}$, for any $x_3\,{\in}\, U$, $\hF_4(x_3)$ has a convergent Taylor expansion at the origin\footnote{It is crucial for this to include the prefactors~$s_{14}$ and~$s_{24}$.}, and can therefore be seen as a function valued in the complex vector space $V\coloneq(\mathbb{C}[[s_{ij}]])^{\oplus 2}$. The matrices~$e_0$ and~$e_1$ from \eqn{eq:230526n1} obviously act on~$V$. The generating series $L(x_3)=\sum_wL_w(x_3)\,w$ and $\Phi=\sum_w\zeta_w\,w$, where~$w$ runs over all words in the endomorphisms~$e_0$ and~$e_1$, also act on~$V$, because they define (for any fixed $x_3\,{\in}\, U$) $2\times 2$ matrices whose entries belong to $\mathbb{C}[[s_{ij}]]$. Combining these observations with \eqn{eq:230526n1}, we conclude that the assumptions of \thmref{thmKZvectorspace} are fulfilled, and therefore that the regularized limits $\hF_4(0)=\lim_{x_3\to 0}x_3^{-e_0}\hF_4(x_3)$ and $\hF_4(1)=\lim_{x_3\to 1}(1\,{-}\,x_3)^{-e_1}\hF_4(x_3)$ are related by the equation 
\begin{equation}\label{eq:lim0to1}
\hF_4(1)=\Phi \cdot \hF_4(0)\,.
\end{equation}

Using the general results discussed in \secref{sec:opengeneral}, let us compute the limit vectors $\hF_4(0)$ and $\hF_4(1)$:
With the matrix exponential of $e_0$, we find as in \propref{prop:limits0open}
\begin{align}\label{eq:230526n2}
\hF_4(0)&=\lim_{x_3\to0}x_3^{-e_0}\hF_4(x_3)=\lim_{x_3\to0}\mat{x_3^{-s_{134}}s_{14}S[1](x_3)-\tfrac{s_{14}}{s_{14}+s_{34}}(x_3^{-s_{13}}\,{-}\,x_3^{-s_{134}})s_{24}S[2](x_3)\\x_3^{-s_{13}}s_{14}S[1](x_3)}\notag\\
&=\mat{\frac{\Gamma(1+s_{14})\Gamma(1+s_{34})}{\Gamma(1+s_{14}+s_{34})}\\0}.
\end{align}
Similarly, for the first component of $\hF_4(1)$ \propref{prop:limit1F} (or direct calculation) yields
\begin{align*}
	\hF_{4,4}(1)&=\lim_{x_3\to 1}(1-x_3)^{-s_{23}}\hF_{4,4}(x_3)\,=\,s_{14}\int_0^1x_4^{s_{14}}(1-x_4)^{s_{24}+s_{34}}\,\frac{dx_4}{x_4}\\
	&=\frac{\Gamma(1+s_{14})\Gamma(1+s_{24}+s_{34})}{\Gamma(1+s_{14}+s_{24}+s_{34})}.
\end{align*}
where in the last equality we made use of a well-known formula for the Euler beta function.

In this simple example it is instructive to also compute the second component $\hF_{4,5}(1)$, which is not predicted through the results of \secref{sec:openboundaryvalues} and which was not computed in the existing literature. Although this component can be regarded as merely an ``artifact'' of the mathematical procedure, we will later see that in the physical limit it is also related to open-string integrals. Using the matrix exponential of $e_1$, we obtain (the lengthy derivation of this expression is spelled out in \appref{app:4pt2ndcomp})
\begin{equation*}
	\hF_{4,5}(1)=\frac{s_{24}}{s_{24}+s_{34}}\left(\frac{\Gamma(1+s_{14})\Gamma(1+s_{24}+s_{34})}{\Gamma(1+s_{14}+s_{24}+s_{34})}-\frac{\Gamma(1+s_{34})\Gamma(1-s_{24}-s_{34})}{\Gamma(1-s_{24})}\right).
\end{equation*}
Thus, the full vector $\hF_4(1)$ is
\begin{equation}
\label{eq:F2}
    \hF_4(1)\,=\,\mat{\frac{\Gamma(1+s_{14})\Gamma(1+s_{24}+s_{34})}{\Gamma(1+s_{14}+s_{24}+s_{34})}\\\frac{s_{24}}{s_{24}+s_{34}}\left(\frac{\Gamma(1+s_{14})\Gamma(1+s_{24}+s_{34})}{\Gamma(1+s_{14}+s_{24}+s_{34})}-\frac{\Gamma(1+s_{34})\Gamma(1-s_{24}-s_{34})}{\Gamma(1-s_{24})}\right)}.
\end{equation}

\subsection*{Relating the limits}
Combining eqs.~\eqref{eq:lim0to1},~\eqref{eq:230526n2} and~\eqref{eq:F2}, we can relate the two limits using the Drinfeld associator $\Phi=\Phi(e_0,e_1)$, to find
\begin{equation*}\label{eqn:4ptopenlimitsrel}
\mat{\frac{\Gamma(1+s_{14})\Gamma(1+s_{24}+s_{34})}{\Gamma(1+s_{14}+s_{24}+s_{34})}\\ \frac{s_{24}}{s_{24}+s_{34}}\left(\frac{\Gamma(1+s_{14})\Gamma(1+s_{24}+s_{34})}{\Gamma(1+s_{14}+s_{24}+s_{34})}-\frac{\Gamma(1+s_{34})\Gamma(1-s_{24}-s_{34})}{\Gamma(1-s_{24})}\right)}\,=\,\Phi\cdot \mat{\frac{\Gamma(1+s_{14})\Gamma(1+s_{34})}{\Gamma(1+s_{14}+s_{34})}\\0}\,.
\end{equation*}

In particular, taking the physical limit $s_{3\bullet}\to 0$, where $\bullet=1,2,4$, we find the above result reducing to the string amplitude relation
\begin{equation*}
	\mat{\frac{\Gamma(1+s_{14})\Gamma(1+s_{24})}{\Gamma(1+s_{14}+s_{24})}\\ \frac{\Gamma(1+s_{14})\Gamma(1+s_{24})}{\Gamma(1+s_{14}+s_{24})}-1}\,=\,\Phi\big|_{s_{3\bullet}=0}\cdot \mat{1\\0}\,,
\end{equation*}
where the first component of the vector on the l.h.s.~is the four-point open-string amplitude $F_{4}$ (cf.~\eqn{eqn:F}) and the first component of the vector on the r.h.s.~the three-point amplitude $F_3$. The second component of the l.h.s.~can be understood as a four-point amplitude (ratio of Gamma functions) minus a three-point amplitude (given by the 1).

\subsection{From four- to five-point amplitude}\label{sec:open5}
Let $L\,{=}\,5$, then the vector $\hF_5(x_3)$ from 
\eqn{eqn:basisBSST} fulfills the KZ equation \eqn{eqn:KZequation} with the two matrices
\begin{subequations}\label{eqn:5e}
{\small
\begin{align}
	&e_0=\left(\begin{array}{cccccc}
		s_{1345} & 0 & -s_{14}-s_{45} & -s_{15} & -s_{15} & s_{15}\\
		0 & s_{1345} & -s_{14} & -s_{15}-s_{45} & s_{14} & -s_{14}\\
		0 & 0 & s_{135} & 0 & -s_{15} & 0\\
		0 & 0 & 0 & s_{134} & 0 & -s_{14}\\
		0 & 0 & 0 & 0 & s_{13} & 0\\
		0 & 0 & 0 & 0 & 0 & s_{13}
	\end{array}\right),
    \\
    &e_1=\left(\begin{array}{cccccc}
		s_{23} & 0 & 0 & 0 & 0 & 0\\
		0 & s_{23} & 0 & 0 & 0 & 0\\
		-s_{24} & 0 & s_{234} & 0 & 0 & 0\\
		0 & -s_{25} & 0 & s_{235} & 0 & 0\\
		-s_{24} & s_{24} & -s_{25}-s_{45} & -s_{24} & s_{2345} & 0\\
		s_{25} & -s_{25} & -s_{25} & -s_{24}-s_{45} & 0 & s_{2345}
	\end{array}\right),
\end{align}
}%
\end{subequations}
where we used the shorthands $s_{ijk}\,{=}\,s_{ij}\,{+}\,s_{ik}\,{+}\,s_{jk}$ and $s_{ijk\ell}\,{=}\,s_{ij}\,{+}\,s_{ik}\,{+}\,s_{i\ell}\,{+}\,s_{jk\ell}$.

Applying the results from \secref{sec:opengeneral} and \secref{sec:openboundaryvalues} yields a recursion bringing us from the four-point amplitude to a basis of five-point amplitudes via the Drinfeld associator, namely
\begin{align*}
	\Phi(\Tilde{e}_0,\Tilde{e}_1)\left(\!\!\begin{array}{c}F_4^{\mathrm{id}}\\0\\0\\0\\0\\0\end{array}\!\!\right)
    &=\left(1+\zeta_2[\Tilde{e}_0,\Tilde{e}_1]+\zeta_3[\Tilde{e}_0+\Tilde{e}_1,[\Tilde{e}_0,\Tilde{e}_1]]+\ldots\right)\left(\!\!\begin{array}{c}\frac{\Gamma (1+s_{15}) \Gamma (1+s_{45})}{\Gamma (1+s_{15}+s_{45})}\\0\\0\\0\\0\\0\end{array}\!\!\right)=\left(\!\!\begin{array}{c}F_{5}^\mathrm{id}\\ F_{5}^{(45)}\\ \ast \\ \ast \\ \ast \\ \ast\end{array}\!\!\right).
\end{align*}
where we have defined the matrices in the physical limit by $\Tilde{e}_i=\lim_{s_{3\bullet}\to0}e_i$ and $F_N^\sigma$ are the genus-zero open-string amplitudes defined in \eqn{eqn:F}. The full calculation of this result, which also includes the lower components of the vector on the r.h.s.~of the above equation, is spelled out in \appref{app:5open}.


\section{Recursion for closed-string genus-zero amplitudes}\label{sec:closedrec}

In this section we will prove our main result, a recursive relation between closed string amplitudes, by adapting to this setting the strategy used for open strings in \cite{Broedel:2013aza} and explained in~\secref{sec:openrec}. We will introduce an appropriate vector of closed Selberg integrals satisfying the KZ equation in \secref{sec:closedgeneral}. The recursion between closed string integrals is then obtained in~\secref{ssec:complexlimits} by studying the regularized limits of the KZ solution and by applying results from~\secref{ssec:svmpls}. We will illustrate the method by working out in more details the special cases $L\,{=}\,4$ and $L\,{=}\,5$ in~\secref{sec:closed4} and~\secref{sec:closed5}, respectively.

\subsection{A single-valued solution of the KZ equation}\label{sec:closedgeneral}
\begin{defn}\label{lem:permSelC}
    For $L\geq4$ we denote by $\SC^\sigma[i_4,\ldots,i_L](z_3)$ the closed Selberg integral obtained by letting a permutation $\sigma\in\mathfrak{S}_L^{1,2,3}$ act on the labels of the holomorphic variables $z_{ij}$ and $s_{ij}$ inside the integrand of a closed Selberg integral $\SC[i_4,i_5,\ldots,i_L](z_3)$. 
\end{defn}
\begin{rmk}
    As described in \rmkref{rmk:permSel}, we find the identity
    \begin{equation*}
        \SC^\sigma[i_4,\ldots,i_L](z_3)=\SC[\sigma(i_{\sigma^{-1}(4)}),\ldots,\sigma(i_{\sigma^{-1}(L)})](z_3),\quad\sigma\in\mathfrak{S}_L^{1,2,3},
    \end{equation*}
    for the closed Selberg integral with permuted indices.
\end{rmk}

For $z_3\in\zC\smallsetminus\{0,1\}$, $i\in\{4,\ldots,L+1\}$ and $|s_{ij}|$ small enough\footnote{As in the definition of $\hF_{L,i}^\sigma(x_3)$, a precise statement on the best possible bound on $|s_{ij}|$ can be deduced by combining straightforward generalizations of Prop.~3.6 and Thm.~4.24 of \cite{Brown:2019wna} to the case of one unintegrated variable.} we define
\begin{subequations}
    \begin{align}\label{eqn:basisclosed}
    \hcF_{L,i}^\sigma(z_3)&\coloneqq\frac{1}{(-2\pi i)^{L-3}}\int_{(\PC)^{L-3}}\prod_{\substack{k,m=1\\k<m}}^L\!|z_{km}|^{2s_{km}}\notag\\
    &\qquad\qquad\qquad\times\sigma\!\left[\prod_{k=i}^L\left(\frac{s_{1k}}{z_k}+\!\sum_{j=k+1}^L\frac{s_{kj}}{z_{kj}}\right)\!\prod_{m=4}^{i-1}\left(\sum_{\substack{n=2\\n\neq 3}}^{m-1}\frac{s_{nm}}{z_{nm}}\right)\right]\frac{\zb_3}{\zb_L}\prod_{k=4}^{L}\frac{dz_kd\zb_k}{\zb_{k(k-1)}}\\
    &=(-1)^i\sum_{(i_4,\ldots,i_L)\in I_i}\left(\prod_{k=4}^Ls_{\sigma(k)\sigma(i_k)}\right)\SC^\sigma[i_{4},\ldots,i_{L}](z_3)\,,\label{eq:combSelbclosed}
\end{align}
\end{subequations}
with the index set $I_i$ as defined below \eqn{eq:combSelbopen} and where the permutation $\sigma\in\mathfrak{S}_L^{1,2,3}$ acts on the indices $4,\ldots,L$ of the positions $z_k$ and Mandelstam variables $s_{mn}$ inside the parenthesis. Similarly to the integrals $\hF^{\sigma}_{L,i}(x_3)$ defined in \secref{sec:opengeneral}, the integrals $\hcF_{L,i}^\sigma(z_3)$ are defined for small~$|s_{ij}|$ by analytically continuing the closed Selberg integrals which appear in \eqref{eq:combSelbclosed}. We refer the reader to the discussion which follows \eqn{eq:combSelbopen}. The functions $\hcF_{L,i}^\sigma$ are real analytic single-valued functions of $z_3$, as can be seen from the decomposition into closed Selberg integrals, which themselves are real analytic single-valued functions. From a physics point of view the functions $\hcF^\sigma_{L,i}$ are essentially deformations of the closed-string integrals $\cF^\sigma$ from \eqn{eqn:cF} with the deformation parameter $z_3$.

As in the open string case, each of the $L{-}2$ vectors $\hcF_{L,i}(z_3)\coloneqq(\hcF_{L,i}^\sigma(z_3))_{\sigma\in \mathfrak{S}^{1,2,3}_{L}}$, with $i\in\{4,\ldots,L{+}1\}$, has $(L{-}3)!$ components. We combine these vectors into an $(L{-}2)!$-dimensional vector $\hcF_L(z_3)$, defined as
\begin{equation}\label{eqn:basisBSSTclosed}
    \hcF_L(z_3)\coloneq\left(\hcF_{L,4}(z_3),\hcF_{L,5}(z_3),\ldots,\hcF_{L,L+1}(z_3)\right)^T\,.
\end{equation}

\begin{rmk}
Since -- as explained in \rmkref{rmk:svSelberg} -- the closed Selberg integrals are the image of the open Selberg integrals under the sv-map, it follows from expansions \eqref{eq:combSelbopen} and \eqref{eq:combSelbclosed} that $\hcF(z_3)$ is the image of $\hF(x_3)$ under the sv-map, i.e.~$\mathrm{sv}\big(\hF(z_3)\big)=\hcF(z_3)$.
\end{rmk}

\begin{theorem}\label{prop:equationclosed}
    Let $V=(\zC[[s_{ij}]])^{(L-2)!}$. Then the vector-valued function $\hcF_L:\mathbb{C}\smallsetminus\{0,1\}\to V$ from \eqn{eqn:basisBSSTclosed} satisfies the differential equation
    \begin{align}
	   \label{eqn:KZequationclosed}
        \frac{\partial}{\partial z_3}\hcF_L(z_3)=\left(\frac{e_0}{z_3}+\frac{e_1}{z_3-1}\right)\hcF_L(z_3)
    \end{align}
    with the same matrices $e_0,e_1\in\Mat_{(L-2)!\times(L-2)!}(\langle s_{ij}\rangle_\zC)$ of \thmref{propKZhol}. Moreover, the boundary values
    \begin{equation}\label{eqn:closedboundval}
        \hcF_L(0)=\lim_{z_3\to0}|z_3|^{-2e_0}\hcF_L(z_3),\quad \hcF_L(1)=\lim_{z_3\to1}|1\,{-}\,z_3|^{-2e_1}\hcF_L(z_3)
    \end{equation}
    are related by
    \begin{equation}\label{eqn:closedDeligne}
        \hcF_L(1)=\Phi^\mathrm{sv} \,\hcF_L(0),    
    \end{equation}
    where $\Phi^\mathrm{sv}$ is the Deligne associator from \defref{defn:Deligne}.
\end{theorem}
\begin{proof}
Let us prove the first statement. As discussed below \eqn{eq:combSelbclosed}, each component of the vector $\hcF_L(z_3)$ defines, for fixed $z_3$, a holomorphic function in the Mandelstam variables at the origin $s_{ij}=0$. However, the closed Selberg integrals appearing as summands in the expression~\eqref{eq:combSelbclosed} may diverge at the origin; in order to ensure their absolute convergence, we suppose that $\mathrm{Re}(s_{ij})>0$ and that $|s_{ij}|$ is small enough. If we show the result in this region, it will hold also in a neighborhood of the origin, by holomorphicity in the Mandelstam variables, and will therefore hold when viewing the components of $\hcF_L(z_3)$ as power series in $s_{ij}$. 

Under the above assumptions on $s_{ij}$, we can interchange derivation and integration. Notice that the integrand of each component~$\hcF_{L,i}^\sigma$ is given (up to a constant depending on $x_3$ and $s_{ij}$) by the product $|u|^2\hf_{L,i}^\sigma \overline{h}$ between the (holomorphic) integrand $u\hf_{L,i}^\sigma$ of $\hF_{L,i}^\sigma$, with $u$ and $\hf_{L,i}^\sigma$ defined in eqs.~\eqref{eq:twistKNfactor} and \eqref{eq:integrandSelbvec}, respectively, and the (antiholomorphic) factor $\overline{u}\overline{h}$, with
\begin{equation*}
    h\defeq \frac{1}{z_L}\prod_{k=4}^{L}\frac{dz_k}{z_{k(k-1)}}\,.
\end{equation*}
Obviously $\frac{\partial}{\partial z_3}(|u|^2\hf_{L,i}^\sigma \overline{h})=\big(\frac{\partial}{\partial z_3}(u\hf_{L,i}^\sigma) \big)\,\overline{u}\overline{h}$. By the same reasoning explained in the proof of \thmref{propKZhol}, it follows that each $\frac{\partial}{\partial z_3}\hcF_{L,i}^\sigma$ can be written as a linear combination of the functions $\hcF_{L,i}^\sigma$, the coefficients being predicted precisely by \eqn{eqn:KZequationclosed}, plus the integral of $d(|u|^2 g\overline{h})$, with~$g$ as in the proof of \thmref{propKZhol}. In order to prove that $\hcF_L:\mathbb{C}\smallsetminus\{0,1\}\to V$ satisfies the KZ equation, with the same matrices $e_0,e_1\in\Mat_{(L-2)!\times(L-2)!}(\langle s_{ij}\rangle_\zC)$ of \thmref{propKZhol}, it is therefore sufficient to prove that $\int_{(\PC)^{L-3}}d(|u|^2 g\overline{h})=0$.

Recall that $g=\sum_{j=4}^L g_j\prod_{k\neq j}z_k$, and we have already seen in the proof of \thmref{propKZhol} that the functions $g_j$ can have at most simple poles, and are regular at the divisors $z_{k1},z_{k2},z_{k(k-1)}=0$, $k=4,\ldots ,L$. One can check that this implies that the exact differential $d(|u|^2 g\overline{h})$ is integrable on $(\PC)^{L-3}$. Notice that $d(|u|^2 g\overline{h})=\sum_{\ell=4}^Ld_{z_\ell}(|u|^2 g\overline{h})$, where $d_{z_\ell}$ is the differential w.r.t. the variable $z_\ell$. It is therefore enough to show that $\int_{(\PC)^{L-3}}d_{z_\ell}(|u|^2 g\overline{h})=0$ for every~$\ell=4,\ldots ,L$. 

Notice that 
\begin{equation*}
\int_{(\PC)^{L-3}}d_{z_\ell}(|u|^2 g\overline{h})=\int_{(\PC)^{L-4}}\frac{\zb_{(\ell+1)\ell}}{\zb_L}\prod_{\substack{k=4\\k\neq \ell}}^{L}\frac{dz_kd\zb_k}{\zb_{k(k-1)}}\int_{\PC}d_{z_\ell}\frac{|u|^2g_\ell\, d\zb_\ell}{\zb_{(\ell+1)\ell}\,\zb_{\ell(\ell-1)}}\,.
\end{equation*}
Let us introduce the notation
\begin{equation*}
    \omega_\ell\defeq \frac{|u|^2g_\ell\, d\zb_\ell}{\zb_{(\ell+1)\ell}\,\zb_{\ell(\ell-1)}}\,,
\end{equation*}
and let us prove that $\int_{\PC}d_{z_\ell}\omega_\ell=0$, which implies the desired statement.

Let $B^\varepsilon_{z}$ be the ball with radius $\varepsilon>0$ around $z$ and define
\begin{align*}
    U_\varepsilon\defeq\PC\smallsetminus\bigcup_{\substack{k=1\\k\neq\ell}}^{L+1}B^\varepsilon_{z_k},
\end{align*}
with $z_1=0,\ z_2=1,\ z_{L+1}=\infty$. The differential form $d_{z_\ell}\omega_\ell$ is smooth over $U_\varepsilon$, hence one can apply Stokes theorem and obtain
\begin{equation*}
    \int_{U_\varepsilon}d_{z_\ell}\omega_\ell=\int_{\partial U_\varepsilon}\omega_\ell.
\end{equation*}
For the computation of the r.h.s.~of this equation we need to sum the antiholomorphic residues of~$\omega_\ell$ at the points $\zb_\ell=\zb_k,\ k=1,\ldots,L+1,\ k\neq\ell$ (this can be seen by passing to polar coordinates centered at these points and integrating term-by-term the Taylor expansion, see \cite{Schnetz:2013hqa} or \cite{VanZerb} for details). The residues are all zero except at the points $\zb_{\ell+1}$ and $\zb_{\ell-1}$. At these points $\omega_\ell$ has the same residue but with opposite signs, so that the sum of all residues vanishes, i.e.
\begin{equation*}
    \int_{U_\varepsilon}d_{z_\ell}\omega_\ell=\int_{\partial U_\varepsilon}\omega_\ell=0.
\end{equation*}
Since $\int_{\mathbb P^1_{\mathbb C}}d_{z_\ell}\omega_\ell$ is absolutely convergent, $\int_{\mathbb P^1_{\mathbb C}}d\omega=\lim_{\varepsilon\to 0}\int_{U_\varepsilon}d\omega=0$.

Let us now prove the second statement of the proposition. To do so, we want to apply \thmref{thmKZvectorspaceSV} \emph{(ii)}. Let us us check that its assumptions are satisfied. Notice that $V=\Mat_{(L-2)!\times 1}(A)$, and $A=\zC[[s_{ij}]]$ is a complete graded $\mathbb C$-algebra, as required. The fact that $e_0$ and $e_1$ are elements of $\End(V)$ is obvious. To conclude, we just need to check that each entry of~$\hcF_L$, as a function of $z_3$, belongs to $A\hat\otimes_{\zC} \cH\overline{\cH}=\cH\overline{\cH}[[s_{ij}]]$, where $\cH$ is the $\mathbb C$-algebra generated by multiple polylogarithms. This follows from results of \cite{VanZerb}. More precisely, one should first notice that the entries of~$\hcF_L$ are special cases of the ``correlation functions'' $\mathcal G_N$ of eq.~(1.4) in \cite{VanZerb}. Therefore, by applying Thm.~5.1 of \cite{VanZerb} we deduce that they are linear combinations with coefficients in $\mathbb C[[s_{ij}]]$ of products $F_{\Delta_1}(z_3)\overline{F_{\Delta_2}(z_3)}$ of holomorphic and antiholomorphic ``Aomoto-Gelfand hypergeometric functions'' $F_{\Delta}$ (see eq.~(1.6) of \cite{VanZerb}), and these combinations are holomorphic at the origin $s_{ij}=0$ because so are the entries of~$\hcF_L$. As explained in \S 4.4 and \S 5.3 of \cite{VanZerb}, it follows from the methods\footnote{In fact, one can deduce this result also by a minor adaptation of the arguments developed by Terasoma in \cite{Terasoma} and presented in this article: up to multiplication by an element of $\zC[s_{ij}]$, every Aomoto-Gelfand hypergeometric function can be written as a linear combination with coefficients in $\langle s_{ij}\rangle_{\mathbb C}$ of admissible Selberg integrals which is holomorphic at the origin $s_{ij}=0$. This can always be seen as an entry of a vector which satisfies the KZ equation for some matrices $e_0, e_1$ with entries in $\langle s_{ij}\rangle_\zC$ \cite{Kaderli}. One concludes from \thmref{thmKZvectorspace} \textit{(i)} that this vector is equal to $L(z)\cdot C$, with $C$ a vector with entries in $\zC[[s_{ij}]]$, whence the claimed statement.} of \cite{Brown:0606} and \cite{Brown:2019wna} that their Taylor expansion at $s_{ij}=0$ must have coefficients in $\cH\overline{\cH}$, as wanted (and, in fact, even in the $\mathbb C$-algebra $\mathcal H^{\rm sv}$ generated by single-valued multiple polylogarithms, either by the argument in \rmkref{rmk:svMPLscplxSelberg} or by the argument of \cite[\S 5.3]{VanZerb}).
\end{proof}

\subsection{Computing the boundary values}\label{ssec:complexlimits} Having shown that the vector $\hcF_L$ satisfies the KZ equation (with the same matrices $e_0$ and $e_1$ as in the open-string case in \secref{sec:openrec}), we will proceed to calculate the (regularized) boundary values $\hcF_L(0)$ and $\hcF_L(1)$, which are finite by \thmref{thmKZvectorspaceSV}. Similarly to the open string case, we will recover closed-string amplitudes by taking the physical limit $s_{3\bullet}\to0$ of these boundary values.

\subsection*{\boldmath Limit $z_3\,{\to}\,0$} Let us compute the vector $\hcF_L(0)\coloneq\lim_{z_3\to 0}|z_3|^{-2e_0}\hcF_L(z_3)$:
\begin{proposition}\label{prop:limits0closed}
    The limit vector $\hcF_L(0)=(\hcF_{L,4}(0),\ldots,\hcF_{L,L+1}(0))^T$ is such that
    \begin{align*}
        &\hcF_{L,4}(0)=\frac{1}{(-2\pi i)^{L-3}}\notag\\
        &\times\!\!\left(\!\int_{(\PC)^{L-3}}\! \frac{1}{\wb_L}\prod_{\substack{i,j=4\\i<j}}^L|w_{ij}|^{2s_{ij}}\prod_{k=4}^L|w_k|^{2s_{1k}}|1\,{-}\,w_k|^{2s_{3k}}\frac{dw_kd\wb_k}{\wb_{k(k-1)}}  \sigma\!\left[\prod_{k=4}^L\!\left(\frac{s_{1k}}{w_k}{+}\!\!\sum_{j=k+1}^L\frac{s_{kj}}{w_{kj}}\right)\!\right]\!\right)_{\!\!\sigma\in \mathfrak{S}^{1,2,3}_{L}}\\
        &\hcF_{L,i}(0)=(0,\ldots,0),\quad\text{for }i\in\{5,\ldots,L+1\},
    \end{align*}
    with $\wb_3\coloneq1$.
\end{proposition}
\begin{proof}
    As argued in the proof of \thmref{prop:equationclosed} the derivative $\partial/\partial z_3$ acting on $\hcF_L(z_3)$ will yield the same relations as found in the open-string case for the vector $\hF_L(x_3)$ when acted upon with the derivative $\partial/\partial x_3$. Thus, the matrix $e_0$ in the KZ equation \eqref{eqn:KZequationclosed} is the same as in the open-string case and therefore of the block form
    \begin{equation*}
        e_0=\left(\begin{array}{cc}\smax\,I_{(L-3)!} & A\\ 0 & B \end{array}\right).
    \end{equation*}
    Using this, we find for the matrix exponential the regulated the scaling relation
    \begin{align*}
		\lim_{z_3\to0}|z_3|^{2\smax}|z_3|^{-2e_0}=\left(
		\begin{array}{cc}
			I_{(L-3)!} & A\\
			0 & 0
		\end{array}
		\right)\,.
	\end{align*}
    Combining this with the result \textit{(i)} of \propref{prop:limitsclosed} concludes the proof.
\end{proof}

The physical limit of the auxiliary Mandelstam variables $s_{3\bullet}\to0$ readily follows:
\begin{proposition}\label{prop:limits0closedphys}
    Let $\hcF_{L,4}(0)=(\hcF_{L,4}^{(1)}(0),\hcF_{L,4}^{(2)}(0))^T$, where $\hcF_{L,4}^{(1)}(0)$ contains the first $(L\,{-}\,4)!$ components $\hcF_{L,4}^{\sigma}(0)$ where $\sigma\in\mathfrak{S}_L^{1,2,3,4}$ fixes the index 4, while $\hcF_{L,4}^{(2)}(0)$ contains the remaining components. Then
    \begin{align*}
        &\lim_{s_{3\bullet}\to0}\hcF_{L,4}^{(1)}(0)=\Bigg(\frac{1}{(-2\pi i)^{L-4}}\int_{(\PC)^{L-4}}\frac{1}{\wb_L}\prod_{\substack{i,j=5\\i<j}}^L|w_{ij}|^{2s_{ij}}\,\notag\\
        &\hspace{20ex}\times\sigma\!\left[\prod_{k=5}^L\left(\frac{s_{1k}}{w_k}+\sum_{j=k+1}^L\frac{s_{kj}}{w_{kj}}\right)\right]\prod_{k=5}^{L}|w_k|^{2s_{1k}}|1\,{-}\,w_k|^{2s_{4k}}\frac{dw_kd\wb_k}{\wb_{k(k-1)}}\Bigg)_{\sigma\in \mathfrak{S}^{1,2,3,4}_{L}}\\
        &\lim_{s_{3\bullet}\to0}\hcF_{L,4}^{(2)}(0)=(0,\ldots,0),
    \end{align*}
    with $\wb_4\coloneq1$ and where $\lim_{s_{3\bullet}\to0}$ refers to the limit where all Mandelstam variables $s_{3i}$, for $i=1,\ldots,L$, tend to~0.
\end{proposition}
\begin{proof}
We proceed as in the open-string case in \propref{prop:limits0openphys}: first, we note that due to the results of \propref{prop:limits0closed} it directly follows that the bottom $(L{-}3)!(L{-}1)$ components of $\hcF_L(0)$ vanish also here. Thus let us consider the first $(L{-}3)!$ components $\hcF_{L,4}(0)=(\hcF_{L,4}^\sigma(0))_{\sigma\in \mathfrak{S}^{1,2,3}_{L}}$. After substitution $w_k=z_k/z_3$ (c.f.~\eqn{eqn:closedsub}) and using the result of \propref{prop:limits0closed}, we can write the integral as
    \begin{align*}
    &\lim_{s_{3\bullet}\to0}\hcF_{L,4}^\sigma(0)=\bigg(\!\!\frac{-1}{2\pi i}\!\!\bigg)^{L{-}3}\!\!\!\int_{(\PC)^{L-3}}\frac{1}{\wb_L}\prod_{\substack{i,j=4\\i<j}}^L|w_{ij}|^{2s_{ij}}\,\sigma\!\!\left[\prod_{k=4}^L\left(\frac{s_{1k}}{w_k}+\sum_{j=k+1}^L\frac{s_{kj}}{w_{kj}}\right)\right]\prod_{k=4}^L |w_k|^{2s_{1k}}\frac{dw_kd\wb_k}{\wb_{k(k-1)}}\\
    &=\bigg(\!\!\frac{-1}{2\pi i}\!\!\bigg)^{L{-}3}\!\!\!\int_{(\PC)^{L-3}}\frac{1}{\wb_L}\prod_{\substack{i,j=4\\i<j}}^L|w_{ij}|^{2s_{ij}}\sigma\!\!\left[\!\!\rule{0cm}{0.8cm}\right.\!\left(\frac{s_{14}}{w_4}\!+\!\sum_{j=5}^L\frac{s_{4j}}{w_{4j}}\right)\!\underbrace{\prod_{k=5}^L\left(\frac{s_{1k}}{w_k}\!+\!\!\!\sum_{j=k+1}^L\frac{s_{jk}}{w_{kj}}\right)\!\!}_{\text{independent of }w_4}\left.\rule{0cm}{0.8cm}\!\right]\!\prod_{k=4}^L |w_k|^{2s_{1k}}\frac{dw_kd\wb_k}{\wb_{k(k-1)}}\!,
    \end{align*}
    where as in the open-string case, the first term in the parentheses can be recast as
    \begin{equation*}
        \frac{s_{14}}{w_4}+\sum_{j=5}^L\frac{s_{4j}}{w_{4j}}=\frac{1}{\prod_{k=4}^Lw_k^{s_{1k}}\prod_{\substack{i,j=4\\i<j}}^Lw_{ij}^{s_{ij}}}\frac{d}{dw_4}\left(\prod_{k=4}^Lw_k^{s_{1k}}\prod_{\substack{i,j=4\\i<j}}^Lw_{ij}^{s_{ij}}\right),
    \end{equation*}
    which will allow us to perform one integration trivially in the following.

    As in the proof of \propref{prop:limits0openphys} we distinguish between $\sigma\in\mathfrak{S}_L^{1,2,3,4}$ and $\sigma\notin\mathfrak{S}_L^{1,2,3,4}$:
    \begin{itemize}
        \item For $\sigma\in\mathfrak{S}_L^{1,2,3,4}$, so $\sigma(4)=4$, the term $\frac{s_{14}}{w_4}+\sum_{j=5}^L\frac{s_{4j}}{w_{4j}}$ is even invariant under $\sigma$ and we can write the integrand as a total derivative w.r.t.~$w_4$:
        \begin{align*}
            &\lim_{s_{3\bullet}\to0}\hcF_{L,4}^\sigma(0)\\
            &=\bigg(\!\!\frac{-1}{2\pi i}\!\!\bigg)^{L{-}3}\!\!\!\int_{(\PC)^{L-3}}d_{w_4}\!\left(\!\frac{|w_4|^{2s_{14}}d\wb_4}{\wb_L(\wb_4\,{-}\,1)}\prod_{\substack{i,j=4\\i<j}}^L|w_{ij}|^{2s_{ij}}\,\sigma\!\!\left[\prod_{k=5}^L\left(\frac{s_{1k}}{w_k}+\sum_{j=k+1}^L\frac{s_{kj}}{w_{kj}}\right)\right]\prod_{k=5}^L|w_k|^{s_{1k}}\frac{dw_kd\wb_k}{\wb_{k(k-1)}}\right)\\
            &=\bigg(\!\!\frac{-1}{2\pi i}\!\!\bigg)^{L{-}4}\!\!\!\int_{(\PC)^{L-4}}\!\frac{1}{\wb_L}\prod_{\substack{i,j=5\\i<j}}^L\!|w_{ij}|^{2s_{ij}}\,\sigma\!\!\left[\prod_{k=5}^L\!\left(\!\frac{s_{1k}}{w_k}{+}\sum_{j=k+1}^L\frac{s_{kj}}{w_{kj}}\!\right)\!\right]\prod_{k=5}^L|w_k|^{s_{1k}}|1\,{-}\,w_k|^{2s_{4k}}\frac{dw_kd\wb_k}{\wb_{k(k-1)}}\Big|_{\wb_4=1},
        \end{align*}
        where analogous to the proof of \thmref{prop:equationclosed} we applied Stoke's theorem and have to take the residues from the poles in $\wb_4$ at 1 and $\wb_5$, where the latter vanishes due to the term $|w_{45}|^{2s_{45}}$.
        A similar representation of this can be derived using \propref{prop:limitsclosedphysical} \textit{(i)}.
        \item Conversely, if $\sigma(4)\neq4$, so $\sigma\notin\mathfrak{S}_L^{1,2,3,4}$, with the same arguments as in \propref{prop:limits0openphys}, we obtain that there is an $\ell\neq4$ so that the second term in the brackets is independent of $w_\ell$, while the first term is
            \begin{equation*}
        \frac{s_{1\ell}}{w_\ell}+\sum_{j=5}^L\frac{s_{\ell j}}{w_{\ell j}}=\frac{1}{\prod_{k=4}^Lw_k^{s_{1k}}\prod_{\substack{i,j=4\\i<j}}^Lw_{ij}^{s_{ij}}}\frac{d}{dw_\ell}\left(\prod_{k=4}^Lw_k^{s_{1k}}\prod_{\substack{i,j=4\\i<j}}^Lw_{ij}^{s_{ij}}\right).
    \end{equation*}
    Now, we can again write the integrand as a total derivative -- but now w.r.t.~$w_\ell$ for $\ell\neq4$ -- and the residues from poles in $\wb_\ell$ at $\wb_{\ell-1}$ and $\wb_{\ell+1}$ will cancel exactly:
        \begin{align*}
            \quad\qquad\lim_{s_{3\bullet}\to0}\hcF_{L,4}^\sigma(0)&=\frac{1}{(-2\pi i)^{L-3}}\int_{(\PC)^{L-3}}\frac{1}{\wb_L}\prod_{\substack{i,j=4\\i<j}}^L|w_{ij}|^{2s_{ij}}\\
            &\quad\times\left(\frac{s_{1\ell}}{w_\ell}+\sum_{j=5}^L\frac{s_{\ell j}}{w_{\ell j}}\right)\underbrace{\prod_{k=5}^L\left(\frac{s_{1\sigma(k)}}{w_{\sigma(k)}}+\sum_{j=k+1}^L\frac{s_{\sigma(k)\sigma(j)}}{w_{\sigma(k)\sigma(j)}}\right)}_{\text{independent of }w_\ell}\prod_{k=4}^L|w_k|^{2s_{1k}}\frac{dw_kd\wb_k}{\wb_{k(k-1)}}\\
            &=\frac{1}{(-2\pi i)^{L-3}}\int_{(\PC)^{L-3}}d_{w_\ell}\Bigg(\frac{d\wb_\ell}{\wb_L(\wb_\ell-\wb_{\ell-1})}
            \prod_{\substack{i,j=4\\i<j}}^L|w_{ij}|^{2s_{ij}}\\
            &\quad\times\prod_{k=5}^L\left(\frac{s_{1\sigma(k)}}{w_{\sigma(k)}}+\sum_{j=k+1}^L\frac{s_{\sigma(k)\sigma(j)}}{w_{\sigma(k)\sigma(j)}}\right)\prod_{\substack{k=4\\k\neq \ell}}|w_k|^{2s_{1k}}\frac{dw_kd\wb_k}{\wb_{k(k-1)}}\Bigg)\\
            &=0.
        \end{align*}
        With that, all of these components vanish.
    \end{itemize}
    Thus -- just as in the open-string case -- only the components corresponding to $\sigma\in\mathfrak{S}_L^{1,2,3,4}$ yield a non-vanishing result in the limit, and the number of non-vanishing components is $(L\,{-}\,4)!$.
\end{proof}

\subsection*{\boldmath Limit $z_3\,{\to}\,1$} We compute now the vector $\hcF_L(1)\coloneq\lim_{z_3\to 1}|1\,{-}\,z_3|^{-2e_1}\hcF_L(z_3)$:
\begin{proposition}\label{prop:limit1cF}
    Let $\hcF_L(1)=(\hcF_{L,4}(1),\ldots, \hcF_{L,L+1}(1))^T$ be the regularized boundary value as defined in \eqn{eqn:closedboundval}. Then the first subvector $\hcF_{L,4}(1)$ reads
    \begin{align}\label{eqn:F1closed}
        \frac{1}{(-2\pi i)^{L-3}}\left(\!\int_{(\PC)^{L-3}}\frac{1}{\wb_L}\prod_{\substack{i,j=1\\i<j;i,j\neq3}}^L\!|w_{ij}|^{2s_{ij}'}\prod_{k=4}^L\frac{dw_kd\wb_k}{\wb_{k(k-1)}}\, \sigma\!\left[\prod_{k=i}^L\left(\frac{s_{1k}}{w_k}+\sum_{j=k+1}^L\frac{s_{kj}}{w_{kj}}\right)\right]\right)_{\sigma\in\mathfrak{S}^{1,2,3}_{L}},
    \end{align}
    where $s_{ij}'\coloneq s_{ij}$ if $i\neq2$ and $s_{2j}'\coloneq s_{2j}+s_{3j}$, and also $\wb_3\coloneq1$.
\end{proposition}
\begin{proof}
    As argued before, the matrices of the KZ equation \eqref{eqn:KZequationclosed} are the same as in the open-string case from \eqn{eqn:KZequation}. Thus, the matrix $e_1$ is also here of the form
    \begin{equation*}
        e_1=\left(\begin{array}{cc}
			s_{23}\,I_{(L-3)!} & 0\\
			c & d
        \end{array}\right),
	\end{equation*}
    which leads to the matrix exponential
    \begin{equation*}
        |1-z_3|^{-2e_1}=\left(\begin{array}{cc}
            |1-z_3|^{-2s_{23}}I_{(L-3)!} & 0 \\
            \ast & \ast
        \end{array}\right).
    \end{equation*}
    This result combined with \propref{prop:limitsclosed} \textit{(ii)} yields \eqn{eqn:F1closed}.
\end{proof}

As for the other limit, we will again proceed to take the physical limit $s_{3\bullet}\to0$ of this result.
\begin{proposition}\label{prop:closedF1phys}
    Taking the physical limit of $\hcF_{L,4}(1)$ yields
    \begin{align*}
        &\lim_{s_{3\bullet}\to0}\hcF_{L,4}(1)\notag\\
        &=\left(\frac{1}{(-2\pi i)^{L-3}}\int_{(\PC)^{L-3}}\frac{1}{\zb_L}\prod_{\substack{i,j=1\\i<j;i,j\neq3}}^L|z_{ij}|^{2s_{ij}}\,\prod_{k=4}^L\frac{dz_kd\zb_k}{\zb_{k(k-1)}}\ \sigma\left[\prod_{k=i}^L\left(\frac{s_{1k}}{z_k}+\sum_{j=k+1}^L\frac{s_{kj}}{z_{kj}}\right)\right]\right)_{\sigma\in\mathfrak{S}^{1,2,3}_{L}},
    \end{align*}
    where we denote $\zb_{43}=\zb_4-1$ in this limit and $\lim_{s_{3\bullet}\to0}$ refers to the limit the Mandelstam variables $s_{3i}$, for $i=1,\ldots,L$, tend to 0.
\end{proposition}
\begin{proof}
    Using the result from the previous \propref{prop:limit1cF}, we can immediately take the limit for the variables $s_{3\bullet}$ and end up with the given result.
\end{proof}

\begin{rmk}
    While the explicit calculation of the lower components of the vector $\hcF_L(1)$ is very involved, they can be calculated in the physical limit $s_{3\bullet}\to0$ explicitly which we show in the examples in \secref{sec:closed4} and \secref{sec:closed5}. The structure that we see for these components in the physical limit is analogous to the one-string case discussed in \rmkref{rmk:openlastcomps}: they are either also of the form of $L$-point amplitudes or of the schematic form ``($L$-point amplitude) $-$ (($L{-}1$)-point amplitude)''.
\end{rmk}

\subsection*{The closed-string amplitude recursion}
Collecting the results from \secref{sec:closedgeneral} and \secref{ssec:complexlimits}, we can formulate the genus-zero closed-string amplitude recursion in the following Theorem.
\begin{theorem}\label{thm:closedrec}
For any $L\geq4$ there exist $(L{-}2)!$-dimensional vectors~$\mathscr{A}_L$ and~$\mathscr{B}_L$ whose entries are functions of the Mandelstam variables $s_{ij}$, holomorphic in a neighborhood of the origin, such that: 
    \begin{enumerate}[label=(\roman*)]
        \item the first $(L{-}4)!$ components of $\mathscr{A}_L$ coincide with the closed-string $(L{-}1)$-point integrals $\cF_{L-1}^\sigma$, $\sigma\in\mathfrak{S}_{L-4}$, from \eqn{eqn:cF}, and the remaining components are identically zero;
        \item the first $(L{-}3)!$ components of $\mathscr{B}_L$ are the closed-string $L$-point integrals~$\cF_L^\sigma$, $\sigma\in\mathfrak{S}_{L-3}$, from \eqn{eqn:cF};
        \item there exist $(L{-}2)!$-dimensional square matrices $e_0$, $e_1$, whose entries are integer linear combinations of Mandelstam variables $s_{ij}$, such that the Taylor expansions of $\mathscr{A}_L$ and $\mathscr{B}_L$ are related by the Deligne associator $\Phi^{\rm sv}(e_0,e_1)$ via
        \begin{equation*}
            \mathscr{B}_L=\Phi(e_0,e_1)\mathscr{A}_{L}.
        \end{equation*}
    \end{enumerate}
    In particular, the Taylor expansion of any closed-string $L$-point integral~$\cF_L^\sigma$ can be obtained from the Taylor expansion of the closed-string $(L{-}1)$-point integrals $\cF_{L-1}^\sigma$ and the Deligne associator $\Phi^{\rm sv}(e_0,e_1)$, and therefore the coefficients of these Taylor expansions belong to the ring $\mathcal Z^{\rm sv}$ of single-valued multiple zeta values.
\end{theorem}
\begin{proof}
    Let $\mathscr{A}_{L}=\lim_{s_{3\bullet}\to0}\hcF_L(0)$ and $\mathscr{B}_{L}=\lim_{s_{3\bullet}\to0}\hcF_L(1)$. Then by comparing the result of \propref{prop:limits0closedphys} with \eqn{eqn:cF} we find that the first $(L{-}4)!$ components of $\mathscr{A}_{L}$ are exactly the $(L{-}1)$-point closed-string amplitudes, proving \emph{(i)}. Accordingly, \propref{prop:closedF1phys} provides that the first $(L{-}3)!$ components of $\mathscr{B}_{L}$ yield the $L$-point closed-string amplitudes from \eqn{eqn:cF}, which shows \emph{(ii)}. To prove \emph{(iii)}, we take the physical limit $s_{3\bullet}\to0$ of \eqn{eqn:closedDeligne} of \thmref{prop:equationclosed} to find the Deligne associator relation with the matrices $e_0$ and $e_1$ from \eqn{eqn:KZequation} after taking the limit $s_{3\bullet}\to0$. The fact that the coefficients of the Taylor expansion of any~$\cF_L^\sigma$ belong to~$\mathcal Z^{\rm sv}$ follows inductively by combining the fact that $\mathscr{A}_{4}=(1,0)^T$ (see Section \ref{sec:closed4} below) and that, for any~$L$, the coefficients of the Taylor expansion in the variables $s_{ij}$ of the Deligne associator $\Phi^{\rm sv}(e_0,e_1)$ belong to $\mathcal Z^{\rm sv}$. This last assertion is a consequence of the fact that the coefficients of $\Phi^{\rm sv}(e_0,e_1)$ as a series in $e_0,e_1$ belong to $\mathcal Z^{\rm sv}$, combined with the fact that the entries of $e_0,e_1$ are $\mathbb Z$-linear combinations of the variables $s_{ij}$.
\end{proof}

\subsection{From three- to four-point amplitude}\label{sec:closed4}

As in the open-string case, we want to illustrate the general results derived in \secref{sec:closedgeneral} with the simplest examples, starting with the case $L\,{=}\,4$ which lets us translate between the three- and four-point closed-string amplitudes. In contrast to the open-string example in \secref{sec:open4}, here in the closed-string case, the real points $x_1\,{=}\,0,\,x_2\,{=}\,1,\,x_3,\,x_4$ are replaced by complex points $z_1\,{=}\,0,\,z_2\,{=}\,1,\,z_3,\,z_4$, while the Mandelstam variables $s_{13},s_{14},s_{23},s_{24},s_{34}\in\mathbb{C}$ are still required to satisfy the condition $\Re(s_{14})$, $\Re(s_{24})$, $\Re(s_{34})\,{>}\,0$, as well as the new condition $\Re(s_{14}+s_{24}+s_{34})<1/2$, so that the integrals in the components of the following vector are absolutely convergent (see \cite[App.~A]{VanZerb}):
\begin{align*}
	\hcF_4(z_3)&=\mat{\hcF_{4,4}(z_3)\\\hcF_{4,5}(z_3)}\\
    &=\mat{s_{14}\,|z_3|^{2s_{13}}|1\,{-}\,z_3|^{2s_{23}}\int_{\mathbb{P}^1_{\mathbb{C}}}|z_4|^{2s_{14}}|1\,{-}\,z_4|^{2s_{24}} |z_3\,{-}\,z_4|^{2s_{34}} \frac{\zb_3dz_4d\zb_4}{2\pi i(\zb_3\,{-}\,\zb_4)|z_4|^2}\\
    s_{24}\,|z_3|^{2s_{13}}|1\,{-}\,z_3|^{2s_{23}}\int_{\mathbb{P}^1_{\mathbb{C}}}|z_4|^{2s_{14}}|1\,{-}\,z_4|^{2s_{24}} |z_3\,{-}\,z_4|^{2s_{34}} \frac{\zb_3dz_4d\zb_4}{2\pi i(\zb_3\,{-}\,\zb_4)(1\,{-}\,z_4)\zb_4}}.
\end{align*}
Comparing to \eqn{eq:combSelbclosed}, we find that $\hcF_{4,4}(z_3)=s_{14}\,\SC[1](z_3)$ and $\hcF_{4,5}(z_3)=-s_{24}\,\SC[2](z_3)$.
For fixed values of the parameters $s_{ij}$, the function $\hcF_4:\mathbb{C}\smallsetminus \{0,1\}\to \mathbb{C}^2$ is real-analytic and single-valued.

Computing the derivative w.r.t.~$z_3$ of the vector $\hcF_4(z_3)$ yields the KZ equation
\begin{equation}\label{eq:010823n3}
	\frac{\pd}{\pd z_3}\hcF_4(z_3)=\left(\frac{e_0}{z_3}+\frac{e_1}{z_3-1}\right)\hcF_4(z_3)\,,
\end{equation}
where $e_0$ and $e_1$ are the same matrices found in the open-string case (see \eqn{eq:230526n1}). One crucial ingredient for performing this computation explicitly is a complex version of the relation \eqref{eqn:open4relation}, which can be derived using the same methods as in the proof of \thmref{prop:equationclosed} to read
\begin{equation*}
	0=s_{14}\SC[1](z_3)+s_{24}\SC[2](z_3)+s_{34}\SC[3](z_3).
\end{equation*}

Let us now calculate the limit vectors $\hcF_4(0)$ and $\hcF_4(0)$ explicitly:
For the limit $z_3\to0$ we calculate
\begin{align*}
	\hcF_4(0)&=\lim_{z_3\to0}|z_3|^{-2e_0}\hcF_4(z_3)\\
	&=\lim_{z_3\to0}\mat{|z_3|^{-2(s_{13}+s_{14}+s_{34})}\hcF_{4,4}(z_3)+\tfrac{s_{14}}{s_{14}+s_{34}}(|z_3|^{-2s_{13}}-|z_3|^{-2(s_{13}+s_{14}+s_{34})})\hcF_{4,5}(z_3)\\|z_3|^{-2s_{13}}\hcF_{4,5}(z_3)},
\end{align*}
with the first component being (in agreement with \propref{prop:limits0closed})
\begin{align*}
\hcF_{4,4}(0)&=\lim_{z_3\to 0}|z_3|^{-2(s_{13}+s_{14}+s_{34})}\hcF_{4,4}(z_3)\\
&=\lim_{z_3\to 0}s_{14}\,|z_3|^{-2(s_{14}+s_{34})}\int_{\mathbb{P}^1_{\mathbb{C}}}|z_4|^{2s_{14}}|1\,{-}\,z_4|^{2s_{24}}|z_3\,{-}\,z_4|^{2s_{34}}\,\frac{\zb_3dz_4d\zb_4}{2\pi i(\zb_3\,{-}\,\zb_4)|z_4|^2}\\
&=\lim_{z_3\to 0}s_{14}\,|z_3|^{-2(s_{14}+s_{34})}\int_{\mathbb{P}^1_{\mathbb{C}}}|w|^{2s_{14}}|z_3|^{2s_{14}}|1\,{-}\,wz_3|^{2s_{24}}|z_3|^{2s_{34}}|1\,{-}\,w|^{2s_{34}}\,\frac{dwd\overline{w}}{2\pi i(1-\overline{w})|w|^2}\\
&=s_{14}\int_{\mathbb{P}^1_{\mathbb{C}}}|w|^{2s_{14}-2}|1\,{-}\,w|^{2s_{34}}\frac{dwd\overline{w}}{2\pi i(1\,{-}\,\overline{w})}=\frac{\Gamma(1+s_{14})\Gamma(1+s_{34})\Gamma(1-s_{14}-s_{34})}{\Gamma(1-s_{14})\Gamma(1-s_{34})\Gamma(1+s_{14}+s_{34})},
\end{align*}
where to pass from the first to the second line we substituted\footnote{We can assume that $z_3\neq 0$.} $w=z_4/z_3$ and the last equality can be shown using \eqn{eqn:complexBeta}.
\propref{prop:limits0closed} (or direct calculation) also tells us that $\hcF_{4,5}(0)=0$.
We have therefore found that 
\begin{equation*}
\hcF_4(0)=\mat{\frac{\Gamma(1+s_{14})\Gamma(1+s_{34})\Gamma(1-s_{14}-s_{34})}{\Gamma(1-s_{14})\Gamma(1-s_{34})\Gamma(1+s_{14}+s_{34})}\\0}.
\end{equation*}

For the limit $z_3\to1$ we calculate the first component of $\hcF_4(1)$ to be
\begin{align*}\label{eq:010823n6}
\hcF_{4,4}(1)&=\lim_{z_3\to 1}\left[|z_3|^{-2e_1}\hcF_4(z_3)\right]^{(1)}=\lim_{z_3\to 1}|1\,{-}\,z_3|^{-2s_{23}}\hcF_{4,4}(z_3)\\
	&=s_{14}\int_{\mathbb{P}^1_{\mathbb{C}}}|z_4|^{2s_{14}}|1\,{-}\,z_4|^{2(s_{24}+s_{34})}\,\frac{dz_4d\zb_4}{2\pi i(1\,{-}\,\zb_4)|z_4|^2}\\
	&=\frac{\Gamma(1+s_{14})\Gamma(1+s_{24}+s_{34})\Gamma(1-s_{14}-s_{24}-s_{34})}{\Gamma(1-s_{14})\Gamma(1-s_{24}-s_{34})\Gamma(1+s_{14}+s_{24}+s_{34})}\,,
\end{align*}
where the last equality follows as before using \eqn{eqn:complexBeta}.

Just as seen for the open-string case in \secref{sec:open4}, the calculation of the second component of the limit vector is more involved. It again relies on identities of hypergeometric functions and their complex analogues. The details of this calculation are shown in \appref{app:4pt2ndcompclosed} and the result is
\begin{align*}
	\hcF_{4,5}(1)&=s_{24}\frac{\Gamma(1+s_{14})\Gamma(s_{24}+s_{34})\Gamma(1-s_{14}-s_{24}-s_{34})}{\Gamma(1-s_{14})\Gamma(1-s_{24}-s_{34})\Gamma(1+s_{14}+s_{24}+s_{34})}\notag\\
 &\quad+s_{24}\frac{\Gamma(1+s_{34})\Gamma(-s_{24}-s_{34})\Gamma(1+s_{24})}{\Gamma(1-s_{34})\Gamma(1+s_{24}+s_{34})\Gamma(1-s_{24})}.
\end{align*}

\subsection*{Relating the limits}
Combining the above results, we can relate the two limit vectors using the Deligne associator $\Phi^{\rm sv}=\Phi^{\rm sv}(e_0,e_1)$, according to \thmref{prop:equationclosed} as follows:
\begin{align*}
    &\mat{\frac{\Gamma(1+s_{14})\Gamma(1+s_{24}+s_{34})\Gamma(1-s_{14}-s_{24}-s_{34})}{\Gamma(1-s_{14})\Gamma(1-s_{24}-s_{34})\Gamma(1+s_{14}+s_{24}+s_{34})}\\ s_{24}\frac{\Gamma(1+s_{14})\Gamma(s_{24}+s_{34})\Gamma(1-s_{14}-s_{24}-s_{34})}{\Gamma(1-s_{14})\Gamma(1-s_{24}-s_{34})\Gamma(1+s_{14}+s_{24}+s_{34})}+s_{24}\frac{\Gamma(1+s_{34})\Gamma(-s_{24}-s_{34})\Gamma(1+s_{24})}{\Gamma(1-s_{34})\Gamma(1+s_{24}+s_{34})\Gamma(1-s_{24})}}\notag\\
    &\hspace{55ex}=\Phi^{\rm sv}\cdot \mat{\frac{\Gamma(1+s_{14})\Gamma(1+s_{34})\Gamma(1-s_{14}-s_{34})}{\Gamma(1-s_{14})\Gamma(1-s_{34})\Gamma(1+s_{14}+s_{34})}\\0}.
\end{align*}
Note that this equation is the single-valued projection of the respective equation for the open-string amplitudes (cf.~\eqn{eqn:4ptopenlimitsrel}).

In particular, taking the physical limit $s_{3\bullet}\to 0$, where $\bullet=1,2,4$, we find that
\begin{equation*}
\mat{\frac{\Gamma(1+s_{14})\Gamma(1+s_{24})\Gamma(1-s_{14}-s_{24})}{\Gamma(1-s_{14})\Gamma(1-s_{24})\Gamma(1+s_{14}+s_{24})}\\ \frac{\Gamma(1+s_{14})\Gamma(1+s_{24})\Gamma(1-s_{14}-s_{24})}{\Gamma(1-s_{14})\Gamma(1-s_{24})\Gamma(1+s_{14}+s_{24})}-1}\,=\,\Phi^{\rm sv}\big|_{s_{3\bullet}= 0}\cdot \mat{1\\0}\,,
\end{equation*}
which relates the closed-string four-point tree amplitude (first component of the l.h.s., called Virasoro--Shapiro amplitude) with the closed-string three-point tree level amplitude (in our normalization just 1 as seen in the first component of the r.h.s.). The second component of the l.h.s.~can be interpreted as a closed-string four-point amplitude minus a three-point amplitude.

\subsection{From four- to five-point amplitude}\label{sec:closed5}
Let $L\,{=}\,5$, then the vector $\hcF_5(z_3)$ from \eqn{eqn:basisclosed} obeys the KZ equation \eqref{eqn:KZequationclosed} with the matrices the same as in the open-string case in \secref{sec:open5}, see \eqn{eqn:5e}.

Applying the results from \secref{sec:closedgeneral} and \secref{ssec:complexlimits} we find a recursion bringing us from the four-point to the five-point closed-string amplitude using the Deligne associator:
\begin{align*}
	\Phi^{\mathrm{sv}}(\Tilde{e}_0,\Tilde{e}_1)\left(\!\!\begin{array}{c}\cF_4^\mathrm{id}\\0\\0\\0\\0\\0\end{array}\!\!\right)&=\left(1{+}2\zeta_3[\Tilde{e}_0{+}\Tilde{e}_1,[\Tilde{e}_0,\Tilde{e}_1]]+\ldots\right)\left(\!\!\begin{array}{c}\frac{\Gamma(1+s_{15})\Gamma(1+s_{45})\Gamma(1-s_{15}-s_{45})}{\Gamma(1-s_{15})\Gamma(1-s_{45})\Gamma(1+s_{15}+s_{45})}\\0\\0\\0\\0\\0\end{array}\!\!\right)=\left(\!\!\!\!\begin{array}{c}\cF_{5}^{\mathrm{id}}\\ \cF_{5}^{(45)}\\ \ast \\ \ast \\ \ast \\ \ast\end{array}\!\!\!\!\right)\!,
\end{align*}
with $\Tilde{e}_i=\lim_{s_{3\bullet}\to0}e_i$ and $\cF_N^\sigma$ the closed-string amplitudes from \eqn{eqn:cF}. The full calculation of this result including the lower components of the vector on the r.h.s.~can be found in \appref{app:closed5}.

\begin{rmk}
    Note that the above relation is the single-valued image of the result from \secref{sec:open5} through applying the single-valued map from ref.~\cite{Brown:2019wna}.
\end{rmk}


\section*{Appendix}
\appendix
\addtocontents{toc}{\protect\setcounter{tocdepth}{0}}

\section{Second components of the limit vectors for \texorpdfstring{$L=4$}{L=4}}\label{app:secondcomp}
In this appendix we show the full computations the second components of the limit vectors for $x_3\to1$ (open-string case) and $z_3\to1$ (closed-string case) for $L\,{=}\,4$. Both calculations work along the same lines and involve identities of hypergeometric functions.

\subsection{Open-string: \texorpdfstring{$x_3\to1$}{x3->1}}\label{app:4pt2ndcomp}
As described in \secref{sec:open4}, for the calculation of the second component of the limit vector $\hF_4(1)$, we use the matrix exponential of $e_1$ and find
{\small
\begin{align*}
	&\hF_{4,5}(1)=\lim_{x_3\to1}\frac{s_{24}}{s_{24}+s_{34}}\left((1-x_3)^{-s_{23}}-(1-x_3)^{-s_{234}}\right)\hF_{4,4}(x_3)+(1-x_3)^{-s_{234}}\hF_{4,5}(x_3)\\
	&=\lim_{x_3\to1}x_3^{s_{134}+1}\int_{0}^{1}dw\,w^{s_{14}}(1-x_3w)^{s_{24}}(1-w)^{s_{34}}\notag\\
	&\hspace{15ex}\times\left[\frac{s_{24}}{s_{24}+s_{34}}\frac{s_{14}}{x_3w}(1-(1-x_3)^{-s_{24}-s_{34}})+\frac{s_{24}}{1-x_3w}(1-x_3)^{-s_{24}-s_{34}}\right]\\
	&=\lim_{x_3\to1}s_{24}x_3^{1+s_{134}}\left(\frac{s_{14}}{s_{24}+s_{34}}\frac{1}{x_3}\frac{\Gamma(s_{14})\Gamma(1+s_{34})}{\Gamma(1+s_{14}+s_{34})}\Fhyper{-s_{24}}{s_{14}}{1+s_{14}+s_{34}}{x_3}\left(1-(1-x_3)^{-s_{24}-s_{34}}\right)\right.\notag\\
	&\left.\hspace{15ex}+(1-x_3)^{-s_{24}-s_{34}}\frac{\Gamma(1+s_{14})\Gamma(1+s_{34})}{\Gamma(2+s_{14}+s_{34})}\Fhyper{1-s_{24}}{1+s_{14}}{2+s_{14}+s_{34}}{x_3}\right),
\end{align*}
}%
where we used the integral formula for the hypergeometric function
{\small
\begin{equation*}
	\label{eqn:Fhyper}
	\Fhyper{a}{b}{c}{z}=\frac{\Gamma(c)}{\Gamma(b)\Gamma(c-b)}\int_0^1dt\,t^{b-1}(1-t)^{c-b-1}(1-tz)^{-a}.
\end{equation*}
}%
The next step is to find the relevant expansions of the divergent terms (i.e.~the terms with the $(1-x_3)^{-s_{24}-s_{34}}$ factor) near $x_3\,{\to}\,1$. This is done by using a functional identity for $_2F_1$ which relates its value at $x$ with its value at $1\,{-}\,x$, namely (see e.g.~\cite[Chapter XIV]{Whittaker_Watson_1996})
{\small
\begin{align}\label{eqn:2F1id}
	\Fhyper{a}{b}{c}{x}&=\frac{\Gamma(c)\Gamma(c-a-b)}{\Gamma(c-a)\Gamma(c-b)}\Fhyper{a}{b}{a+b-c+1}{1-x}\notag\\
	&\quad+\frac{\Gamma(c)\Gamma(a+b-c)}{\Gamma(a)\Gamma(b)}(1-x)^{c-a-b}\Fhyper{c-a}{c-b}{1+c-a-b}{1-x}.
\end{align}
}%
Using this identity for the diverging terms we find (the first term can just be evaluated at $x_3\,{=}\,1$ immediately and higher terms be written as $\cO(x_3\,{-}\,1)$)
\begin{small}
\begin{align*}
	&\hF_{4,5}(1)=\lim_{x_3\to1}\left(\frac{s_{24}}{s_{24}\,{+}\,s_{34}}x_3^{s_{134}}\frac{\Gamma(1\,{+}\,s_{14})\Gamma(1\,{+}\,s_{24}\,{+}\,s_{34})}{\Gamma(1\,{+}\,s_{14}\,{+}\,s_{24}\,{+}\,s_{34})}+\cO(x_3\,{-}\,1)\right)\notag\\
	&+\lim_{x_3\to1}s_{24}x_3^{s_{134}+1}(1\,{-}\,x_3)^{-s_{24}-s_{34}}\notag\\
	&\times\left[-\frac{s_{14}}{s_{24}\,{+}\,s_{34}}\frac{1}{x_3}\frac{\Gamma(s_{14})\Gamma(1\,{+}\,s_{34})}{\Gamma(1\,{+}\,s_{14}\,{+}\,s_{34})}\left(\frac{\Gamma(1\,{+}\,s_{14}\,{+}\,s_{34})\Gamma(1\,{+}\,s_{24}\,{+}\,s_{34})}{\Gamma(1\,{+}\,s_{14}\,{+}\,s_{24}\,{+}\,s_{34})\Gamma(1\,{+}\,s_{34})}\Fhyper{-s_{24}}{s_{14}}{-s_{24}-s_{34}}{1\,{-}\,x_3}\right.\right.\notag\\
	&\left.\quad+\frac{\Gamma(1\,{+}\,s_{14}\,{+}\,s_{34})\Gamma(\,{-}\,s_{24}\,{-}\,s_{34}\,{-}\,1)}{\Gamma(\,{-}\,s_{24})\Gamma(s_{14})}(1\,{-}\,x_3)^{s_{24}+s_{34}+1}\Fhyper{1+s_{14}+s_{24}+s_{34}}{1+s_{34}}{2+s_{24}+s_{34}}{1\,{-}\,x_3}\right)\notag\\
	&\quad+\frac{\Gamma(1\,{+}\,s_{14})\Gamma(1\,{+}\,s_{34})}{\Gamma(2\,{+}\,s_{14}\,{+}\,s_{34})}\left(\frac{\Gamma(2\,{+}\,s_{14}\,{+}\,s_{34})\Gamma(s_{24}\,{+}\,s_{34})}{\Gamma(1\,{+}\,s_{14}\,{+}\,s_{24}\,{+}\,s_{34})\Gamma(1\,{+}\,s_{34})}\Fhyper{1-s_{24}}{1+s_{14}}{1-s_{24}-s_{34}}{1\,{-}\,x_3}\right.\notag\\
	&\quad\left.\left.+\frac{\Gamma(2\,{+}\,s_{14}\,{+}\,s_{34})\Gamma(\,{-}\,s_{24}\,{-}\,s_{34})}{\Gamma(1\,{-}\,s_{24})\Gamma(1\,{+}\,s_{14})}(1\,{-}\,x_3)^{s_{24}+s_{34}}\Fhyper{1+s_{14}+s_{24}+s_{34}}{1+s_{34}}{1+s_{24}+s_{34}}{1\,{-}\,x_3}\right)\right]\\
	\label{eqn:limit_no_expans}
	&=\frac{s_{24}}{s_{24}\,{+}\,s_{34}}\frac{\Gamma(1\,{+}\,s_{14})\Gamma(1\,{+}\,s_{24}\,{+}\,s_{34})}{\Gamma(1\,{+}\,s_{14}\,{+}\,s_{24}\,{+}\,s_{34})}\notag\\
	&+\lim_{x_3\to1}s_{24}x_3^{s_{134}+1}(1\,{-}\,x_3)^{-s_{24}-s_{34}}\notag\\
	&\quad\times\left(-\frac{s_{14}}{s_{24}\,{+}\,s_{34}}\frac{1}{x_3}\frac{\Gamma(s_{14})\Gamma(1\,{+}\,s_{24}\,{+}\,s_{34})}{\Gamma(1\,{+}\,s_{14}\,{+}\,s_{24}\,{+}\,s_{34})}\Fhyper{-s_{24}}{s_{14}}{-s_{24}-s_{34}}{1\,{-}\,x_3}\right.\notag\\
	&\quad\quad\quad\left.+\frac{\Gamma(1\,{+}\,s_{14})\Gamma(s_{24}\,{+}\,s_{34})}{\Gamma(1\,{+}\,s_{14}\,{+}\,s_{24}\,{+}\,s_{34})}\Fhyper{1-s_{24}}{1+s_{14}}{1-s_{24}-s_{34}}{1\,{-}\,x_3}\right)\notag\\
	&\quad+s_{24}x_3^{s_{134}+1}\frac{\Gamma(1\,{+}\,s_{34})\Gamma(\,{-}\,s_{24}\,{-}\,s_{34})}{\Gamma(1\,{-}\,s_{24})}\Fhyper{1+s_{14}+s_{24}+s_{34}}{1+s_{34}}{1+s_{24}+s_{34}}{1\,{-}\,x_3}\notag
	\\&\quad+\cO(x_3\,{-}\,1)\\
	&=\lim_{x_3\to1}\frac{s_{24}}{s_{24}\,{+}\,s_{34}}\frac{\Gamma(1\,{+}\,s_{14})\Gamma(1\,{+}\,s_{24}\,{+}\,s_{34})}{\Gamma(1\,{+}\,s_{14}\,{+}\,s_{24}\,{+}\,s_{34})}\notag\\
	&\quad+s_{24}(1\,{-}\,x_3)^{-s_{24}-s_{34}}\notag\\
	&\quad\times\left(\underbrace{-\frac{\Gamma(1\,{+}\,s_{14})\Gamma(s_{24}\,{+}\,s_{34})}{\Gamma(1\,{+}\,s_{14}\,{+}\,s_{24}\,{+}\,s_{34})}+\frac{\Gamma(1\,{+}\,s_{14})\Gamma(s_{24}\,{+}\,s_{34})}{\Gamma(1\,{+}\,s_{14}\,{+}\,s_{24}\,{+}\,s_{34})}}_{=0}+\cO(x_3\,{-}\,1)\right)\notag\\
	&\quad-\frac{s_{24}}{s_{24}\,{+}\,s_{34}}\frac{\Gamma(1\,{+}\,s_{34})\Gamma(1\,{-}\,s_{24}\,{-}\,s_{34})}{\Gamma(1\,{-}\,s_{24})}\notag
	\\&\quad+\cO(x_3\,{-}\,1)\\
	&=\lim_{x_3\to1}\frac{s_{24}}{s_{24}\,{+}\,s_{34}}\left(\frac{\Gamma(1\,{+}\,s_{14})\Gamma(1\,{+}\,s_{24}\,{+}\,s_{34})}{\Gamma(1\,{+}\,s_{14}\,{+}\,s_{24}\,{+}\,s_{34})}-\frac{\Gamma(1\,{+}\,s_{34})\Gamma(1\,{-}\,s_{24}\,{-}\,s_{34})}{\Gamma(1\,{-}\,s_{24})}\right)\notag\\
	&\quad+\cO\left((1\,{-}\,x_3)^{1-s_{24}-s_{34}}\right)+\cO(x_3\,{-}\,1).
\end{align*}
\end{small}
In the above equations, we used the Gamma function equation $x\,\Gamma(x)=\Gamma(x+1)$, the expansion $x_3^a=1+\cO(x_3-1),$ for $a\neq0,$ and the series expansion of the hypergeometric function
{\small
\begin{align*}
	\Fhyper{a}{b}{c}{1-x}&=\sum_{k=0}^\infty\frac{\Gamma(a+k)\Gamma(b+k)\Gamma(c)}{k!\,\Gamma(a)\Gamma(b)\Gamma(c+k)}(1-x)^k\\
	&=1+\cO(x-1),
\end{align*}
}%
and by assuming $\Re(s_{24})$ and $\Re(s_{34})$ to be sufficiently small we only needed to keep the zeroth order term (i.e.~the 1 in the above expansion) when going to the limit $x_3\to1$. 
	
This provides the final result (the limit $x_3\,{\to}\,1$ now removes the higher order terms ${\sim}(x_3-1)^{1-s_{24}-s_{34}}$ and ${\sim}(x_3-1)$)
{\small
\begin{equation*}
	\hF_{4,5}(1)=\frac{s_{24}}{s_{24}+s_{34}}\left(\frac{\Gamma(1+s_{14})\Gamma(1+s_{24}+s_{34})}{\Gamma(1+s_{14}+s_{24}+s_{34})}-\frac{\Gamma(1+s_{34})\Gamma(1-s_{24}-s_{34})}{\Gamma(1-s_{24})}\right).
\end{equation*}
}%

\subsection{Closed string: \texorpdfstring{$z_3\to1$}{z3->1}} \label{app:4pt2ndcompclosed}
In order to calculate the second component of the limit vector $\hcF_4(1)$ one can use hypergeometric functions, similar to the open-string case. For this purpose we use the definition of the complex analogue of $_2F_1$ (see e.g.~\cite{Molchanov_2021})
{\small
\begin{equation*}
	{_2F}_1^{\mC}\left[\!\!\begin{array}{c}{\bf{a}},\ {\bf{b}}\\ {\bf{c}}\end{array}\!\!\Big|z\right]=\frac{1}{2\pi i}\frac{\Gamma^{\mC}({\bf{c}})}{\Gamma^{\mC}({\bf{b}})\Gamma^{\mC}({\bf{c}}-{\bf{b}})}\int_{\PC}t^{{\bf{b}}-{\bf{1}}}(1-t)^{{\bf{c}}-{\bf{b}}-{\bf{1}}}(1-zt)^{-{\bf{a}}}dtd\tb,
\end{equation*}
}%
as well as the complex Gamma function
{\small
\begin{equation}
\label{eqn:GammaC}
	\Gamma^{\mC}({\bf{a}})=\Gamma^{\mC}(a|a')=\frac{1}{2\pi i}\int_{\PC}z^{{\bf{a}}-1}e^{2i\mathrm{Re}(z)}dzd\zb=i^{a-a'}\frac{\Gamma(a)}{\Gamma(1-a')}=i^{a'-a}\frac{\Gamma(a')}{\Gamma(1-a)}.
\end{equation}
}%
These definitions use the notation
{\small
\begin{equation*}
	z^{\bf{a}}=z^{a|a'}=z^a\,\zb^{a'}.
\end{equation*}
}%
Ref.~\cite[Thm.~3.9 b)]{Molchanov_2021} shows the identity (for $|1-z|<1$)
{\small
\begin{align}
\label{eqn:2F1idC}
	&{_2F}_1^{\mC}\left[\!\!\begin{array}{c}{\bf{a}},\ {\bf{b}}\\ {\bf{c}}\end{array}\!\!\Big|z\right]=\frac{\Gamma^{\mC}({\bf{c}})\Gamma^{\mC}({\bf{c}}-{\bf{a}}-{\bf{b}})}{\Gamma^{\mC}({\bf{c}}-{\bf{a}})\Gamma^{\mC}({\bf{c}}-{\bf{b}})}\Fhyper{a}{b}{a+b+1-c}{1-z}\Fhyper{a'}{b'}{a'+b'+1-c'}{1-\zb}\notag\\
	&\qquad+\frac{\Gamma^{\mC}({\bf{c}})\Gamma^{\mC}({\bf{a}}+{\bf{b}}-{\bf{c}})}{\Gamma^{\mC}({\bf{a}})\Gamma^{\mC}({\bf{b}})}(1-z)^{{\bf{c}}-{\bf{a}}-{\bf{b}}}\Fhyper{c-a}{c-b}{c+1-a-b}{1-z}\Fhyper{c'-a'}{c'-b'}{c'+1-a'-b'}{1-\zb},
\end{align}
}%
which is the complex analogue of the $_2F_1$ identity \eqref{eqn:2F1id} used in the open-string case for the second component of the four-point $x_3\to1$ limit.

Now, using the matrix exponential
{\small
\begin{equation*}
	|1-z_3|^{-2e_1}=\left(\!\!\begin{array}{cc}
		|1-z_3|^{-2s_{23}} & 0 \\ \frac{s_{24}}{s_{24}+s_{34}}\left(|1-z_3|^{-2s_{23}}-|1-z_3|^{-2s_{234}}\right) & |1-z_3|^{-2s_{234}}
	\end{array}\!\!\right)
\end{equation*}
}%
and the vector $\hcF_4$ after the substitution $w_4=z_4/z_3$
{\small
\begin{equation*}
	\hcF_4(z_3)=|z_3|^{2(s_{134}+1)}|1-z_3|^{2s_{23}}\int_{\PC}|w_4|^{2s_{14}}|1-z_3w_4|^{2s_{24}}|1-w_4|^{2s_{34}}\frac{dw_4 d\wb_4}{2\pi i(1-\wb_4)\zb_3\wb_4}
	\mat{
	\frac{s_{14}}{z_3w_4}\\
	\frac{s_{24}}{1-z_3w_4}
	}\!,
\end{equation*}
}%
we calculate the second component of the limit:
{\small
\begin{align*}
	\hcF_{4,5}(1)&\coloneq\lim_{z_3\to1}\left[|1-z_3|^{-2e_1}\hcF_4(z_3)\right]^{(2)}\\
	&=\lim_{z_3\to1}\frac{s_{24}}{s_{24}+s_{34}}\left(|1-z_3|^{-2s_{23}}-|1-z_3|^{-2s_{234}}\right)\hcF_{4,4}(z_3)+|1-z_3|^{-2s_{234}}\hcF_{4,5}(z_3)\\
	&=\lim_{z_3\to1}\frac{s_{24}}{s_{24}+s_{34}}|z_3|^{2(s_{134}+1)}\int_{\PC}|w_4|^{2s_{14}}|1-z_3w_4|^{2s_{24}}|1-w_4|^{2s_{34}}\frac{dw_4 d\wb_4}{2\pi i(1-\wb_4)\zb_3\wb_4}\frac{s_{14}}{z_3w_4}\notag\\
	&\quad+|1-z_3|^{-2s_{24}-2s_{34}}|z_3|^{2(s_{134}+1)}\int_{\PC}|w_4|^{2s_{14}}|1-z_3w_4|^{2s_{24}}|1-w_4|^{2s_{34}}\frac{dw_4 d\wb_4}{2\pi i(1-\wb_4)\zb_3\wb_4}\notag\\
	&\hspace{45ex}\times\left(\frac{s_{24}}{1-z_3w_4}-\frac{s_{24}}{s_{24}+s_{34}}\frac{s_{14}}{z_3w_4}\right)\\
	&=|z_3|^{2s_{134}}\frac{s_{14}s_{24}}{s_{24}+s_{34}}\frac{\Gamma^{\mC}(s_{14}|s_{14})\Gamma^{\mC}(1+s_{34}|s_{34})}{\Gamma^{\mC}(1+s_{14}+s_{34}|s_{14}+s_{34})} {_2F}_1^{\mC}\left[\!\!\begin{array}{c}-s_{24}|-s_{24},\ s_{14}|s_{14}\\ 1+s_{14}+s_{34}|s_{14}+s_{34}\end{array}\!\!\Big|z\right]\notag\\
	&\quad+s_{24}|1-z_3|^{-2s_{24}-2s_{34}}|z_3|^{2(s_{134}+1)}\frac{1}{\zb_3}\notag\\
	&\hspace{5ex}\times\left(\frac{\Gamma^{\mC}(1+s_{14}|s_{14})\Gamma^{\mC}(1+s_{34}|s_{34})}{\Gamma^{\mC}(2+s_{14}+s_{34}|s_{14}+s_{34})} {_2F}_1^{\mC}\left[\!\!\begin{array}{c}1-s_{24}|-s_{24},\ 1+s_{14}|s_{14}\\ 2+s_{14}+s_{34}|s_{14}+s_{34}\end{array}\!\!\Big|z\right]\right.\notag\\
	&\hspace{10ex}\left.-\frac{s_{14}}{s_{24}+s_{34}}\frac{1}{z_3}\frac{\Gamma^{\mC}(s_{14}|s_{14})\Gamma^{\mC}(1+s_{34}|s_{34})}{\Gamma^{\mC}(1+s_{14}+s_{34}|s_{14}+s_{34})} {_2F}_1^{\mC}\left[\!\!\begin{array}{c}-s_{24}|-s_{24},\ s_{14}|s_{14}\\ 1+s_{14}+s_{34}|s_{14}+s_{34}\end{array}\!\!\Big|z\right]\right)\!.
\end{align*}
}%
Now we can use \eqn{eqn:2F1idC}, then expand around $z_3=1$ using
{\small
\begin{equation*}
	\Fhyper{a}{b}{c}{1-x}=1+\cO(x-1),
\end{equation*}
}%
and finally use the complex Gamma function from \eqn{eqn:GammaC} in terms of the normal Gamma functions:
\begin{small}
\begin{align*}
    &\hcF_{4,5}(1)=\lim_{z_3\to1}|z_3|^{2s_{134}}\frac{s_{14}s_{24}}{s_{24}{+}s_{34}}\frac{\Gamma^{\mC}(s_{14}|s_{14})\Gamma^{\mC}(1{+}s_{34}|s_{34})}{\Gamma^{\mC}(1{+}s_{14}{+}s_{34}|s_{14}{+}s_{34})}\notag\\
	&\quad\times\left(\frac{\Gamma^{\mC}(1{+}s_{14}{+}s_{34}|s_{14}{+}s_{34})\Gamma^{\mC}(1{+}s_{24}{+}s_{34}|s_{24}{+}s_{34})}{\Gamma^{\mC}(1{+}s_{14}{+}s_{24}{+}s_{34}|s_{14}{+}s_{24}{+}s_{34})\Gamma^{\mC}(1{+}s_{34}|s_{34})}\Fhyper{{-}s_{24}}{s_{14}}{{-}s_{24}{-}s_{34}}{1\,{-}\,z_3}\Fhyper{{-}s_{24}}{s_{14}}{1{-}s_{24}{-}s_{34}}{1\,{-}\,\zb_3}\right.\notag\\
	&\quad\left.+\frac{\Gamma^{\mC}(1{+}s_{14}{+}s_{34}|s_{14}{+}s_{34})\Gamma^{\mC}({-}s_{24}{-}s_{34}{-}1|{-}s_{24}{-}s_{34})}{\Gamma^{\mC}({-}s_{24}|{-}s_{24})\Gamma^{\mC}(s_{14}|s_{14})}\right.\notag\\
	&\qquad\left.\times(1\,{-}\,z_3)^{s_{24}{+}s_{34}{+}1|s_{24}{+}s_{34}}\Fhyper{1{+}s_{14}{+}s_{24}{+}s_{34}}{1{+}s_{34}}{s_{24}{+}s_{34}{+}2}{1{-}z_3}\Fhyper{s_{14}{+}s_{24}{+}s_{34}}{s_{34}}{1{+}s_{24}{+}s_{34}}{1{-}\zb_3}\right)\notag\\
	&\quad+s_{24}|1\,{-}\,z_3|^{-2s_{24}-2s_{34}}|z_3|^{2(s_{134}+1)}\frac{1}{\zb_3}\notag\\
	&\hspace{10ex}\times\left(\frac{\Gamma^{\mC}(1{+}s_{14}|s_{14})\Gamma^{\mC}(1{+}s_{34}|s_{34})}{\Gamma^{\mC}(2{+}s_{14}{+}s_{34}|s_{14}{+}s_{34})}\right.\notag\\
	&\quad\left\{\frac{\Gamma^{\mC}(2{+}s_{14}{+}s_{34}|s_{14}{+}s_{34})\Gamma^{\mC}(s_{24}{+}s_{34}|s_{24}{+}s_{34})}{\Gamma^{\mC}(1{+}s_{14}{+}s_{24}{+}s_{34}|s_{14}{+}s_{24}{+}s_{34})\Gamma^{\mC}(1{+}s_{34}|s_{34})}\right.\notag\\
	&\qquad\Fhyper{1{-}s_{24}}{1{+}s_{14}}{1{-}s_{24}{-}s_{34}}{1{-}z_3}\Fhyper{{-}s_{24}}{s_{14}}{{-}s_{24}{-}s_{34}{-}1}{1{-}\zb_3}\notag\\
	&\qquad+\frac{\Gamma^{\mC}(2{+}s_{14}{+}s_{34}|s_{14}{+}s_{34})\Gamma^{\mC}({-}s_{24}{-}s_{34}|{-}s_{24}{-}s_{34})}{\Gamma^{\mC}(1{-}s_{24}|{-}s_{24})\Gamma^{\mC}(1{+}s_{14}|s_{14})}|1\,{-}\,z_3|^{2s_{34}+2s_{24}}\notag\\
	&\qquad\left.\Fhyper{1{+}s_{14}{+}s_{34}{+}s_{24}}{1{+}s_{34}}{1{+}s_{34}{+}s_{24}}{1{-}z}\Fhyper{s_{14}{+}s_{34}{+}s_{24}}{s_{34}}{1{+}s_{34}{+}s_{24}}{1{-}\zb}\right\}\notag\\
	&\hspace{20ex}\left.-\frac{s_{14}}{s_{24}{+}s_{34}}\frac{1}{z_3}\frac{\Gamma^{\mC}(s_{14}|s_{14})\Gamma^{\mC}(1{+}s_{34}|s_{34})}{\Gamma^{\mC}(1{+}s_{14}{+}s_{34}|s_{14}{+}s_{34})} \right.\notag\\
	&\quad\left\{\frac{\Gamma^{\mC}(1{+}s_{14}{+}s_{34}|s_{14}{+}s_{34})\Gamma^{\mC}(1{+}s_{24}{+}s_{34}|s_{24}{+}s_{34})}{\Gamma^{\mC}(1{+}s_{14}{+}s_{24}{+}s_{34}|s_{14}{+}s_{24}{+}s_{34})\Gamma^{\mC}(1{+}s_{34}|s_{34})}\Fhyper{{-}s_{24}}{s_{14}}{{-}s_{24}{-}s_{34}}{1{-}z_3}\Fhyper{{-}s_{24}}{s_{14}}{1{-}s_{24}{-}s_{34}}{1{-}\zb_3}\right.\notag\\
	&\quad\left.+\frac{\Gamma^{\mC}(1{+}s_{14}{+}s_{34}|s_{14}{+}s_{34})\Gamma^{\mC}({-}s_{24}{-}s_{34}{-}1|{-}s_{24}{-}s_{34})}{\Gamma^{\mC}({-}s_{24}|{-}s_{24})\Gamma^{\mC}(s_{14}|s_{14})}\right.\notag\\
	&\qquad\left.\left.\times(1\,{-}\,z_3)^{1+s_{24}+s_{34}|s_{24}+s_{34}}\Fhyper{1{+}s_{14}{+}s_{24}{+}s_{34}}{1{+}s_{34}}{2{+}s_{24}{+}s_{34}}{1{-}z_3}\Fhyper{s_{14}{+}s_{24}{+}s_{34}}{s_{34}}{1{+}s_{24}{+}s_{34}}{1{-}\zb_3}\right\}\right)\\
	&=\lim_{z_3\to1}\frac{s_{14}s_{24}}{s_{24}{+}s_{34}}\frac{\Gamma^{\mC}(s_{14}|s_{14})\Gamma^{\mC}(1{+}s_{24}{+}s_{34}|s_{24}{+}s_{34})}{\Gamma^{\mC}(1{+}s_{14}{+}s_{24}{+}s_{34}|s_{14}{+}s_{24}{+}s_{34})}+\cO(z_3\,{-}\,1)+\cO(\zb_3\,{-}\,1)\notag\\
	&\quad+s_{24}|1\,{-}\,z_3|^{-2s_{24}-2s_{34}}|z_3|^{2(s_{134}+1)}\frac{1}{\zb_3}\notag\\
	&\quad\times\left(\frac{\Gamma^{\mC}(1{+}s_{14}|s_{14})\Gamma^{\mC}(1{+}s_{34}|s_{34})}{\Gamma^{\mC}(2{+}s_{14}{+}s_{34}|s_{14}{+}s_{34})}\right.\notag\\
	&\quad\left\{\frac{\Gamma^{\mC}(2{+}s_{14}{+}s_{34}|s_{14}{+}s_{34})\Gamma^{\mC}(s_{24}{+}s_{34}|s_{24}{+}s_{34})}{\Gamma^{\mC}(1{+}s_{14}{+}s_{24}{+}s_{34}|s_{14}{+}s_{24}{+}s_{34})\Gamma^{\mC}(1{+}s_{34}|s_{34})}+\cO(z_3\,{-}\,1)+\cO(\zb_3\,{-}\,1)\right.\notag\\
	&\qquad\left.+\frac{\Gamma^{\mC}(2{+}s_{14}{+}s_{34}|s_{14}{+}s_{34})\Gamma^{\mC}({-}s_{24}{-}s_{34}|{-}s_{24}{-}s_{34})}{\Gamma^{\mC}(1{-}s_{24}|{-}s_{24})\Gamma^{\mC}(1{+}s_{14}|s_{14})}|1\,{-}\,z_3|^{2s_{34}+2s_{24}}(1{+}\cO(z_3\,{-}\,1)+\cO(\zb_3\,{-}\,1))\right\}\notag\\
	&\quad\left.-\frac{s_{14}}{s_{24}{+}s_{34}}\frac{\Gamma^{\mC}(s_{14}|s_{14})\Gamma^{\mC}(1{+}s_{34}|s_{34})}{\Gamma^{\mC}(1{+}s_{14}{+}s_{34}|s_{14}{+}s_{34})} \right.\notag\\
	&\quad\left\{\frac{\Gamma^{\mC}(1{+}s_{14}{+}s_{34}|s_{14}{+}s_{34})\Gamma^{\mC}(1{+}s_{24}{+}s_{34}|s_{24}{+}s_{34})}{\Gamma^{\mC}(1{+}s_{14}{+}s_{24}{+}s_{34}|s_{14}{+}s_{24}{+}s_{34})\Gamma^{\mC}(1{+}s_{34}|s_{34})}+\cO(z_3\,{-}\,1)+\cO(\zb_3\,{-}\,1)\right.\notag\\
	&\hspace{60ex}\left.\left.+\cO((1\,{-}\,z_3)^{s_{24}+s_{34}+1|s_{24}+s_{34}})\right\}\right)\\
	&=\lim_{z_3\to1}\frac{s_{14}s_{24}}{s_{24}{+}s_{34}}\frac{\Gamma^{\mC}(s_{14}|s_{14})\Gamma^{\mC}(1{+}s_{24}{+}s_{34}|s_{24}{+}s_{34})}{\Gamma^{\mC}(1{+}s_{14}{+}s_{24}{+}s_{34}|s_{14}{+}s_{24}{+}s_{34})}\notag\\
	&\quad+s_{24}\frac{\Gamma^{\mC}(1{+}s_{34}|s_{34})\Gamma^{\mC}({-}s_{24}{-}s_{34}|{-}s_{24}{-}s_{34})}{\Gamma^{\mC}(1{-}s_{24}|{-}s_{24})}\notag\\
	&\quad+s_{24}|1\,{-}\,z_3|^{-2s_{24}-2s_{34}}\notag\\
	&\hspace{5ex}\times\left(\frac{\Gamma^{\mC}(1{+}s_{14}|s_{14})\Gamma^{\mC}(s_{24}{+}s_{34}|s_{24}{+}s_{34})}{\Gamma^{\mC}(1{+}s_{14}{+}s_{24}{+}s_{34}|s_{14}{+}s_{24}{+}s_{34})}\right.\notag\\
	&\hspace{10ex}\left.-\frac{s_{14}}{s_{24}{+}s_{34}}\frac{\Gamma^{\mC}(s_{14}|s_{14})\Gamma^{\mC}(1{+}s_{24}{+}s_{34}|s_{24}{+}s_{34})}{\Gamma^{\mC}(1{+}s_{14}{+}s_{24}{+}s_{34}|s_{14}{+}s_{24}{+}s_{34})}+\cO(z_3\,{-}\,1)+\cO(\zb_3\,{-}\,1)\right)\notag\\
	&\quad+\cO(z_3\,{-}\,1)+\cO(\zb_3\,{-}\,1)\\
	&=\lim_{z_3\to1}s_{24}\frac{\Gamma(1{+}s_{14})\Gamma(s_{24}{+}s_{34})\Gamma(1{-}s_{14}{-}s_{24}{-}s_{34})}{\Gamma(1{-}s_{14})\Gamma(1{-}s_{24}{-}s_{34})\Gamma(1{+}s_{14}{+}s_{24}{+}s_{34})}\notag\\
	&\quad+s_{24}\frac{\Gamma(1{+}s_{34})\Gamma({-}s_{24}{-}s_{34})\Gamma(1{+}s_{24})}{\Gamma(1{-}s_{34})\Gamma(1{+}s_{24}{+}s_{34})\Gamma(1{-}s_{24})}\notag\\
	&\quad+s_{24}|1\,{-}\,z_3|^{-2s_{24}-2s_{34}}\notag\\
	&\hspace{5ex}\times\left(\frac{\Gamma(1{+}s_{14})\Gamma(s_{24}{+}s_{34})\Gamma(1{-}s_{14}{-}s_{24}{-}s_{34})}{\Gamma(1{-}s_{14})\Gamma(1{-}s_{24}{-}s_{34})\Gamma(1{+}s_{14}{+}s_{24}{+}s_{34})}\right.\notag\\
	&\hspace{10ex}\left.-\frac{s_{14}}{s_{24}{+}s_{34}}\frac{\Gamma(s_{14})\Gamma(1{+}s_{24}{+}s_{34})\Gamma(1{-}s_{14}{-}s_{24}{-}s_{34})}{\Gamma(1{-}s_{14})\Gamma(1{-}s_{24}{-}s_{34})\Gamma(1{+}s_{14}{+}s_{24}{+}s_{34})}+\cO(z_3\,{-}\,1)+\cO(\zb_3\,{-}\,1)\right)\notag\\
	&\quad+\cO(z_3\,{-}\,1)+\cO(\zb_3\,{-}\,1)\notag\\
	&=\lim_{z_3\to1}s_{24}\frac{\Gamma(1{+}s_{14})\Gamma(s_{24}{+}s_{34})\Gamma(1{-}s_{14}{-}s_{24}{-}s_{34})}{\Gamma(1{-}s_{14})\Gamma(1{-}s_{24}{-}s_{34})\Gamma(1{+}s_{14}{+}s_{24}{+}s_{34})}+s_{24}\frac{\Gamma(1{+}s_{34})\Gamma({-}s_{24}{-}s_{34})\Gamma(1{+}s_{24})}{\Gamma(1{-}s_{34})\Gamma(1{+}s_{24}{+}s_{34})\Gamma(1{-}s_{24})}\notag\\
	&\quad+\cO(z_3\,{-}\,1)+\cO(\zb_3\,{-}\,1)+\cO((z_3\,{-}\,1)|1-z_3|^{-2s_{24}-2s_{34}})+\cO((\zb_3-1)|1\,{-}\,z_3|^{-2s_{24}-2s_{34}}).
\end{align*}
\end{small}
We can see that similar to the open-string case, the divergences cancel each other to first order. Thus, using the condition $\mathrm{Re}(s_{24}+s_{34})<1/2$, all divergences cancel in the limit and we arrive at
{\small
\begin{align*}
	\hcF_{4,5}(1)=s_{24}\frac{\Gamma(1{+}s_{14})\Gamma(s_{24}{+}s_{34})\Gamma(1{-}s_{14}{-}s_{24}{-}s_{34})}{\Gamma(1{-}s_{14})\Gamma(1{-}s_{24}{-}s_{34})\Gamma(1{+}s_{14}{+}s_{24}{+}s_{34})}{+}s_{24}\frac{\Gamma(1{+}s_{34})\Gamma({-}s_{24}{-}s_{34})\Gamma(1{+}s_{24})}{\Gamma(1{-}s_{34})\Gamma(1{+}s_{24}{+}s_{34})\Gamma(1{-}s_{24})}.
\end{align*}
}%


\section{Details of \texorpdfstring{$L=5$}{L=5} calculations}

\subsection{Open strings}\label{app:5open}
Let $L\,{=}\,5$, then the basis from \cite{Broedel:2013aza} (see \eqn{eqn:basisBSST}), written in terms of open Selberg integrals is
{\small
\begin{equation*}
    \hF_5(x_3)=\left(\begin{array}{c}
		\hF_{5,4}^{\mathrm{id}}(x_3)\\
		\hF_{5,4}^{(45)}(x_3)\\
		\hF_{5,5}^{\mathrm{id}}(x_3)\\
		\hF_{5,5}^{(45)}(x_3)\\
		\hF_{5,6}^{\mathrm{id}}(x_3)\\
        \hF_{5,6}^{(45)}(x_3)
	\end{array}\right)
    =
    \left(\begin{array}{c}
		s_{14}s_{15}S[1,1](x_3)+s_{45}s_{15}S[5,1](x_3)\\
		s_{14}s_{15}S[1,1](x_3)+s_{14}s_{45}S[1,4](x_3)\\
		-s_{24}s_{15}S[2,1](x_3)\\
		-s_{14}s_{25}S[1,2](x_3)\\
		s_{24}s_{25}S[2,2](x_3)+s_{24}s_{45}S[2,4](x_3)\\
        s_{24}s_{25}S[2,2](x_3)+s_{45}s_{25}S[5,2](x_3)
	\end{array}\right).
\end{equation*}
}%
Calculating the derivate of $\hF_5(x_3)$ yields the KZ equation
{\small
\begin{equation*}
    \frac{d}{dx_3}\hF_5(x_3)=\left(\frac{e_0}{x_3}+\frac{e_1}{x_3-1}\right)\hF_5(x_3)\,,
\end{equation*}
}%
with the two matrices from \eqn{eqn:5e}.

Using the Propositions from \secref{sec:opengeneral}, we find the limit vectors
{\small
\begin{align*}
    \hF_5(0)=\lim_{x_3\to0}x_3^{-e_0}\hF_5(x_3)=\left(\begin{array}{c}
		\hF_{5,4}^{\mathrm{id}}(0)\\
		\hF_{5,4}^{(45)}(0)\\
		0\\0\\0\\0
	\end{array}\right),\quad
    \hF_5(1)=\lim_{x_3\to1}(1-x_3)^{-e_1}\hF_5(x_3)=\left(\begin{array}{c}
		\hF_{5,4}^{\mathrm{id}}(1)\\
		\hF_{5,4}^{(45)}(1)\\
		\hF_{5,5}^{\mathrm{id}}(1)\\
        \hF_{5,5}^{(45)}(1)\\
        \hF_{5,6}^{\mathrm{id}}(1)\\
        \hF_{5,6}^{(45)}(1)
	\end{array}\right),
\end{align*}
}%
with
{\small
\begin{align*}
    \hF_{5,4}^\mathrm{id}(0)&=s_{14}s_{15}\int_{0}^{1}dw_4\int_{0}^{w_4}dw_5\,w_4^{s_{14}-1}w_5^{s_{15}-1}(1\,{-}\,w_4)^{s_{34}}(1\,{-}\,w_5)^{s_{35}}(w_4\,{-}\,w_5)^{s_{45}}\notag\notag\\
    &\quad+s_{45}s_{15}\int_{0}^{1}dw_4\int_{0}^{w_4}dw_5\,w_4^{s_{14}}w_5^{s_{15}-1}(1\,{-}\,w_4)^{s_{34}}(1\,{-}\,w_5)^{s_{35}}(w_4\,{-}\,w_5)^{s_{45}-1}\notag\\
	&\overset{s_{3\bullet}\to0}{\longrightarrow}s_{15}\int_{0}^{1}dw_5\,w_5^{s_{15}-1}(1\,{-}\,w_5)^{s_{45}}=s_{14}s_{15}\frac{\Gamma(1+s_{15})\Gamma(1+s_{45})}{\Gamma(1+s_{15}+s_{45})}\\
    &=F_4,\\
    \hF_{5,4}^{(45)}(0)&=s_{14}s_{15}\int_{0}^{1}dw_4\int_{0}^{w_4}dw_5\,w_4^{s_{14}-1}w_5^{s_{15}-1}(1\,{-}\,w_4)^{s_{34}}(1\,{-}\,w_5)^{s_{35}}(w_4\,{-}\,w_5)^{s_{45}}\notag\\
    &\quad-s_{14}s_{45}\int_{0}^{1}dw_4\int_{0}^{w_4}dw_5\,w_4^{s_{14}-1}w_5^{s_{15}}(1\,{-}\,w_4)^{s_{34}}(1\,{-}\,w_5)^{s_{35}}(w_4\,{-}\,w_5)^{s_{45}-1}\notag\\
    &\overset{s_{3\bullet}\to0}{\longrightarrow}s_{14}s_{15}\int_{0}^{1}dw_4\int_{0}^{w_4}dw_5\,w_4^{s_{14}-1}w_5^{s_{15}-1}(w_4\,{-}\,w_5)^{s_{45}}\notag\\
    &\quad-s_{14}s_{45}\int_{0}^{1}dw_4\int_{0}^{w_4}dw_5\,w_4^{s_{14}-1}w_5^{s_{15}}(w_4\,{-}\,w_5)^{s_{45}-1}\notag\\
    &=\frac{s_{14} \Gamma (1+s_{15}) \Gamma(1+s_{45})}{(s_{14}+s_{15}+s_{45}) \Gamma (1+s_{15}+s_{45})}-\frac{s_{14} \Gamma (1+s_{15}) \Gamma(1+s_{45})}{(s_{14}+s_{15}+s_{45}) \Gamma (1+s_{15}+s_{45})}\notag\\
    &=0,\\
    \hF_{5,4}^\mathrm{id}(1)&=s_{14}s_{15}\int_0^1dw_4\int_{0}^{w_4}dw_5\,w_4^{s_{14}-1}w_5^{s_{15}-1}(1\,{-}\,w_4)^{s_{24}+s_{34}}(1\,{-}\,w_5)^{s_{25}+s_{35}}(w_4\,{-}\,w_5)^{s_{45}}\notag\\
    &\quad+s_{45}s_{15}\int_0^1dw_4\int_{0}^{w_4}dw_5\,w_4^{s_{14}}w_5^{s_{15}-1}(1\,{-}\,w_4)^{s_{24}+s_{34}}(1\,{-}\,w_5)^{s_{25}+s_{35}}(w_4\,{-}\,w_5)^{s_{45}-1}\notag\\
    &\overset{s_{3\bullet}\to0}{\longrightarrow}s_{14}s_{15}\int_0^1dw_4\int_{0}^{w_4}dw_5\,w_4^{s_{14}-1}w_5^{s_{15}-1}(1\,{-}\,w_4)^{s_{24}}(1\,{-}\,w_5)^{s_{25}}(w_4\,{-}\,w_5)^{s_{45}}\notag\\
    &\quad+s_{45}s_{15}\int_0^1dw_4\int_{0}^{w_4}dw_5\,w_4^{s_{14}}w_5^{s_{15}-1}(1\,{-}\,w_4)^{s_{24}}(1\,{-}\,w_5)^{s_{25}}(w_4\,{-}\,w_5)^{s_{45}-1}\notag\\
    &=\frac{\Gamma(1+s_{145})\Gamma(1+s_{15})\Gamma(1+s_{24})\Gamma(1+s_{45})}{\Gamma(1+s_{145}+s_{24})\Gamma(1+s_{15}+s_{45})}{_3F_2}\left[\!\!\begin{array}{c}1+s_{145},\ s_{15},\ -s_{25}\\1+s_{145}+s_{24},\ 1+s_{15}+s_{45}\end{array}\!\!\Big|1\right]\\
    &=F_5^\mathrm{id},\\
    \hF_{5,4}^{(45)}(1)&=s_{14}s_{15}\int_0^1dw_4\int_{0}^{w_4}dw_5\,w_4^{s_{14}-1}w_5^{s_{15}-1}(1\,{-}\,w_4)^{s_{24}+s_{34}}(1\,{-}\,w_5)^{s_{25}+s_{35}}(w_4\,{-}\,w_5)^{s_{45}}\notag\\
    &\quad-s_{14}s_{45}\int_0^1dw_4\int_{0}^{w_4}dw_5\,w_4^{s_{14}-1}w_5^{s_{15}}(1\,{-}\,w_4)^{s_{24}+s_{34}}(1\,{-}\,w_5)^{s_{25}+s_{35}}(w_4\,{-}\,w_5)^{s_{45}-1}\notag\\
    &\overset{s_{3\bullet}\to0}{\longrightarrow}s_{14}\frac{\Gamma(s_{145})\Gamma(1+s_{15})\Gamma(1+s_{24})\Gamma(1+s_{45})}{\Gamma(1+s_{145}+s_{24})\Gamma(1+s_{15}+s_{45})}{_3F_2}\left[\!\!\begin{array}{c}s_{145},\ s_{15},\ -s_{25}\\1+s_{145}+s_{24},\ 1+s_{15}+s_{45}\end{array}\!\!\Big|1\right]\notag\\
	&\quad-s_{14}\frac{\Gamma(s_{145})\Gamma(1+s_{15})\Gamma(1+s_{24})\Gamma(1+s_{45})}{\Gamma(1+s_{145}+s_{24})\Gamma(1+s_{15}+s_{45})} {_3F_2}\left[\!\!\begin{array}{c}s_{145},\ 1+s_{15},\ -s_{25}\\1+s_{145}+s_{24},\ 1+s_{15}+s_{45}\end{array}\!\!\Big|1\right]\notag\\
	&=s_{14}s_{25}\frac{\Gamma(1+s_{145})\Gamma(1+s_{15})\Gamma(1+s_{24})\Gamma(1+s_{45})}{\Gamma(2+s_{145}+s_{24})\Gamma(2+s_{15}+s_{45})}{_3F_2}\left[\!\!\begin{array}{c}1\,{+}\,s_{145},\ 1\,{+}\,s_{15},\ 1\,{-}\,s_{25}\\2\,{+}\,s_{145}\,{+}\,s_{24},\ 2\,{+}\,s_{15}\,{+}\,s_{45}\end{array}\!\!\Big|1\right]\!\\
    &=F_5^{(45)},
\end{align*}
}%
where the limit $s_{3\bullet}\to0$ refers again to the physical limit where the Mandelstam variables related to the auxiliary point $x_3$ vanish and the functions $F$ without the hat denote the disk string amplitudes defined in \eqn{eqn:F}. We furthermore used the integral formula for the hypergeometric~$_3F_2$ function\footnote{In the calculation for the limit vector $\hF(1)$, we furthermore used the $_3F_2$ identity
{\small
\begin{equation*}
	a_1\,\cdot {_3F_2}\left[\!\!\begin{array}{c}a_1+1,\ a_2,\ a_3\\b_1,\ b_2\end{array}\!\!\Big|z\right]-a_2\,\cdot {_3F_2}\left[\!\!\begin{array}{c}a_1,\ a_2+1,\ a_3\\b_1,\ b_2\end{array}\!\!\Big|z\right]+(a_2-a_1)\,\cdot {_3F_2}\left[\!\!\begin{array}{c}a_1,\ a_2,\ a_3\\b_1,\ b_2\end{array}\!\!\Big|z\right]=0,
\end{equation*}
}%
which follows from \cite[Eq.~(2.1)]{Bailey_1954}.} (substitution $t_2\mapsto t_2t_1$ in the second line)
{\small
\begin{align*}
	&\frac{\Gamma(a_1)\Gamma(a_2)\Gamma(b_1-a_1)\Gamma(b_2-a_2)}{\Gamma(b_1)\Gamma(b_2)} {_3F_2}\left[\!\!\begin{array}{c}a_1,\ a_2,\ a_3\\b_1,\ b_2\end{array}\!\!\Big|z\right]\\
	&\hspace{15ex}=\int_0^1dt_1\int_0^1dt_2\,t_1^{a_1-1}t_2^{a_2-1}(1-t_1)^{b_1-a_1-1}(1-t_2)^{b_2-a_2-1}(1-zt_1t_2)^{-a_3}\\
	&\hspace{15ex}=\int_0^1dt_1\int_0^{t_1}dt_2\,t_1^{a_1-b_2}t_2^{a_2-1}(1-t_1)^{b_1-a_1-1}(1-zt_2)^{-a_3}(t_1-t_2)^{b_2-a_2-1},
\end{align*}
}%
at $z=1$.

In the physical limit $s_{3\bullet}\to0$ we can also calculate the other components of the limit vector $\hF_5(1)$: in a way, they are just mathematical accessories of the underlying formalism and only the first components are physically relevant. However, an interesting structure arises when looking at these components. We will do so in the limit, where the auxiliary Mandelstam variables are already taken to zero. This requires exchanging the two limits $x_3\to1$ and $s_{3\bullet}\to0$, which is allowed since both individual limits exist and are finite. For the third component of the limit vector we then calculate
{\small
\begin{align*}
    \hF_{5,5}^{\mathrm{id}}&(1)\big|_{s_{3\bullet}\to0}
    =\lim_{x_3\to1}\left[\hF_{5,4}^\mathrm{id}(x_3)\big|_{s_{3\bullet}\to0}-(1\,{-}\,x_3)^{-s_{24}}\left(\hF_{5,4}^\mathrm{id}(x_3)\big|_{s_{3\bullet}\to0}-\hF_{5,5}^{\mathrm{id}}(x_3)\big|_{s_{3\bullet}\to0}\right)\right]\\
    &=\lim_{x_3\to1}\left[\int_0^{x_3}dx_4\int_0^{x_4}dx_5\,x_4^{s_{14}}x_5^{s_{15}}(1\,{-}\,x_4)^{s_{24}}(1\,{-}\,x_5)^{s_{25}}(x_4\,{-}\,x_5)^{s_{45}}\frac{s_{15}}{x_5}\left(\frac{s_{14}}{x_4}+\frac{s_{45}}{x_4\,{-}\,x_5}\right)\right.\\
    &\hspace{8ex}-(1-x_3)^{-s_{24}}\int_0^{x_3}dx_4\int_0^{x_4}dx_5\,x_4^{s_{14}}x_5^{s_{15}}(1\,{-}\,x_4)^{s_{24}}(1\,{-}\,x_5)^{s_{25}}(x_4\,{-}\,x_5)^{s_{45}}\\
    &\hspace{45ex}\left.\times\frac{s_{15}}{x_5}\left(\frac{s_{14}}{x_4}+\frac{s_{45}}{x_4\,{-}\,x_5}+\frac{s_{24}}{x_4\,{-}\,1}\right)\right]\\
    &=\int_0^{1}dx_4\int_0^{x_4}dx_5\,x_4^{s_{14}}x_5^{s_{15}}(1\,{-}\,x_4)^{s_{24}}(1\,{-}\,x_5)^{s_{25}}(x_4\,{-}\,x_5)^{s_{45}}\frac{s_{15}}{x_5}\left(\frac{s_{14}}{x_4}+\frac{s_{45}}{x_4\,{-}\,x_5}\right)\\
    &\quad-\lim_{x_3\to1}(1\,{-}\,x_3)^{-s_{24}}\int_0^{x_3}dx_4\frac{d}{dx_4}\left[\int_0^{x_4}dx_5\,x_4^{s_{14}}x_5^{s_{15}}(1\,{-}\,x_4)^{s_{24}}(1\,{-}\,x_5)^{s_{25}}(x_4\,{-}\,x_5)^{s_{45}}\frac{s_{15}}{x_5}\right]\\
    &=\int_0^{1}dx_4\int_0^{x_4}dx_5\,x_4^{s_{14}}x_5^{s_{15}}(1\,{-}\,x_4)^{s_{24}}(1\,{-}\,x_5)^{s_{25}}(x_4\,{-}\,x_5)^{s_{45}}\frac{s_{15}}{x_5}\left(\frac{s_{14}}{x_4}+\frac{s_{45}}{x_4\,{-}\,x_5}\right)\\
    &\quad-\lim_{x_3\to1}(1\,{-}\,x_3)^{-s_{24}}\int_0^{x_3}dx_5\,x_3^{s_{14}}x_5^{s_{15}}(1\,{-}\,x_3)^{s_{24}}(1\,{-}\,x_5)^{s_{25}}(x_3\,{-}\,x_5)^{s_{45}}\frac{s_{15}}{x_5}\\
    &=\int_0^{1}dx_4\int_0^{x_4}dx_5\,x_4^{s_{14}}x_5^{s_{15}}(1\,{-}\,x_4)^{s_{24}}(1\,{-}\,x_5)^{s_{25}}(x_4\,{-}\,x_5)^{s_{45}}\frac{s_{15}}{x_5}\left(\frac{s_{14}}{x_4}+\frac{s_{45}}{x_4\,{-}\,x_5}\right)\\
    &\quad-s_{15}\int_0^{1}dx_5\,x_5^{s_{15}-1}(1\,{-}\,x_5)^{s_{25}+s_{45}}\,.
\end{align*}
}%
This is precisely equal the first component of the limit vector $\hF_{5,4}^\mathrm{id}(1)\big|_{s_{3\bullet}\to0}$ (i.e.~a five-point open string integral) minus a four-point open-string amplitude (Veneziano amplitude/Euler-Beta function). Similar calculations show (for $\hF_{5,5}^{(45)}(1)\big|_{s_{3\bullet}\to0}$ it is really just different by one permutation of indices 4 and 5 in the integrand):
{\small
\begin{align*}
    \hF_{5,5}^{(45)}(1)\big|_{s_{3\bullet}\to0}
    &=\int_0^{1}dx_4\int_0^{x_4}dx_5\,x_4^{s_{14}}x_5^{s_{15}}(1-x_4)^{s_{24}}(1-x_5)^{s_{25}}(x_4-x_5)^{s_{45}}\frac{s_{14}}{x_4}\left(\frac{s_{15}}{x_5}+\frac{s_{45}}{x_5{-}x_4}\right)\\
    &=F_5^{(45)}\,,\\
    \hF_{5,6}^{\mathrm{id}}(1)\big|_{s_{3\bullet}\to0}&=\hF_{5,5}^{\mathrm{id}}(1)\big|_{s_{3\bullet}\to0}\,,\\
    \hF_{5,6}^{(45)}(1)\big|_{s_{3\bullet}\to0}&=\int_0^1dx_4\int_0^{x_4}dx_5\,x_4^{s_{14}}x_5^{s_{15}}(1-x_4)^{s_{24}}(1-x_5)^{s_{25}}(x_4-x_5)^{s_{45}}\frac{s_{14}}{x_4}\left(\frac{s_{15}}{x_5}+\frac{s_{45}}{x_5{-}x_4}\right)\\
    &\quad-s_{25}\frac{\Gamma(1+s_{15})\Gamma(s_{25}+s_{45})}{\Gamma(1+s_{15}+s_{25}+s_{45})}+\frac{\Gamma(1+s_{45})\Gamma(-s_{25}-s_{45})}{\Gamma(-s_{25})}\,,    
\end{align*}
}%
where for the derivation of the last equation the condition $\Re(s_{25}+s_{45})<1$ is needed (the same identities for hypergeometric functions as for the $L\,{=}\,4$ case were used).

\subsection*{Relating the limits}
\thmref{thmKZvectorspace} lets us relate the limit vectors $\hF_5(0)$, $\hF_5(1)$ using the Drinfeld associator $\Phi(e_0,e_1)$:
{\small
\begin{equation*}
    \hF_5(1)=\Phi(e_0,e_1)\hF_5(0).
\end{equation*}
}%
This equality also holds after taking the physical limit $s_{3\bullet}\to0$, yielding a recursion bringing us from the four-point amplitude to a basis of five-point amplitudes:
{\small
\begin{align*}
	\left[\Phi(e_0,e_1)\hF_5(0)\right]\Big|_{s_{3\bullet\to0}}&=\left(1+\zeta_2[e_0,e_1]+\zeta_3[e_0+e_1,[e_0,e_1]]+\ldots\right)\big|_{s_{3\bullet\to0}}\left(\begin{array}{c}\frac{\Gamma (1+s_{15}) \Gamma (1+s_{45})}{\Gamma (1+s_{15}+s_{45})}\\0\\0\\0\\0\\0\end{array}\right)\\
 &=\left(\begin{array}{c}\hF_{5,4}^\mathrm{id}(1)\big|_{s_{3\bullet\to0}}\\ \hF_{5,4}^{(45)}(1)\big|_{s_{3\bullet\to0}} \\ \hF_{5,5}^{\mathrm{id}}(1)\big|_{s_{3\bullet\to0}}\\\hF_{5,5}^{(45)}(1)\big|_{s_{3\bullet\to0}}\\\hF_{5,6}^{\mathrm{id}}(1)\big|_{s_{3\bullet\to0}}\\ \hF_{5,6}^{(45)}(1)\big|_{s_{3\bullet\to0}}\end{array}\right)=\hF_5(1)\big|_{s_{3\bullet\to0}}.
\end{align*}
}%

\subsection{Closed strings}\label{app:closed5}
For $L\,{=}\,5$, \eqn{eqn:basisclosed} yields the vector
{\small
\begin{equation*}
	\hcF_5(z_3)=\left(\begin{array}{c}
		\hcF_{5,4}^{\mathrm{id}}(z_3)\\
        \hcF_{5,4}^{(45)}(z_3)\\
        \hcF_{5,5}^{\mathrm{id}}(z_3)\\
        \hcF_{5,5}^{(45)}(z_3)\\
        \hcF_{5,6}^{\mathrm{id}}(z_3)\\
        \hcF_{5,6}^{(45)}(z_3)
	\end{array}\right)
    =\left(\begin{array}{c}
		s_{14}s_{15}\SC[1,1](z_3)+s_{45}s_{15}\SC[5,1](z_3)\\
		s_{14}s_{15}\SC[1,1](z_3)+s_{14}s_{45}\SC[1,4](z_3)\\
		-s_{24}s_{15}\SC[2,1](z_3)\\
		-s_{14}s_{25}\SC[1,2](z_3)\\
		s_{24}s_{25}\SC[2,2](z_3)+s_{24}s_{45}\SC[2,4](z_3)\\
		s_{24}s_{25}\SC[2,2](z_3)+s_{45}s_{25}\SC[5,2](z_3)
	\end{array}\right).
\end{equation*}
}%
As the partial derivative w.r.t.~$z_3$ only acts on the holomorphic part of the integrand and this part is the same as in the open-string case, the matrices for the KZ equation are the same as before, see \eqn{eqn:5e}.

Applying the results from \secref{sec:closedgeneral}, the limit vectors are
{\small
\begin{equation*}
    \hcF_5(0)=\lim_{z_3\to0}|z_3|^{-2e_0}\hcF_5(z_3)=\left(\!\!\!\begin{array}{c}\hcF_{5,4}^{\mathrm{id}}(0)\\ \hcF_{5,4}^{(45)}(0)\\0\\0\\0\\0\end{array}\!\!\!\right)\!,\quad \hcF_5(1)=\lim_{z_3\to1}|1-z_3|^{-2e_1}\hcF(z_3)=\left(\!\!\!\begin{array}{c}\hcF_{5,4}^{\mathrm{id}}(1)\\ \hcF_{5,4}^{(45)}(1)\\ \ast\\ \ast\\ \ast\\ \ast\end{array}\!\!\!\right)\!,
\end{equation*}
}%
with components\footnote{To solve the integrals in the physical limit in terms of Gamma functions, one can use again the formula \eqref{eqn:complexBeta}, which is for example proven in \rcite{Mimachi:2018}.}
\begin{small}
\begin{align*}
	\hcF_{5,4}^{\mathrm{id}}(0)&=\frac{s_{14}s_{15}}{(-2\pi i)^2}\int_{\PC}\frac{dw_4 d\wb_4}{w_4(\wb_4{-}1)}\int_{\PC}\frac{dw_5 d\wb_5}{|w_5|^2(\wb_5{-}\wb_4)}|w_4|^{2s_{14}}|w_5|^{2s_{15}}|1{-}w_4|^{2s_{34}}|1{-}w_5|^{2s_{35}}|w_4{-}w_5|^{2s_{45}}\\
	&\quad-\frac{s_{45}s_{15}}{(-2\pi i)^2}\int_{\PC}\frac{dw_4 d\wb_4}{\wb_4{-}1}\int_{\PC}\frac{dw_5 d\wb_5}{|w_5|^2|w_5{-}w_4|^2}|w_4|^{2s_{14}}|w_5|^{2s_{15}}|1{-}w_4|^{2s_{34}}|1{-}w_5|^{2s_{35}}|w_4{-}w_5|^{2s_{45}}\\
    &\overset{s_{3\bullet}\to0}{\longrightarrow}\frac{s_{15}}{-2\pi i}\int_{\PC}\frac{dw_5 d\wb_5}{|w_5|^2(\wb_5\,{-}\,1)}|w_5|^{2s_{15}}|1\,{-}\,w_5|^{2s_{45}}=\frac{\Gamma(1+s_{15})\Gamma(1+s_{45})\Gamma(1-s_{15}-s_{45})}{\Gamma(1-s_{15})\Gamma(1-s_{45})\Gamma(1+s_{15}+s_{45})},\\
	\hcF_{5,4}^{(45)}(0)&=\frac{s_{14}s_{15}}{(-2\pi i)^2}\int_{\PC}\frac{dw_4 d\wb_4}{w_4(\wb_4{-}1)}\int_{\PC}\frac{dw_5 d\wb_5}{|w_5|^2(\wb_5{-}\wb_4)}|w_4|^{2s_{14}}|w_5|^{2s_{15}}|1{-}w_4|^{2s_{34}}|1{-}w_5|^{2s_{35}}|w_4{-}w_5|^{2s_{45}}\\
	&\quad+\frac{s_{14}s_{45}}{(-2\pi i)^2}\int_{\PC}\frac{dw_4 d\wb_4}{w_4(\wb_4{-}1)}\int_{\PC}\frac{dw_5 d\wb_5}{\wb_5|w_5{-}w_4|^2}|w_4|^{2s_{14}}|w_5|^{2s_{15}}|1{-}w_4|^{2s_{34}}|1{-}w_5|^{2s_{35}}|w_4{-}w_5|^{2s_{45}}\\
    &\overset{s_{3\bullet}\to0}{\longrightarrow}0\,,\\
	\hcF_{5,4}^{\mathrm{id}}(1)&=\frac{s_{14}s_{15}}{(-2\pi i)^2}\int_{\PC}\frac{dw_4 d\wb_4}{w_4(\wb_4\,{-}\,1)}\int_{\PC}\frac{dw_5 d\wb_5}{|w_5|^2(\wb_5\,{-}\,\wb_{4})}\notag\\
    &\hspace{20ex}\times|w_4|^{2s_{14}}|w_5|^{2s_{15}}|1\,{-}\,w_4|^{2(s_{24}+s_{34})}|1\,{-}\,w_5|^{2(s_{25}+s_{35})}|w_4\,{-}\,w_5|^{2s_{45}}\notag\\
	&\quad-\frac{s_{15}s_{45}}{(-2\pi i)^2}\int_{\PC}\frac{dw_4 d\wb_4}{\wb_4\,{-}\,1}\int_{\PC}\frac{dw_5 d\wb_5}{|w_5|^2|w_5\,{-}\,w_4|^2}\notag\\
    &\hspace{20ex}\times|w_4|^{2s_{14}}|w_5|^{2s_{15}}|1\,{-}\,w_4|^{2(s_{24}+s_{34})}|1\,{-}\,w_5|^{2(s_{25}+s_{35})}|w_4\,{-}\,w_5|^{2s_{45}}\\
    &\overset{s_{3\bullet}\to0}{\longrightarrow}\frac{s_{14}s_{15}}{(-2\pi i)^2}\int_{\PC}\frac{dw_4 d\wb_4}{w_4(\wb_4\,{-}\,1)}\int_{\PC}\frac{dw_5 d\wb_5}{|w_5|^2(\wb_5\,{-}\,\wb_{4})}\notag\\
    &\hspace{20ex}\times|w_4|^{2s_{14}}|w_5|^{2s_{15}}|1\,{-}\,w_4|^{2s_{24}}|1\,{-}\,w_5|^{2s_{25}}|w_4\,{-}\,w_5|^{2s_{45}}\notag\\
	&\quad-\frac{s_{15}s_{45}}{(-2\pi i)^2}\int_{\PC}\frac{dw_4 d\wb_4}{\wb_4\,{-}\,1}\int_{\PC}\frac{dw_5 d\wb_5}{|w_5|^2|w_5\,{-}\,w_4|^2}\notag\\
    &\hspace{20ex}\times|w_4|^{2s_{14}}|w_5|^{2s_{15}}|1\,{-}\,w_4|^{2s_{24}}|1\,{-}\,w_5|^{2s_{25}}|w_4\,{-}\,w_5|^{2s_{45}},\\
	\hcF_{5,4}^{(45)}(1)&=\frac{s_{14}s_{15}}{(-2\pi i)^2}\int_{\PC}\frac{dw_4 d\wb_4}{w_4(\wb_4\,{-}\,1)}\int_{\PC}\frac{dw_5 d\wb_5}{|w_5|^2(\wb_5\,{-}\,\wb_{4})}\notag\\
    &\hspace{20ex}\times|w_4|^{2s_{14}}|w_5|^{2s_{15}}|1\,{-}\,w_4|^{2(s_{24}+s_{34})}|1\,{-}\,w_5|^{2(s_{25}+s_{35})}|w_4\,{-}\,w_5|^{2s_{45}}\notag\\
	&\quad+\frac{s_{14}s_{45}}{(-2\pi i)^2}\int_{\PC}\frac{dw_4 d\wb_4}{w_4(\wb_4\,{-}\,1)}\int_{\PC}\frac{dw_5 d\wb_5}{\wb_5|w_5\,{-}\,w_4|^2}\notag\\
    &\hspace{20ex}\times|w_4|^{2s_{14}}|w_5|^{2s_{15}}|1\,{-}\,w_4|^{2(s_{24}+s_{34})}|1\,{-}\,w_5|^{2(s_{25}+s_{35})}|w_4\,{-}\,w_5|^{2s_{45}}\\
    &\overset{s_{3\bullet}\to0}{\longrightarrow}\frac{s_{14}s_{15}}{(-2\pi i)^2}\int_{\PC}\frac{dw_4 d\wb_4}{w_4(\wb_4\,{-}\,1)}\int_{\PC}\frac{dw_5 d\wb_5}{|w_5|^2(\wb_5\,{-}\,\wb_{4})}\notag\\
    &\hspace{20ex}\times|w_4|^{2s_{14}}|w_5|^{2s_{15}}|1\,{-}\,w_4|^{2s_{24}}|1\,{-}\,w_5|^{2s_{25}}|w_4\,{-}\,w_5|^{2s_{45}}\notag\\
	&\quad+\frac{s_{14}s_{45}}{(-2\pi i)^2}\int_{\PC}\frac{dw_4 d\wb_4}{w_4(\wb_4\,{-}\,1)}\int_{\PC}\frac{dw_5 d\wb_5}{\wb_5|w_5\,{-}\,w_4|^2}\notag\\
	&\hspace{20ex}\times|w_4|^{2s_{14}}|w_5|^{2s_{15}}|1\,{-}\,w_4|^{2s_{24}}|1\,{-}\,w_5|^{2s_{25}}|w_4\,{-}\,w_5|^{2s_{45}}.
\end{align*}
\end{small}

As in the open-string four- to five-point example in \secref{sec:open5}, we can calculate also the lower the vector $\hcF_5(1)$ in the physical limit by exchanging the limits $z_3\to0$ and $s_{3\bullet}\to0$, which again is allowed since both individual limits exist and are finite. This goes along the same lines as in~\secref{sec:open5}, where we just need to be careful with residues when encountering total derivatives in the integrand (along the same lines as shown in the proof of \thmref{prop:equationclosed}). We find
{\small
\begin{align*}
    \hcF_{5,5}^{\mathrm{id}}&(1)\big|_{s_{3\bullet}\to0}
    =\lim_{z_3\to1}\left[\hcF_{5,4}^{\mathrm{id}}(z_3)\big|_{s_{3\bullet}\to0}-|1\,{-}\,z_3|^{-2s_{24}}\left(\hcF_{5,4}^{\mathrm{id}}(z_3)\big|_{s_{3\bullet}\to0}-\hcF_{5,5}^{\mathrm{id}}(z_3)\big|_{s_{3\bullet}\to0}\right)\right]\\
    &=\lim_{z_3\to1}\Bigg[\frac{1}{(-2\pi i)^2}\int_{(\PC)^2}\frac{\zb_3\,dz_4 d\zb_4\,dz_5 d\zb_5}{\zb_5(\zb_4\,{-}\,\zb_3)(\zb_5\,{-}\,\zb_4)}\\
    &\hspace{20ex}\times|z_4|^{2s_{14}}|z_5|^{2s_{15}}|1\,{-}\,z_4|^{2s_{24}}|1\,{-}\,z_5|^{2s_{25}}|z_4\,{-}\,z_5|^{2s_{45}}\frac{s_{15}}{z_5}\left(\frac{s_{14}}{z_4}+\frac{s_{45}}{z_4\,{-}\,z_5}\right)\\
    &\hspace{8ex}-\frac{|1\,{-}\,z_3|^{-2s_{24}}}{(-2\pi i)^2}\int_{(\PC)^2}\frac{\zb_3\,dz_4 d\zb_4\,dz_5 d\zb_5}{\zb_5(\zb_4\,{-}\,\zb_3)(\zb_5\,{-}\,\zb_4)}\\
    &\hspace{10ex}\times|z_4|^{2s_{14}}|z_5|^{2s_{15}}|1\,{-}\,z_4|^{2s_{24}}|1\,{-}\,z_5|^{2s_{25}}|z_4\,{-}\,z_5|^{2s_{45}}\frac{s_{15}}{z_5}\left(\frac{s_{14}}{z_4}+\frac{s_{45}}{z_4\,{-}\,z_5}+\frac{s_{24}}{z_4\,{-}\,1}\right)\Bigg]\\
    &=\frac{1}{(-2\pi i)^2}\int_{(\PC)^2}\frac{dz_4 d\zb_4\,dz_5 d\zb_5}{\zb_5(\zb_4\,{-}\,1)(\zb_5\,{-}\,\zb_4)}\\
    &\hspace{20ex}\times |z_4|^{2s_{14}}|z_5|^{2s_{15}}|1\,{-}\,z_4|^{2s_{24}}|1\,{-}\,z_5|^{2s_{25}}|z_4\,{-}\,z_5|^{2s_{45}}\frac{s_{15}}{z_5}\left(\frac{s_{14}}{z_4}+\frac{s_{45}}{z_4\,{-}\,z_5}\right)\\
    &\quad-\lim_{z_3\to1}\frac{|1\,{-}\,z_3|^{-2s_{24}}}{(-2\pi i)^2}\int_{\PC}d_{z_4}\Bigg[\int_{\PC}\frac{\zb_3\,d\zb_4\,dz_5 d\zb_5}{\zb_5(\zb_4\,{-}\,\zb_3)(\zb_5\,{-}\,\zb_4)}\\
    &\hspace{37ex}\times|z_4|^{2s_{14}}|z_5|^{2s_{15}}|1\,{-}\,z_4|^{2s_{24}}|1\,{-}\,z_5|^{2s_{25}}|z_4\,{-}\,z_5|^{2s_{45}}\frac{s_{15}}{z_5}\Bigg]\\
    &=\frac{1}{(-2\pi i)^2}\int_{(\PC)^2}\frac{dz_4 d\zb_4\,dz_5 d\zb_5}{\zb_5(\zb_4\,{-}\,1)(\zb_5\,{-}\,\zb_4)} \\
    &\hspace{20ex}\times|z_4|^{2s_{14}}|z_5|^{2s_{15}}|1\,{-}\,z_4|^{2s_{24}}|1\,{-}\,z_5|^{2s_{25}}|z_4\,{-}\,z_5|^{2s_{45}}\frac{s_{15}}{z_5}\left(\frac{s_{14}}{z_4}+\frac{s_{45}}{z_4\,{-}\,z_5}\right)\\
    &\quad-\lim_{z_3\to1}\frac{|1\,{-}\,z_3|^{-2s_{24}}}{(-2\pi i)^2}\int_{\PC}\frac{\zb_3\,dz_5 d\zb_5}{\zb_5(\zb_5\,{-}\,\zb_3)}|z_3|^{2s_{14}}|z_5|^{2s_{15}}|1\,{-}\,z_3|^{2s_{24}}|1\,{-}\,z_5|^{2s_{25}}|z_3\,{-}\,z_5|^{2s_{45}}\frac{s_{15}}{z_5}\\
    &=\frac{1}{(-2\pi i)^2}\int_{(\PC)^2}\frac{dz_4 d\zb_4\,dz_5 d\zb_5}{\zb_5(\zb_4\,{-}\,1)(\zb_5\,{-}\,\zb_4)}\\
    &\hspace{20ex}\times |z_4|^{2s_{14}}|z_5|^{2s_{15}}|1\,{-}\,z_4|^{2s_{24}}|1\,{-}\,z_5|^{2s_{25}}|z_4\,{-}\,z_5|^{2s_{45}}\frac{s_{15}}{z_5}\left(\frac{s_{14}}{z_4}+\frac{s_{45}}{z_4\,{-}\,z_5}\right)\\
    &\quad-\frac{s_{15}}{-2\pi i}\int_{\PC}\frac{dz_5 d\zb_5}{|z_5|^2(\zb_5\,{-}\,1)}|z_5|^{2s_{15}}|1\,{-}\,z_5|^{2(s_{25}+s_{45})}\\
    &=\frac{1}{(-2\pi i)^2}\int_{(\PC)^2}\frac{dz_4 d\zb_4\,dz_5 d\zb_5}{\zb_5(\zb_4\,{-}\,1)(\zb_5\,{-}\,\zb_4)}\\
    &\hspace{20ex}\times |z_4|^{2s_{14}}|z_5|^{2s_{15}}|1\,{-}\,z_4|^{2s_{24}}|1\,{-}\,z_5|^{2s_{25}}|z_4\,{-}\,z_5|^{2s_{45}}\frac{s_{15}}{z_5}\left(\frac{s_{14}}{z_4}+\frac{s_{45}}{z_4\,{-}\,z_5}\right)\\
    &\quad-\frac{\Gamma(1+s_{15})\Gamma(s_{25}+s_{45})\Gamma(1-s_{15}-s_{25}-s_{45})}{\Gamma(1-s_{15})\Gamma(1-s_{25}-s_{45})\Gamma(1+s_{15}+s_{25}+s_{45})}.
\end{align*}
}%
In the second term of this calculation we had to consider the residues in $\zb_4$ at $0,1,\zb_5,\zb_3$, where the first three vanish but the last one contributes the calculated term. Analogous to the open-string case, we find this result to equal $\hcF_{5,4}^{\mathrm{id}}(1)\big|_{s_{3\bullet}\to0}$ (i.e.~a five-point closed-string integral) minus a four-point closed-string amplitude (Virasoro--Shapiro amplitude). In the same way we find the other components to read
{\small
\begin{align*}
    \hcF_{5,5}^{(45)}(1)\big|_{s_{3\bullet}\to0}&=\frac{1}{(-2\pi i)^2}\int_{(\PC)^2}\frac{dz_4 d\zb_4\,dz_5 d\zb_5}{\zb_5(\zb_4-1)(\zb_5-\zb_4)}\\
    &\hspace{12ex}\times|z_4|^{2s_{14}}|z_5|^{2s_{15}}|1\,{-}\,z_4|^{2s_{24}}|1\,{-}\,z_5|^{2s_{25}}|z_4\,{-}\,z_5|^{2s_{45}}\frac{s_{14}}{z_4}\left(\frac{s_{15}}{z_5}{+}\frac{s_{45}}{z_5{-}z_4}\right)\\
    &=\hcF_{5,4}^{(45)}(1)\big|_{s_{3\bullet}\to0}\,,\\
    \hcF_{5,6}^{\mathrm{id}}(1)\big|_{s_{3\bullet}\to0}&=\hcF_{5,5}^{\mathrm{id}}(1)\big|_{s_{3\bullet}\to0}\,,\\
    \hcF_{5,6}^{(45)}(1)\big|_{s_{3\bullet}\to0}&=\frac{1}{(-2\pi i)^2}\int_{(\PC)^2}\frac{dz_4 d\zb_4\,dz_5 d\zb_5}{\zb_5(\zb_4\,{-}\,1)(\zb_5\,{-}\,\zb_4)}\\
    &\hspace{12ex}\times|z_4|^{2s_{14}}|z_5|^{2s_{15}}|1\,{-}\,z_4|^{2s_{24}}|1\,{-}\,z_5|^{2s_{25}}|z_4\,{-}\,z_5|^{2s_{45}}\frac{s_{14}}{z_4}\left(\frac{s_{15}}{z_5}{+}\frac{s_{45}}{z_5{-}z_4}\right)\\
    &\quad-s_{25}\frac{\Gamma(1+s_{15})\Gamma(s_{25}+s_{45})\Gamma(1-s_{15}-s_{25}-s_{45})}{\Gamma(1-s_{15})\Gamma(1-s_{25}-s_{45})\Gamma(1+s_{15}+s_{25}+s_{45})}\notag\\
    &\quad+\frac{\Gamma(1+s_{45})\Gamma(-s_{25}-s_{45})\Gamma(1+s_{25})}{\Gamma(1-s_{45})\Gamma(1+s_{25}+s_{45})\Gamma(-s_{25})}\,,    
\end{align*}
}%
where for the derivation of the last equation the condition $\Re(s_{25}+s_{45})<\frac12$ is needed (the same identities for hypergeometric functions as for the $L\,{=}\,4$ case were used, see \appref{app:4pt2ndcompclosed}).

\begin{rmk}
    All of the above results in this section are the single-valued images of the results from \secref{sec:open5} through applying the single-valued map from ref.~\cite{Brown:2019wna}.
\end{rmk}

\subsection*{Relating the limits}
Again, using \thmref{thmKZvectorspaceSV}, the vectors $\hcF_5(0)$ and $\hcF_5(1)$ are related through application of the Deligne associator $\Phi^{\mathrm{sv}}$, namely
{\small
\begin{equation*}
    \Phi^\mathrm{sv}(e_0,e_1)\hcF_5(0)=\hcF_5(1),
\end{equation*}
}%
and taking the physical limit of this equation, we obtain the relation translating between the four- and five-point closed-string amplitudes:
{\small
\begin{align*}
	\left[\Phi^{\mathrm{sv}}(e_0,e_1)\hcF_5(0)\right]\Big|_{s_{3\bullet}\to0}&=\left(1+2\zeta_3[e_0+e_1,[e_0,e_1]]+\ldots\right)\left(\begin{array}{c}\frac{\Gamma(1+s_{15})\Gamma(1+s_{45})\Gamma(1-s_{15}-s_{45})}{\Gamma(1-s_{15})\Gamma(1-s_{45})\Gamma(1+s_{15}+s_{45})}\\0\\0\\0\\0\\0\end{array}\right)\notag\\
    &=\left(\begin{array}{c}\hcF_{5,4}^{\mathrm{id}}(1)|_{s_{3\bullet}\to0}\\ \hcF_{5,4}^{(45)}(1)|_{s_{3\bullet}\to0}\\ \hcF_{5,5}^{\mathrm{id}}(1)|_{s_{3\bullet}\to0}\\ \hcF_{5,5}^{(45)}(1)|_{s_{3\bullet}\to0}\\ \hcF_{5,6}^{\mathrm{id}}(1)|_{s_{3\bullet}\to0}\\ \hcF_{5,6}^{(45)}(1)|_{s_{3\bullet}\to0}\end{array}\right)=\hcF_5(1)|_{s_{3\bullet}\to0}.
\end{align*}
}%

\addtocontents{toc}{\protect\setcounter{tocdepth}{1}}
\bibliographystyle{alphaurl}
\bibliography{DeligneRecursion}

\end{document}